\documentclass{amsart}
\usepackage{latexsym}
\usepackage{linearb}
\usepackage[LGR, OT1]{fontenc}
\usepackage{amsmath,amsthm}
\usepackage{amssymb}
\usepackage{upgreek}
\usepackage[greek,english]{babel}
\usepackage{teubner}
\usepackage{MnSymbol, wasysym}
\usepackage{graphicx}
\usepackage{marvosym}
\usepackage{eufrak}
\usepackage{color}
\def\f12{\frac 1 2}

\def\a{\alpha}

\def\de{\delta}

\def\pa{\partial}

\def\f12{\frac 1 2}

\newcommand{\nabb}{\mbox{$\nabla \mkern-13mu /$\,}}

\newcommand{\slashg}{\mbox{$g \mkern-9mu /$\,}}
\newcommand\lessflat{{{\mbox{$\flat  \mkern-12mu {}^{\_}$}}}}

\makeindex

\newtheorem{definition}{Definition}[section]

\newtheorem{theorem}{Theorem}[section]
\newtheorem{proposition}{Proposition}[subsection]
\newtheorem{bigprop}{Proposition}[section]

\newtheorem{corollary}{Corollary}[section]

\title[Decay for solutions of the wave equation on Kerr \emph{I-II}]{Decay for solutions of the wave equation\\
on Kerr exterior spacetimes \emph{I-II}:\\ The
cases $|a|\ll M$ or axisymmetry}
\author{Mihalis Dafermos}
\address{University of Cambridge,
Department of Pure Mathematics and Mathematical Statistics,
Wilberforce Road, Cambridge CB3 0WB United Kingdom}
\author{Igor Rodnianski}
\address{Princeton University,
Department of Mathematics, Fine Hall, Washington Road,
Princeton, NJ 08544 United States} 

\date\today
\begin{document}
\maketitle
\begin{abstract}
This paper contains the first two parts (I-II) of a three-part series 
concerning
the scalar wave equation 
$\Box_{g}\psi=0$ on a fixed Kerr background
$(\mathcal{M},g_{a,M})$.  We here restrict to  two cases:
({II}${}_1$) $|a|\ll M$, general $\psi$ or ({II}${}_2$) $|a|<M$,
$\psi$ axisymmetric. In either case, we  
prove a version of `integrated local energy decay', specifically,
that the $4$-integral
of an energy-type density (degenerating in a neighborhood of the Schwarzschild photon sphere
and at infinity), integrated over the domain of
dependence of a spacelike hypersurface $\Sigma$ connecting the future
event horizon with spacelike infinity or a sphere on null infinity, is bounded
by a natural (non-degenerate) energy flux of $\psi$ through $\Sigma$.
(The case (II${}_1$) has in fact been treated previously in our Clay Lecture notes:
\emph{Lectures on black holes and linear waves}, arXiv:0811.0354.)
In our forthcoming Part III, the
restriction to axisymmetry for the general $|a|<M$ case is removed.
The complete proof  is surveyed in our companion paper
\emph{The black hole stability problem for linear scalar perturbations}, which
includes the essential details of our forthcoming Part III.
Together with previous work (see our: \emph{A new physical-space approach to decay for the wave 
equation with applications to black hole spacetimes}, in XVIth International Congress
on Mathematical Physics, Pavel Exner ed., Prague 2009
pp.~421--433, 2009, http://arxiv.org/abs/0910.4957), 
this result leads, under suitable assumptions on
initial data of $\psi$, to polynomial decay bounds for the energy flux of $\psi$ 
through the foliation of the black hole exterior
defined by the time translates of a spacelike hypersurface
$\Sigma$ terminating on null infinity, 
as well as to pointwise decay
estimates, of a definitive form useful for nonlinear applications.
\end{abstract}

\tableofcontents
\section{Introduction}
The \emph{Kerr metrics} $g_{M,a}$ constitute a remarkable
two-parameter family  of
explicit solutions to the Einstein vacuum equations
\begin{equation}
\label{Eeq}
R_{\mu\nu}=0
\end{equation}
and describe spacetimes corresponding to rotating stationary black holes
with mass $M$ and angular momentum $aM$, as measured at infinity.

The family was discovered in local coordinates in 1963 by Kerr~\cite{kerr}, 
and within the subsequent decade, its
salient local and global
geometric features were definitively understood, especially 
through work of Carter~\cite{carter}.  It is widely expected that the exterior region of
the Kerr
family is dynamically stable in the context of the Cauchy problem (see~\cite{chge:givp})
for $(\ref{Eeq})$, 
in fact, that for a wide range of initial data for $(\ref{Eeq})$, not necessarily
close to data arising from $g_{M,a}$, the solution metric $g$ 
in an appropriate region asymptotically settles down to a member of the family.
See the discussion in~\cite{challenge}.
Indeed, it is this expectation which lies at the basis for the centrality of the Kerr
metric in our current astrophysical world-view. 
At present, however,
a mathematical resolution of
even the stability question (i.e.~the dynamics for spacetimes $g$ initially very near $g_{M,a}$)
remains a great open problem in classical general 
relativity. Many of the important difficulties of this problem 
are already manifest at the linear level, for the conjectured stability mechanism
would rest fundamentally on the dispersive properties of waves (an essentially
linear phenomenon) in the black hole
exterior region. One must certainly address,
thus, the following central mathematical question:
\begin{enumerate}
\item[]
\emph{Do waves indeed disperse on Kerr backgrounds, and how does one properly quantify 
this?}
\end{enumerate}

\subsection{Overview}
Motivated by the above, we consider the scalar homogeneous
wave equation
\begin{equation}
\label{WAVE}
\Box_g\psi =0
\end{equation}
on\index{$\psi$-related! $\psi$ (a solution of $\Box_g\psi =0$ for $g=g_{M,a}$)}
 exactly Kerr spacetimes $(\mathcal{M},g_{M,a})$. The study of $(\ref{WAVE})$
can be thought\index{operators! $\Box_g$ (wave operator associated to Lorentzian metric
$g$)}
of as a poor-man's linear theory associated to $(\ref{Eeq})$, neglecting, in particular,
its tensorial structure. For results pertinent to the latter, see Holzegel~\cite{kostakis2}.
In the present paper, we shall restrict to either of the following cases:
\begin{itemize}
\item
The parameters of the Kerr metric satisfy $|a|\ll M$.
\item
The parameters of the Kerr metric are allowed to lie
in the entire subextremal range $|a|<M$, but
$\psi$ is assumed axisymmetric.
\end{itemize}
The significance of the restriction $|a|\ll M$ is that the effect of 
superradiance (to be discussed below) is weak, and this
weakness can be exploited as a small parameter. In the axisymmetric case, 
superradiance is in fact entirely absent. 
A forthcoming paper~\cite{drf1} will consider non-axisymmetric solutions for
the general subextremal case $|a|<M$, thus completing the study of the problem
at hand. We provide in our companion paper~\cite{stabi} the
essential details of the proof in~\cite{drf1}.
\vskip1pc
{\bf
(Note: 
Part I of the present paper (consisting of
Sections 1--7) already provides 
certain preliminary results relevant to the entire series. With the completion of
Part III, Part I 
will be further extended so as to serve as an introduction to the entire series
and will be broken off from part II (Sections 8--10), which will retain the
proofs in the cases ({II}${}_1$), ({II}${}_2$) above.)}

\subsubsection{Boundedness\index{uniform boundedness}}
\label{bouint}
The most primitive global question concerning
 solutions to $(\ref{WAVE})$ which one must address
is that of boundedness. The essential difficulty of this problem lies
in the phenomenon of `superradiance'. \index{superradiance}
Briefly put, superradiance means that, 
since the stationary Killing field $T$ becomes
spacelike  at some points of the exterior region for all Kerr spacetimes
with $|a|\ne 0$,
the energy radiated to infinity can be greater than the initial energy on a Cauchy hypersurface.
In fact, it is not a priori obvious that the energy radiated to infinity (and hence the solution
pointwise, etc.)~can be controlled at all.

This long-standing open problem of energy and pointwise boundedness has
been resolved previously in  our~\cite{dr6}
for solutions to $(\ref{WAVE})$ on a general class of axisymmetric, stationary
black hole exterior spacetimes sufficiently close to the Schwarzschild metric,
a class which includes as a special case the Kerr family $g_{M,a}$ for
$|a|\ll M$, as well as the Kerr-Newman family $g_{M,a,Q}$ 
for $|a|\ll M$, $|Q|\ll M$. 

Recall that the Schwarzschild family, discovered already in 1916,
is precisely the subfamily of Kerr correponding to $a=0$. For the Schwarzschild
case, superradiance\index{superradiance} is absent as the stationary Killing field $T$ is in fact
non-spacelike everywhere
in the exterior. Thus, the boundedness statement is much easier, though still non-trivial if
one desires uniform control up to the horizon where $T$ becomes null; for this,
results go back to Wald~\cite{drimos} and Kay--Wald~\cite{kw:lss}. 
See  the discussion in~\cite{jnotes}   for many
further references.

The fact that under the assumptions of~\cite{dr5},
superradiance\index{superradiance} is weak and can be treated as a small parameter
is fundamental for the boundedness proof. It is in fact more natural to discuss briefly the
strategy of this proof in the following section, 
after discussing some of the ideas  involved in proving decay, in particular because
the main theorem of the present paper will retrieve the boundedness statement
in the special case of the Kerr family with $|a|\ll M$.

We note finally that
boundedness for axisymmetric solutions of $(\ref{WAVE})$ on Kerr exteriors in the entire subextremal range
$|a|<M$ also follows from~\cite{dr6}. This latter case shares the feature with Schwarzschild that
superradiance is absent, as the energy defined by $T$ is nonnegative definite
in the exterior when restricted to axisymmetric $\psi$.

\subsubsection{Integrated local energy decay}
The present paper thus concerns the problem of decay
for $\psi$ precisely in the cases for which boundedness was previously
known. Our main theorems (Theorem~\ref{nmt} and~\ref{nmt2} below)  state that,
for the two cases outlined above:
\begin{enumerate}
\item[]
\emph{The \emph{spacetime} integral of an energy-type
density integrated over the domain of development of an appropriate 
surface $\Sigma$
is bounded by an initial energy-type flux through $\Sigma$ (see $(\ref{protasn1b})$). }
\end{enumerate} 
The above spacetime energy density is non-degenerate at the horizon, but
degenerates  at infinity and in a region to which ``trapped null geodesics'' (see
Section~\ref{TNG}) approach. Both these degenerations are necessary.
The boundedness of a spacetime integral degenerating only at infinity
can then be obtained (see $(\ref{protasn2})$)
in terms of a higher order initial energy type quantity.
Estimates $(\ref{protasn1b})$ and $(\ref{protasn2})$ should be thought of as
statements of dispersion which often go by the name
``integrated local energy decay''. 

Let us briefly put these results in perspective.
In Minkowski space, the analogue of estimate $(\ref{protasn1b})$, with degeneration
only at infinity, was first shown in seminal work of Morawetz~\cite{mora2}, via use of an
identity associated to a virial-type current.  
In that case, the analogous estimate applied to $\Sigma=\{t=0\}$ would take the form
\begin{equation}
\label{analogueM}
\int_{\mathbb R^4} (1+r)^{-1-\delta}|\partial_t\psi|^2+(1+r)^{-1-\delta}|\partial_r\psi|^2+
r^{-1}|\nabb\psi|^2 \le C_\delta
\int_{t=0} |\partial_t\psi|^2+|\partial_r\psi|^2+|\nabb\psi|^2.
\end{equation}
for any $\delta>0$, where $r$ is a standard spherical coordinate and $\nabb$ denotes
the induced gradient on the constant $(r,t)$ spheres.
Such estimates were extended to solutions of the wave equation outside
of non-trapping obstacles in~\cite{mora3}.
On the other hand, it was shown by Ralston~\cite{Ralston} 
that in the presence of trapping, the integrand of the left-hand side of such an estimate must degenerate near trapped rays.
Analogues of $(\ref{protasn1b})$ in the Schwarzschild case $a=0$\index{Schwarzschild},
exhibiting for the first time the expected
degeneration associated with the photon sphere\index{photon sphere} 
of trapped null geodesics\index{trapped null geodesics} (see 
Section~\ref{TNG}),
were pioneered by~\cite{labasoffer, BlueSof0}, with unnecessary additional
degeneration on the horizon, however. This degeneration was related
to the fact that naively, the null generators of the horizon may seem 
to be subject to the trapping obstruction of Ralston, if one does not remark
that their tangent vectors shrink exponentially with respect to the stationary Killing field.
This is essentially the \emph{redshift effect}.\index{redshift effect} The degeneration at the 
horizon was then overcome
in~\cite{dr3} by using a new vector field multiplier current which quantifies
this celebrated effect. As we shall see, this is essential for the
stability of the construction. Analogues of $(\ref{protasn1b})$, $(\ref{protasn2})$ 
for $a=0$ have by now been obtained, extended
and refined by many authors~\cite{dr3, BlueSof, dr5, alinhac, newblue, marzuola}. 
See~\cite{jnotes} for a detailed review.

In passing from Schwarzschild to Kerr, two major difficulties appear immediately.
One is the superradiance effect mentioned earlier,
the other is the more complicated structure
of trapped null geodesics.
In view of the previous work described just above in Section~\ref{bouint}, 
the former difficulty has been completely resolved
in the cases under consideration. The latter difficulty was beautifully
expounded upon by Alinhac~\cite{alinhac}, who explicitly showed 
that currents of the form of~\cite{dr3} can never yield nonnegative definite
spacetime estimates on Kerr spacetimes for $|a|\ne 0$,
even in the high frequency approximation;
this is essentially related to the fact that there are trapped null geodesics
for an open range of $r$, whereas the construction of the currents of~\cite{dr3}
implicitly rely
on the fact that in Schwarzschild, all such geodesics approach the codimension-$1$
hypersurface $r=3M$.  
 
It turns out that the 
``exceptionalism'' of Schwarzschild from this point of view is
in some sense merely an accident of the projection of geodesic flow to the spacetime
manifold. On the tangent bundle, the dynamics of geodesic flow 
near trapped null geodesics is stable upon passing from $a=0$ to $|a|\ll M$.
This stability is in turn related  to the separability and thus 
complete integrability of geodesic flow for
Kerr metrics, first discovered by Carter~\cite{cartersep}.\index{separability}

The geometric
origin of the separability of geodesic flow
is subtle. Besides Killing fields $T$ and $\Phi$, the Kerr geometry admits a
nontrivial independent  Killing tensor~\cite{walker}. 
It turns out that using the Ricci flatness, this implies
both separability of geodesic flow and  
separability of the wave equation~\cite{cartersep2}. 
The latter provides a convenient way to frequency-localise the virial
currents ${\bf J}^{X,w}$ of~\cite{dr3, BlueSof} and subsequent papers, 
in a way intimately tied to both the local and global geometry of the solution.
This localisation allows us to capture the obstruction to integrated decay provided by
trapped null geodesics, overcoming the difficulty described in~\cite{alinhac}.

The properties of
geodesic flow are only suggestive of the high frequency regime of solutions
$\psi$ of $(\ref{WAVE})$.
To retrieve the original spacetime integral estimates proven in the
Schwarzschild case
with the Morawetz-type current, one must understand the behaviour
of such currents on all frequencies,
including low frequencies 
outside of the ``trapped'' regime. For the case $|a|\ll M$, one
can essentially infer positivity by stability considerations  from 
the previous constructions on Schwarzschild. In the general case $|a|<M$, however,
the existence of suitable currents providing positivity 
would depend on global geometric features
of the Kerr metric.  The geometric frequency-localisation we have employed
is particularly suitable for exploiting this. This is all the more important
to make contact with our forthcoming part III~\cite{drf1}, 
where low frequencies in the
non-trapped but superradiant
regime introduce a new important difficulty. See the discussion in our
companion paper~\cite{stabi}.

It is perhaps useful at this point to recall how boundedness was proven in~\cite{dr6}
for the wave equation on suitable axisymmetric, stationary perturbations of
Schwarzschild. After applying appropriate cutoffs and taking Fourier expansions associated
to the Killing directions $T$ and $\Phi$, one can partition general solutions
into their superradiant and non-superradiant part, where the latter part is
distinguished by the positivity of its energy flux through the horizon.
The crucial observation 
is that for sufficiently
small perturbations of Schwarzschild space, the superradiant part is non-trapped,
and a Morawetz-type current could be constructed without degeneration, essentially
from the original Schwarzschild construction and stability considerations. Adding,
a small part of the `red-shift vector field'\index{redshift effect}
component of the current as first introduced in~\cite{dr3}, and in addition, the
stationary
Killing vector field $T$, the boundary terms of this current could also be
shown to be positive, without destroying the positivity of the spacetime integral. 
This observation (together with a new robust argument to show boundedness
in the case where the energy flux is positive through the horizon) 
was sufficient to prove boundedness for the sum of the superradiant and
non-superradiant part.

Turning back to the problem at hand,
in view of the fact that boundedness has already been proven, we need
not worry here about the sign of boundary terms in our construction.
Of course,  the observation due to~\cite{dr3} that the boundary terms
can be made positive by adding a small amount of the red-shift\index{redshift effect} 
current can be
applied directly here. Thus, the proof of Theorem~\ref{nmt}
in fact reproves the boundedness statement, when specialised to the
case of Kerr spacetimes with $|a|\ll M$, and we have indeed set
things up so that  the boundedness statement (see~$(\ref{bndts1})$) 
is retrieved rather than used. This emphasizes a philosophical  point:
when superradiance\index{superradiance} is small, one need not distinguish the superradiant
part from the nonsuperradiant part, provided one is indeed proving dispersion
for the total solution. 
This distinction, however, is an important
aspect of the general non-axisymmetric case $|a|<M$. In particular, as in 
the general boundedness theorem~\cite{dr6}, it is essential to construct multipliers
separately for the superradiant and non-superradiant regimes. 
See the discussion in Section~\ref{versus}. These constructions are already
given in detail in Section 11 of our companion paper~\cite{stabi}.

\subsubsection{Red-shift\index{redshift effect} commutation and higher derivatives}
Non-linear applications make it imperative to understand boundedness
and decay properties for higher order quantities. In the usual application
of the vector field method,
these are typically proven via commutation with Killing fields, or, on dynamical
geometries, vector fields which become Killing asymptotically in time.
It is here essential that the set of commutators span a timelike direction,
so that all derivatives can then be recovered by elliptic estimates.

Already in the Schwarzschild case, one is faced with the difficulty   that
the span of the Killing fields does not include a timelike direction at points
on the event horizon. This difficulty was circumvented by Kay
and Wald in their celebrated~\cite{kw:lss} 
for the purpose of proving pointwise estimates for $\psi$
itself on the horizon. Unfortunately, the method of~\cite{kw:lss} was
very fragile. In addition,  the method was not able, for instance, to uniformly
bound transversal (to the horizon) derivatives of $\psi$.

It was first shown in the context of our general boundedness theorem of~\cite{dr6},
referred to above, that commutation by a suitable transverse vector field to the
horizon, though not Killing, has the property that the most dangerous error
terms which arise have a favourable sign.
This is yet another manifestation of the redshift\index{redshift effect}, and, as shown in Section 7 
of~\cite{jnotes},
is in fact true for all Killing horizons with positive surface gravity. 
Using this commutation, one can obtain higher
order decay results which do not degenerate at the horizon. 
We have included these statements in our results.
See Theorem~\ref{h.o.s.}.

\subsubsection{Further decay}
On the surface, the local integrated decay
statement of Theorems~\ref{nmt} and~\ref{nmt2} appears to be
a very weak statement of dispersion, seemingly
far from the type of statements necessary to obtain non-linear stability results as 
in~\cite{book}. It turns out, however, that it is the essential
ingredient in obtaining further decay.

This fact was already partially apparent from work 
in the Schwarzschild case. In addition to proving
a version of Theorem~\ref{nmt} for $a=0$ incorporating for the first time a vector
field estimate capturing the red-shift\index{redshift effect} effect, our~\cite{dr3} gave a method
for then inferring quantitative decay of energy through
a suitable foliation, through use of a suitably constructed conformal 
Morawetz multiplier. Independently, a related argument 
was developed in~\cite{BlueSter}.  The construction of the conformal
Morawetz multiplier appeared however to 
still be sensitive to the global geometry all the way up to the event horizon.
In~\cite{icmp} and our companion~\cite{drf2}, however, 
it is shown that for a wide class of metrics, (i) a local
 integrated decay statement as in Theorem~\ref{nmt}, (ii) a uniform
 boundedness statement as in~\cite{dr6}, 
and (iii) an additional estimate whose validity depends \emph{only
on the asymptotic behaviour of the metric at null infinity}
are together in fact sufficient to prove quantitative decay of energy (through a suitable
foliation) and pointwise decay 
of the form useful for non-linear stability problems. 
We note that this argument has other potential applications unrelated
to the black hole case. In particular, one need not assume the 
metric to be stationary (or even asymptotically stationary as $t\to \infty$; see
the recent nonlinear application of~\cite{Yang}).
It is fundamental, however, that the local
energy decay statement be nondegenerate at horizons, if the latter are present. 
For completeness, we will include 
in this paper a version of the resulting decay estimates (Theorem~\ref{DT})
for the Kerr case.
In this context, we note also an independent Fourier based method of obtaining
refined pointwise decay estimates~\cite{tatar} from Theorem~\ref{nmt}, 
relying heavily however on exact stationarity.

\subsubsection{Final remarks}
A version of Theorem~\ref{nmt} and its proof first appeared as 
Proposition~5.3.1 of our Clay lecture notes~\cite{jnotes}, 
where, in addition, an earlier method for
turning this into energy and pointwise decay (Theorem~5.2), closely
following~\cite{dr3},
was also presented.  
The original method of proof of Theorem 5.2 has been extended to yield additional
decay on Kerr by Luk~\cite{luk2}  via commutation by the scaling vector field, 
following a method he introduced in the Schwarzschild case~\cite{luk}.
Furthermore, this additional decay is then used in~\cite{luk2} to obtain a stability result
for the wave-map equation on Kerr for $|a|\ll M$. Note that 
this additional decay  is also retrieved in the statement of 
Theorem~\ref{DT} and its corollaries.

We give here a  new proof of Theorem~\ref{nmt} in order (i)
to give a self-contained presentation in paper form
not depending on the previous delicate
constructions on Schwarzschild (see \cite{dr3, BlueSof, dr5} and subsequent papers),
(ii) to unify it with the proof of Theorem~\ref{nmt2} which appears here for the first time,
and in addition, (iii) to introduce constructions and notations which
will be used in the forthcoming Part III~\cite{drf1} concerning the case of non-axisymmetric
solutions for the general parameter range $|a|<M$.
See also our companion~\cite{stabi} which gives a complete survey, including the
essential details of the constructions of~\cite{drf1}.

We mention finally that in parallel with the results of~\cite{jnotes} concerning decay, 
there are two independent other
approaches to the problem for $|a|\ll M$ which have since appeared, due
to Tataru--Tohaneanu~\cite{tattoh}, and, more recently, Andersson--Blue~\cite{anblue}.
The former replaces frequency localisation via Carter's separation with 
a technique based on the standard pseudo-differential calculus, whearas 
the latter  introduces an innovative alternative technique for frequency localisation 
via commutation by higher order operators related to the Killing tensor, 
but requires a high degree of regularity
and fast decay of initial data near spatial infinity.
Another important related development is due to Tohaneanu~\cite{toh}, who,
starting from integrated decay estimates
in~\cite{tattoh}, shows
Strichartz estimates with nice applications to semi-linear problems.

\subsection{Statement of the main theorems}
\label{themainsec}
\subsubsection{Notations}
The statement of the main theorem will depend on a number of definitions and notations,
which we briefly summarise here. These notions will be developed formally further in the paper.

The reader familiar with Carter--Penrose\index{Carter-Penrose diagramme}
diagrammatic representations may wish
to refer to the figure below:
\[
\begin{picture}(0,0)%
\includegraphics{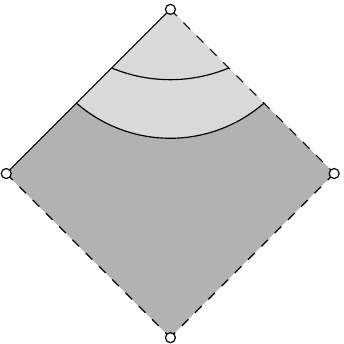}%
\end{picture}%
\setlength{\unitlength}{2960sp}%
\begingroup\makeatletter\ifx\SetFigFont\undefined%
\gdef\SetFigFont#1#2#3#4#5{%
  \reset@font\fontsize{#1}{#2pt}%
  \fontfamily{#3}\fontseries{#4}\fontshape{#5}%
  \selectfont}%
\fi\endgroup%
\begin{picture}(2180,2178)(5136,-6100)
\put(5926,-5386){\makebox(0,0)[lb]{\smash{{\SetFigFont{9}{10.8}{\rmdefault}{\mddefault}{\updefault}{\color[rgb]{0,0,0}$\mathcal{R}$}%
}}}}
\put(5551,-4561){\rotatebox{45.0}{\makebox(0,0)[lb]{\smash{{\SetFigFont{9}{10.8}{\rmdefault}{\mddefault}{\updefault}{\color[rgb]{0,0,0}$\mathcal{H}^+$}%
}}}}}
\put(6760,-4391){\rotatebox{315.0}{\makebox(0,0)[lb]{\smash{{\SetFigFont{9}{10.8}{\rmdefault}{\mddefault}{\updefault}{\color[rgb]{0,0,0}$\mathcal{I}^+$}%
}}}}}
\put(6783,-5771){\rotatebox{45.0}{\makebox(0,0)[lb]{\smash{{\SetFigFont{9}{10.8}{\rmdefault}{\mddefault}{\updefault}{\color[rgb]{0,0,0}$\mathcal{I}^-$}%
}}}}}
\put(5851,-4936){\makebox(0,0)[lb]{\smash{{\SetFigFont{9}{10.8}{\rmdefault}{\mddefault}{\updefault}{\color[rgb]{0,0,0}$\Sigma$}%
}}}}
\put(6001,-4336){\makebox(0,0)[lb]{\smash{{\SetFigFont{9}{10.8}{\rmdefault}{\mddefault}{\updefault}{\color[rgb]{0,0,0}$\phi_\tau(\Sigma)$}%
}}}}
\put(5926,-4636){\makebox(0,0)[lb]{\smash{{\SetFigFont{9}{10.8}{\rmdefault}{\mddefault}{\updefault}{\color[rgb]{0,0,0}$D^+(\Sigma)$}%
}}}}
\end{picture}%

\]

Let $\mathcal{R}$ denote the underlying manifold
of Section~\ref{mwbi}, and, for parameters $|a|<M$, let
$g_{M,a}$ denote the Kerr metric as defined in Section~\ref{Kstar}.
(Note that $\mathcal{R}$
corresponds to an exterior region including as its boundary the future event
horizon $\mathcal{H}^+$ \emph{but not}
the past event horizon.) Integrals on general spacelike or timelike submanifolds
of $\mathcal{R}$ will always be taken with respect to the induced volume form from 
a $g_{M,a}$ with parameters to be specified at each instance.

Let $T$ and $\Phi$ denote the vector fields of Sections~\ref{mwbi}--\ref{Kstar}
and let $\varphi_\tau$ denote the $1$-parameter
family of diffeomorphisms generated by $T$.
Note that for all $g_{M,a}$, $T$ will correspond to  the stationary Killing field, and
$\Phi$ will correspond to the Killing field generator of axisymmetry.

We will in general consider the class of so-called ``admissible'' hypersurfaces,
to be defined in Section~\ref{admissible}. The definition of admissibility depends
on the parameters $M>a_0$, and applies for all metrics $g_{M,a}$ with $|a|\le a_0$.
What is depicted in the diagramme is an admissible hypersurface ``of the second kind''.

The following notions will depend on the metric $g_{a,M}$, i.e.~on
both the parameters $|a|<M$ chosen:
Let $r$ denote the coordinate
of Section~\ref{Kstar}. Let $Z^*$ denote the unique smooth extension to
$\mathcal{R}$ of the coordinate vector field $\partial_{r}$ defined with respect
to Kerr-star coordinates, whereas 
let  $Z$ denote the unique smooth extension to ${\rm int}(\mathcal{R})$
of $\partial_{r}$ defined with respect to
Boyer-Lindquist coordinates (see Section~\ref{BLlc}). Note that $Z^*$ is transverse
to $\mathcal{H}^+$, whereas $Z$ would not extend smoothly to $\mathcal{H}^+$. 
Note also that for $|a|\ll M$, we have $Z^*=Z$ in $r\ge 5M/2$, by our definition
of Kerr star cordinates. Let $N$ denote the vector field
of Section~\ref{Nmult}. This is a timelike, $\varphi_\tau$-invariant 
vector field transverse to $\mathcal{H}^+$, 
which coincides with $T$ identically sufficiently far from the horizon.
Let $r_+$ denote the limiting value of $r$ on $\mathcal{H}^+$,
$r_+=M+\sqrt{M^2-a^2}$.
Let $\slashg$, $\nabb$ denote the induced metric and covariant derivative from $g_{a,M}$
on the $\mathbb S^2$ factors of $\mathcal{R}$ in the differential-topological 
product $(\ref{pmanif})$.
For a general spacelike hypersurface $\Sigma\subset \mathcal{R}$, 
we let $n_\Sigma$ denote
the unit future normal.  Let $D^+(\Sigma)$ denote the future domain of dependence of
$\Sigma$  in  $\mathcal{R}$.
Finally, 
for a general vector field $V$,
let ${\bf J}^V[\psi]$ denote the $1$-form `energy'
current defined in Section~\ref{multsandcomts}.

\subsubsection{The statement for $|a|\ll M$}
\begin{theorem}
\label{nmt}
Fix $M>0$. 
There exists a constant $\epsilon=\epsilon(M)>0$ such that
for all $0\le a_0\le \epsilon$,
the following statement holds:\index{fixed parameters! in statements of main theorems! 
$\epsilon$ ($\epsilon=\epsilon(M)$, determines smallness of $a_0$, $0\le a_0\le \epsilon$)}

There exist smoothly depending positive values
$s^\pm(a_0,M)$ with\index{fixed parameters! in statements of main theorems! 
$s^\pm$ ($s^\pm=s^\pm(a,M)$)}
with $s^\pm(0,M)=0$
such that the following holds.
Let $\Sigma$ be an arbitrary admissible hypersurface in $\mathcal{R}$ (of the
first or second kind) for parameter
$M$. Then, for all $\delta>0$,
there exists a constant $B=B(\Sigma,\delta)=B(\varphi_\tau(\Sigma),\delta)$
such that for all $|a|\le a_0$, considering the metric
 $g_{M,a}$ on $\mathcal{R}$, $r_+<5M/2<3M-s^-$, and such that
for all sufficiently regular solutions $\psi$ of~$(\ref{WAVE})$ with $g=g_{M,a}$
on  $D^+(\Sigma)$, 
the following inequalities hold (all metric dependent quantities referring
to $g=g_{M,a}$, as described above) for $\psi$
\begin{eqnarray}
\nonumber
\label{protasn1b}
\int_{D^+(\Sigma)}& \Big(r^{-1}(1-\eta_{[3M-s^-,3M+s^+]})(1-3M/r)^2 \big(|\nabb\psi|^2+r^{-\delta}((T\psi)^2)\big)\\
&+r^{-1-\delta}(Z^*\psi)^2+ r^{-3-\delta} (\psi-\psi_\infty)^2\Big)\\
\nonumber
&\le B 
\int_{\Sigma} {\bf J}_\mu^N[\psi]n^\mu_{\Sigma},
\end{eqnarray}
\begin{equation}
\label{bndts1b}
 \int_{\varphi_\tau (\Sigma)} {\bf J}_\mu^N[\psi]
n^\mu_{\Sigma}\le  B \int_{\Sigma} {\bf J}_\mu^N[\psi]n^\mu_\Sigma, \qquad
\forall\tau\ge 0,
\end{equation}
\begin{equation}
\label{fluxho...}
\int_{\mathcal{H}^+\cap D^+(\Sigma)} {\bf J}^N_\mu[\psi]n^\mu_{\mathcal{H}^+}
+(\psi-\psi_{\infty})^2
\le B \int_{\Sigma} {\bf J}_\mu^N[\psi]n^\mu_\Sigma,
\end{equation}
where $4\pi\psi_\infty^2= \lim_{r'\to\infty}\int_{\Sigma\cap\{r=r'\}} r^{-2}|\psi|^2$,
and $\eta$ denotes the indicator function.
\end{theorem}

The term `sufficiently regular' above just means solutions of $\Box_g\psi$ on $D^+(\Sigma)$
such that  on any compact spacelike hypersurface-with-boundary
$\tilde\Sigma\subset D^+(\Sigma)$, then $\psi$ 
and $n_{\tilde\Sigma}\psi$ have well-defined traces on $\tilde\Sigma$ which are
in $H^{1}_{\rm loc}(\tilde\Sigma)$, $L^2_{\rm loc}(\tilde\Sigma)$
respectively. 
There exists a unique such solution for appropriately defined initial data along $\Sigma$.
See Section~\ref{WPsec}.
Note that under these regularity assumptions, $\psi_\infty$ is well defined and
is finite if the right hand side of $(\ref{protasn1b})$ is finite.

Statement~$(\ref{protasn1b})$ is the analogue for Kerr of Morawetz's 
estimate $(\ref{analogueM})$
for Minkowski space and the integrated decay estimate of~\cite{dr3}
for Schwarzschild.  
The inclusion of the $(1-3M/r)^2$ factor is so as to
obtain uniform constants
as $a_0\to 0$, so as to retrieve in  particular the Schwarzschild result,
for which Theorem~\ref{nmt} provides in fact an independent proof.

Let us note that $(\ref{bndts1b})$
followed from our previous 
boundedness theorem of~\cite{dr6}, when the statement
of the latter is restricted to the Kerr family.  As discussed
above, our theorem gives an independent proof of this statement in the Kerr case.
The theorem of~\cite{dr6} also gave an inequality similar to but
weaker than $(\ref{fluxho...})$, with ${\bf J}^N_\mu[\psi]n^\mu_{\mathcal{H}^+}$ replaced
by $|{\bf J}^T_\mu[\psi]n^\mu_{\mathcal{H}^+}|$
and without the $0$'th order term.

We note immediately:
\begin{corollary}
Under the assumptions of the above theorem,
\begin{equation}
\label{protasn2}
\int_{D^+(\Sigma)} r^{-1-\delta} {\bf J}^N_\mu[\psi]
N^\mu+r^{-3-\delta} (\psi-\psi_\infty)^2
\le B 
\int_{\Sigma} ({\bf J}_\mu^N[\psi]+{\bf J}_\mu^N[T\psi])n^\mu_{\Sigma}.
\end{equation}
\end{corollary}

Let us note finally 
that the proof in fact yields an estimate for solutions to  the inhomogeneous
equation 
\begin{equation}
\label{inhomogeq}
\Box_g\psi =F,
\end{equation}
which we here omit for reasons of brevity.

\subsubsection{The statement for axisymmetric solutions}
Our second main theorem of the present paper concerns the entire subextremal
range $|a|<M$ but is restricted to axisymmetric solutions.
\begin{theorem}
\label{nmt2}
Fix $M>0$. Then for each $0\le a_0<M$, $\delta>0$, and each 
$\Sigma$ admissible for $a_0$, $M$, then for all $|a|\le a_0$,
there exist\index{fixed parameters! in statements of main theorems! 
$s^\pm$ ($s^\pm=s^\pm(a,M)$)}
values $s^\pm(a)$, with $r_+(a,M)<3M-s^-$,
a constant $B=B(a_0,\delta, \Sigma)$, and a cutoff
function $\chi_{a}(r)$, with $\chi_{a}=1$ for $r\ge 3M-s^-$
and $\chi_{a}=0$ for $r\le (r_+ +3M-s^-)/2$,
such that defining $\tilde{Z}^*=\chi_a Z+(1-\chi_a)Z^*$,
then the following statement holds.

For all $|a|\le a_0$,
 and all
sufficiently regular solutions $\psi$ of $(\ref{WAVE})$ with
$g=g_{M,a}$ on $D^+(\Sigma)$, which moreover
satisfy
\begin{equation}
\label{axisym}
\Phi(\psi)=0,
\end{equation}
the inequalities $(\ref{protasn1b})$--$(\ref{protasn2})$ hold, with $\tilde{Z}^*$
in place of $Z^*$.
\end{theorem}
The restriction  $(\ref{axisym})$ is removed in our companion paper~\cite{drf1}.

A posteriori, we note that one could have defined
Kerr star coordinates so as for the distinction between $\tilde{Z}^*$ and ${Z}^*$ 
to be unnecessary, but this would require making the definition depend
on certain constants that only appear in the proof of our theorem. 

\subsubsection{The higher order statement}

\begin{theorem}
\label{h.o.s.}
Let $M$, $a_0$, $a$, $g_{M,a}$, $s^\pm$, $\psi$ be as in Theorem~\ref{nmt} or~\ref{nmt2}.

Then, for all $\delta>0$ and all integers $j\ge 1$,
there exists a constant $B=B(\Sigma,\delta,j)=B(\varphi_\tau(\Sigma),\delta,j)$
such that
the following inequalities hold (all metric dependent quantities referring
to $g=g_{M,a}$, as described above) for $\psi$
\begin{eqnarray}
\nonumber
\label{protasn1}
\int_{D^+(\Sigma)} &
r^{-1-\delta}(1-\eta_{[3M-s^-,3M+s^+]})(1-3M/r)^2\sum_{1\le i_1+i_2+i_3\le j}
|\nabb^{i_1}T^{i_2}(Z^*)^{i_3}\psi|^2\\
\nonumber
&+r^{-1-\delta}\sum_{1\le i_1+i_2+i_3\le j-1}
\left(|\nabb^{i_1}T^{i_2}(Z^*)^{i_3+1}\psi|^2+|\nabb^{i_1}T^{i_2}(Z^*)^{i_3}\psi|^2\right)\\
&\le B\sum_{0\le i\le j-1}  \int_{\Sigma} {\bf J}_\mu^N[N^i\psi]
n^\mu_{\Sigma},
\end{eqnarray}
\begin{eqnarray}
\label{bndts1}
\nonumber
\int_{\varphi_\tau (\Sigma)}\sum_{1\le i_1+i_2+i_3\le j}
|\nabb^{i_1}T^{i_2}(Z^*)^{i_3}\psi|^2
&\le& B\sum_{0\le i\le j-1} \int_{\varphi_\tau (\Sigma)} {\bf J}_\mu^N[N^i \psi]
n^\mu_{\Sigma}\\
&\le&  B\sum_{0\le i\le j-1} \int_{\Sigma} {\bf J}_\mu^N[N^i\psi]n^\mu_\Sigma, \qquad
\forall\tau\ge 0
\end{eqnarray}
\begin{equation}
\label{bndts2}
\sum_{1\le i_1+i_2+i_3\le j+1, i_3\le j} \int_{\mathcal{H}^+\cap D^+(\Sigma)} |\nabb^{i_1}(n_{\mathcal{H}^+}^{i_2}
N^{i_3}\psi)|^2_{\slashg}
\le  B\sum_{0\le i\le j-1}\int_{\Sigma} {\bf J}_\mu^N[N^i\psi]n^\mu_\Sigma,
\end{equation}
where again, $Z^*$ should be replaced by $\tilde{Z}^*$ in the case of Theorem~\ref{nmt2}.
\end{theorem}
For given fixed $j\ge 1$, 
the regularity assumption on $\psi$ implicit in the above statement is
that on any compact spacelike hypersurface-with-boundary
$\tilde\Sigma\subset D^+(\Sigma)$, then $\psi$ 
and $n_{\tilde\Sigma}\psi$ have well-defined traces on $\tilde\Sigma$ which are
in $H^{j}_{\rm loc}(\tilde\Sigma)$, $H^{j-1}_{\rm loc}(\tilde\Sigma)$
respectively.

Note that the first line of $(\ref{bndts1})$ is simply an elliptic estimate
whereas the second asserts uniform boundedness of non-degenerate higher order energies.
Again, $(\ref{bndts1})$ in the case of the assumptions of Theorem~\ref{nmt}
followed from our previous 
boundedness theorem of~\cite{dr6}, when the statement
of the latter is restricted to the Kerr family.

\subsubsection{The statement of further decay}
In view of Theorem~\ref{nmt} and the results of~\cite{icmp, drf2}, we
may obtain further decay estimates, precisely of the form in principle
compatible with non-linear applications.

We give here only an example of the type of statement that
can be shown. The formulation of the result will
require expanding the set of multiplier and
commutator vector fields:
Let $\Omega_i$, $i=1,2,3$, denote a set of angular momentum operators for
the ambient Schwarzschild metric $g_M$
(see Section~\ref{angmo}), let $L$ denote the outgoing null vector
\[
L= \frac{\rho^2}{\sqrt{\Delta(\rho^2-2Mr)}}T+Z
\]
and let $\tilde{L}=\chi L$ where $\chi=0$ for $r\le 3M$ say, and $\chi=1$ for $r\ge 5M$.
We have

\begin{theorem} 
\label{DT}
Under the assumptions of the above theorems and $\Sigma$ an admissible hypersurface
of the second kind, then
there exists 
a constant $B=B(\Sigma)=B(\phi_\tau(\Sigma))$,  and, for each
$\delta>0$ a constant $B_\delta=B(\Sigma,\delta)=
B(\phi_\tau(\Sigma),\delta)$,
such that 
for all sufficiently regular solutions $\psi$ of~$(\ref{WAVE})$ on  $D^+(\Sigma)$,
\[
\int_{\varphi_\tau(\Sigma)} {\bf J}^N_\mu [\psi] n^\mu_{\varphi_\tau(\Sigma)} \le
B\, \tau^{-2}\, \sum_{0\le k\le 2}
\int_{\Sigma} {\bf J}_\mu^{N+r^2\tilde{L}}[T^k\psi],
\]
\[
\int_{\varphi_\tau(\Sigma)} {\bf J}^N_\mu [N\psi] n^\mu_{\varphi_\tau(\Sigma)} \le
B\tau^{-4+\delta}\sum_{i_1+i_2\le1}\sum_{0\le k\le 2}\sum_{\ell=1}^3
\int_{\Sigma} ({\bf J}_\mu^{N+r^2\tilde{L}}[N^{i_1}\Omega_l^{i_2}T^k\psi] + 
{\bf J}_\mu^{N+r^{4-\delta}\tilde{L}}[T^kL\psi]) n^\mu_{\Sigma}
\]
\end{theorem}

Note that the above statement, together with a re-application of Theorem~\ref{nmt}--applied now with $\Sigma$ replaced by $\phi_\tau(\Sigma)$, 
yields decay
  results (in $\tau$) for spacetime integral
  on the left hand side of $(\ref{protasn1b})$. This type of
  statement is in fact used
  in  the course of the proof of Theorem~\ref{DT}. 
  
From
the above results, applied to $\psi$, $T\psi$, etc., and standard application of
the vector field method, elliptic estimates, Sobolev inequalities, together
with the additional
red-shift  
commutation of Section~\ref{rscsec}, one can obtain
pointwise decay results of the form:
\begin{corollary}
Under the assumptions of the above, we
have
\[
\sup_{\varphi_\tau(\Sigma)}
r
|\psi-\psi_\infty|\le B \sqrt{E}\, \tau^{-1/2},
\]
\[
\sup_{\varphi_\tau(\Sigma)\cap\{r\le R\}} |\psi-\psi_{\infty}|\le B_{R,\delta} \sqrt{E} \tau^{-3/2+\delta},
\]
\[
\sup_{\varphi_\tau(\Sigma)\cap\{r\le R\}}|{\bf J}^{N}_\mu [\psi]N^\mu| 
\le B_{R,\delta} E \tau^{-4+\delta},
\] 
where in each inequality, $E$ denotes an
appropriate quadratic integral quantity defined on $\Sigma$.
\end{corollary}

All these estimates have extensions to arbitrary derivatives of the solution (of
arbitrary order), including derivatives transverse to the horizon.

From the point of view of applications to quasilinear problems,  past
experience would suggest that the above
type of results should be viewed as the definitive statements of decay
(cf.~the role of decay results in the stability of Minkowski space~\cite{book}).
Less decay may not be enough to absorb error terms in stability proofs, 
whereas more refined
decay statements (e.g.~the familiar Price law tails from the physics literature, 
see~\cite{dr1}),
requiring of course more restrictive assumptions on data,
are
potentially unstable in the context of the dynamical geometries of
interest, and
in any case, do not yield any added benefit for
non-linear stability proofs.

In fact, it is immediate
from the identity of Proposition~\ref{specialises..} 
extended (simply by continuity!)
to a region $r\ge r_+-\epsilon(a,M)$,
that any boundedness and suitable decay result holding uniformly up to the horizon is 
propagated for instance to $r_+-\epsilon$.
More interestingly, this observation (first due to~\cite{cbh})
 is stable in the consequence of dynamical spacetimes
where the event horizon and apparent horizon do not in general coincide. 
The interesting question for
black hole interiors is not, however, one of stability, but \emph{instability} 
as one approaches
the Cauchy horizon. This problem has a long tradition of study: 
see~\cite{ispo} and~\cite{mth, cbh} 
in the context of the 
Einstein-Maxwell-scalar field system under spherical symmetry. 
(In the Kerr solution, the Cauchy horizon corresponds to a hypersurface $r=r_-$ 
in maximally
extended Kerr.) Further discussion of the behaviour $\psi$ in the black hole
interior is thus best left to such
a context.

\subsection{Other decay-type statements}
There are many other important statements (mode stability, some partial nonquantitative
results
on azimuthal modes) which have appeared in the literature, at
various levels of rigour.  As a representative sample, we  
mention~\cite{whiting, press, hartle, fksy, fksy2}. 
See our~\cite{jnotes} for a discussion of this literature.

\subsection{Related spacetimes and equations}
\label{relateds}
We note finally that
there has also been interesting work on the Dirac equation~\cite{fksy0, hn} on Kerr backgrounds,
for which superradiance does not occur, and the Klein-Gordon
equation in~\cite{beyer, haf}, in the case of non-superradiant modes.
In the Schwarzschild case, there are a number of papers on related equations. 
We note especially the work of Blue~\cite{BlueMax} on the Maxwell equations.

All questions considered have analogues in higher dimensions. For 
the wave equation on higher dimensional Schwarzschild,
Schlue~\cite{schlue} 
has recently obtained both integrated energy decay and
decay bounds for the energy flux and the solution pointwise.
See also~\cite{laul}. 

When a cosmological constant $\Lambda$ is added to the Einstein equations to give
\[
R_{\mu\nu}= \Lambda g_{\mu\nu},
\]
there corresponds a related family of solution spacetimes, 
which in the case $\Lambda>0$ is
known as Kerr-de Sitter, whereas in the case $\Lambda<0$ is known as Kerr-anti de Sitter.
For the $\Lambda>0$ case, results have been obtained for the Schwarzschild-de Sitter
subcase by~\cite{dr4, bh} and recently extended in~\cite{sbvm, sbvm2}. 
Putting together the methods described here with~\cite{dr4} should yield
a generalisation to the Kerr-de Sitter case for small $a$, but it would be nice
to work this out explicitly.
In the case of Kerr-anti de Sitter, we note the recent work of Holzegel~\cite{kostakis}.

\section{The Kerr spacetime}

\subsection{The fixed manifold-with-boundary $\mathcal{R}$} 
\label{mwbi}
Let $\mathcal{R}$\index{sets! $\mathcal{R}$ (manifold-with-boundary on which $\psi$ is defined)} denote the manifold with boundary
\begin{equation}
\label{pmanif}
\mathcal{R}=\mathbb R^+\times\mathbb R\times  \mathbb S^ 2.
\end{equation}
We define standard coordinates $y^*$\index{coordinates! $y^*$ (fixed coordinate)}
 for $\mathbb R^+$, 
$t^*$\index{coordinates! $t^*$ (fixed or Kerr-star coordinate)} for $\mathbb R$, and standard spherical coordinates 
$(\theta^*, \phi^*)$\index{coordinates! $\theta^*$ (fixed or Kerr-star coordinate)}
\index{coordinates! $\phi^*$ (fixed or Kerr-star coordinate)}. Strictly speaking, the latter is only a coordinate
system  on the subset of $\mathbb S^2$ corresponding to
$(0,\pi)\times(0,2\pi)$, but we shall extend these functions
and the associated coordinate vector fields to all of $\mathbb S^2$, in the
usual fashion.

The collection $(y^*,t^*,\theta^*,\phi^*)$ define a coordinate system on $\mathcal{R}$,
global modulo the degeneration of the spherical coordinates remarked upon above.
We will refer to these as \emph{fixed coordinates}.

Let us introduce the notation $\mathcal{H}^+=\partial\mathcal{R}=\{y^*=0\}$.
We will refer to $\mathcal{H}^+$ as the \emph{event horizon}.\index{sets! $\mathcal{H}^+$
(event horizon, $\mathcal{H}^+=\partial\mathcal{R}$)}

Let us denote  $T=\partial_{t^*}$\index{vector fields! $T$ (Killing `stationary' vector field)}, 
$\Phi=\partial_{\phi^*}$\index{vector fields! $\Phi$ (Killing `axisymmetric' vector field)}, 
where, as remarked above, the latter
is to be understood as the extension of the coordinate vector field to all of
$\mathbb S^2$ in the usual way.

Let $\varphi_\tau$ denote the one-parameter family of diffeomorphisms
generated by $T$.\index{maps! $\varphi_\tau$ (one-parameter family of
diffeos generated by $T$)}

We have defined fixed coordinates precisely so that 
the differential 
structure of the ambient spacetimes 
$\mathcal{R}$, the vector fields which are to be Killing, and the location
of the horizon, are all independent of the Kerr parameters to be introduced
in what follows.

\subsection{Kerr-star coordinates}
\label{Kstar}
Let $\mathcal{P}\subset \mathbb R^2$ denote the subset $\{(x_1,x_2): 0\le x_1< x_2\}$.
Define a smooth map
\[
r:\mathcal{P}\times(0,\infty) \to (x_2+\sqrt{x_2^2-x_1^2}, \infty)
\]
such that $r|_{\{(x_1,x_2)\}\times (0,\infty)}$ 
is a diffeomorphism $(0,\infty)\to  (x_2+\sqrt{x_2^2-x_1^2}, \infty)$ which moreover
restricts to the identity map restricted
to $\{(x_1,x_2)\}\times(3x_2,\infty)$.

Now, 
for each fixed parameters $0\le a< M$,\index{metric parameters! 
$a$ (specific angular momentum parameter of Kerr metric)}
\index{metric parameters! $M$ (mass parameter of Kerr metric)}
the collection $(r(a,M,y^*),t^*,\theta^*,\phi^*)$\index{coordinates! $r$ (Kerr-star or Boyer-Lindquist coordinate)} 
\index{coordinates! $\phi^*$ (fixed or Kerr-star coordinate)}\index{coordinates! $t^*$ (fixed or Kerr-star coordinate)} 
\index{coordinates! $\theta^*$ (fixed or Kerr-star coordinate)}
determines a coordinate
system on
$\mathcal{R}$,
global modulo the well-known
degeneration of the spherical coordinates on $\mathbb S^2$.
These will be known as \emph{Kerr-star coordinates}.
In what follows, we shall denote $r(a,M, y^*)$ simply as $r$.
As opposed to our fixed coordinates on $\mathcal{R}$, this latter coordinate
depends on the parameters, although in the region $y^*\ge 3M$, then $r=y^*$ and is
thus independent of parameter $a$. The coordinate vector
fields $\partial_{t^*}$ and $\partial_{\phi^*}$ of Kerr-star coordinates always correspond
to the coordinate vector fields of the original 
fixed coordinate system, i.e.~for all $a, M$ we have
$\partial_{t^*}=T$, $\partial_{\phi^*}=\Phi$.

Let us define 
\[
r_\pm=M\pm \sqrt{M^2-a^2};
\index{fixed parameters! $r$-parameters! $r_+$ (larger root of $\Delta=0$)}
\index{fixed parameters! $r$-parameters! $r_-$ (smaller root of $\Delta=0$)}
\]
note that these are smooth functions of the parameters in the range in question,
i.e.~as functions $r_\pm:\mathcal{P}\to\mathbb R$. The horizon $\mathcal{H}^+$
corresponds to $r=r_+$.

Let us define $Z^*$ to be the smooth extension of the
Kerr-star coordinate vector field $\partial_{r}$ to $\mathcal{R}$.
\index{vector fields! $Z^*$ (the Kerr-star $\partial_r$ coordinate vector field)}

\subsection{The coordinate $r^*$}
Given parameters $|a|<M$, we 
define a rescaled version of the $r$ coordinate on $r>r_+$ by
\begin{equation}
\label{r*def}
\frac{dr^*}{dr}=\frac{r^2+a^2}{\Delta},
\end{equation}
where
\begin{equation}
\label{deltade}
\Delta=r^2-2Mr+a^2,\index{fixed functions! spacetime functions! $\Delta$ ($\Delta= r^2-2Mr+a^2$)}
\end{equation}
after suitably chosing a value for $r^*=0$. For definiteness, let us say $r^*(3M)=0$.
Note that $\Delta$ vanishes to first order on $\mathcal{H}^+$.\index{coordinates! $r^*$ (Regge-Wheeler-type rescaled $r$ coodinate $r^*=r^*(r)$)}
The coordinate range $r>r_+$ corresponds to the range $r^*>-\infty$.

\subsection{Boyer-Lindquist coordinates}
\label{BLlc}
Let $|a|<M$ be fixed parameters and $(t^*,r,\theta^*,\phi^*)$ be an associated
system of Kerr star coordinates on $\mathcal{R}$.
In ${\rm int}(\mathcal{R})$, i.e.~for $r>r_+$, we define
\[
t(t^*,r)= t^* -  \bar t(r)
\]
\[
\phi(\phi^*,r)= \phi^*- \bar \phi(r) \mod 2\pi
\]
\[
\theta=\theta^*
\]
where 
\begin{equation}
\label{tnearh}
\bar t(r)=r^*(r)-r-r^*(9M/4)+9M/4, \qquad {\rm for}\qquad r_+\le r\le 15M/8,
\end{equation}
\begin{equation}
\label{tnonearh}
\bar t(r)=0 \qquad {\rm for}\qquad r\ge 9M/4
\end{equation}
and $\bar{t}$ is chosen so as to satisfy everywhere
\begin{equation}
\label{forspacel}
\frac{d(r^*-\bar{t})}{dr}>0, \qquad 2-\left(1-\frac{2Mr}\rho^2\right)\frac{d(r^*-\bar{t})}{dr}>0.
\end{equation}
One easily constructs such a $\bar{t}$ by smoothing the function 
defined by $(\ref{tnearh})$ for $r \le 9M/4$ and $(\ref{tnonearh})$ for $r\ge 9M/4$.

The coordinates $(t,r,\theta, \phi)$ map ${\rm int}(\mathcal{R})$
to $(-\infty,\infty)\times(r_+,\infty)\times[0,\pi]\times[0,2\pi)$ and are global
modulo the degeneration of the spherical coordinates; we shall call these
 \emph{Boyer-Lindquist local coordinates}.\index{coordinates! $r$ (Kerr-star or Boyer-Lindquist coordinate)}
 \index{coordinates! $t$ (Boyer-Lindquist coordinate)}\index{coordinates! $\theta$ (Boyer-Lindquist coordinate)}
 \index{coordinates! $\phi$ (Boyer-Lindquist coordinate)}
 Note that all coordinates 
 except $\theta$ depend in fact on the parameters $a$, $M$.

As before,
the Killing fields $T$ and $\Phi$ correspond to the coordinate vector fields
$\partial_t$ and $\partial_\phi$.  
We apply here the standard abuse of notation in considering these coordinates
at $\theta=0,\pi$, and at $\phi=0$, where they are not regular.
 
 Let us define $Z$ to be (the  extension to  ${\rm int}(\mathcal{R})$ of) the
Boyer-Lindquist coordinate vector 
field $\partial_{r}$.\index{vector fields! $Z$ (the Boyer-Lindquist $\partial_r$ coordinate 
vector field)}
This vector field is significant as
it will define the directional derivative that does not degenerate in the integrated
decay estimate due to trapping.

\subsection{The Kerr metric}
Given parameters $|a|<M$, 
with the help of the associated (parameter-dependent) Boyer--Lindquist coordinates, we may now define the Kerr metric
$g_{M.a}$.  Let us first define the function
\[
\rho^2=r^2+a^2\cos^2\theta.
\index{fixed functions! spacetime functions! $\rho$ ($\rho=\sqrt{r^2+a^2\cos^2\theta}$)}
\]
This notation, along with the notation $\Delta$ from $(\ref{deltade})$, are traditional.
The Kerr metric $g_{M,a}$ is then defined as 
 the unique smooth extension to $\mathcal{R}$
of the tensor given on the Boyer--Lindquist chart by
\begin{align}
\label{eleme}
g_{M,a}=
-\frac{\Delta}{\rho^2}\left(dt-a\sin^2\theta d\phi\right)^2
+\frac{\rho^2}{\Delta}dr^2+\rho^2d\theta^2 +\frac{\sin^2\theta}{\rho^2}
\left(a\,dt-(r^2+a^2)d\phi\right)^2.
\end{align}
\index{metrics! $g_{M,a}$ (the Kerr metric with mass $M$ and rotation parameter $a$)}

In Boyer-Lindquist coordinates, 
one immediately recognizes the standard Schwarzschild  form
\begin{equation}
\label{incords}
-\left(1-\frac{2M}r\right) dt^2+\left(1-\frac{2M}r\right)^{-1}
dr^2 +r^2 (d\theta^2+\sin^2\theta \,d\phi^2)
\end{equation}
of the metric when $a=0$.
The significance of Boyer-Lindquist coordinates is that it is with respect to these
that the covariant wave operator $\Box_g$ separates.

That the expression $(\ref{eleme})$ indeed extends to a smooth metric on $\mathcal{R}$
is clear from examining its form in Kerr-star coordinates, in the region $r\le 15M/8$
where $\bar t$ is given by exactly $(\ref{tnearh})$.
There we compute 
\begin{align}
\nonumber
g_{M,a}=&-\left(1-\frac{2Mr}{\rho^{2}}\right)(dt^{*})^{2}+\frac{4Mr}{\rho^{2}}dt^{*}dr+\left(1+\frac{2Mr}{\rho^{2}}\right)dr^{2}\\
\nonumber
&\qquad -\frac{4aMr\sin^{2}\theta}{\rho^{2}}dt^{*}d\phi^{*}+\rho^{2}(d\theta^{*})^{2}\\ 
&\qquad +\frac{[(r^{2}+a^{2})^{2}-\Delta a^{2}\sin^{2}\theta]}{\rho^{2}}\sin^{2}\theta (d\phi^{*})^{2}-\frac{2a(2Mr+\rho^{2})\sin^{2}\theta}{\rho^{2}}dr\, d\phi^{*}.
\end{align}

We note from the explicit form of the metric $g_{M,a}$
that the vector fields $T$ and $\Phi$ are manifestly Killing.
Moreover, we easily see from the form of the metric that when $a\ne0$,  $T$ is
spacelike on the horizon $\mathcal{H}^+$, except where 
$\theta=0,\pi$, i.e.~on the so-called axis of symmetry.  We will have more to say
about this in Section~\ref{ergosec}.

Let us note that we have arranged the definition of Kerr-star coordinates
so that the hypersurfaces $t^*=c$ are spacelike (see the conditions $(\ref{forspacel})$).
In the region $r\le 15M/8$, we have in fact
\[
g(\nabla t^*, \nabla t^*)= -1-\frac{2Mr}{\rho^2}.
\]

A final remark: It would be possible to define the Kerr family on the fixed $\mathcal{R}$ 
so that one has smooth dependence up to and including the extremal case
$a=M$, but one would have to set things up slightly differently.

\subsection{The Carter-Penrose diagramme\index{Carter-Penrose diagramme}}
For the reader familiar with Carter-Penrose representations, then the region $\mathcal{R}$
corresponds to      
\[
\begin{picture}(0,0)%
\includegraphics{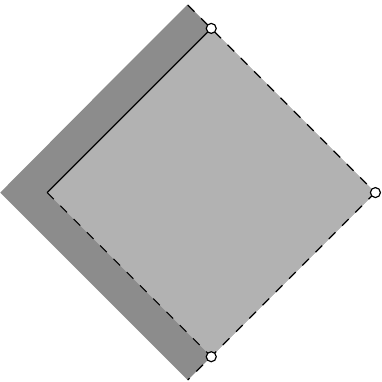}%
\end{picture}%
\setlength{\unitlength}{2960sp}%
\begingroup\makeatletter\ifx\SetFigFont\undefined%
\gdef\SetFigFont#1#2#3#4#5{%
  \reset@font\fontsize{#1}{#2pt}%
  \fontfamily{#3}\fontseries{#4}\fontshape{#5}%
  \selectfont}%
\fi\endgroup%
\begin{picture}(2441,2424)(4875,-6223)
\put(5551,-4561){\rotatebox{45.0}{\makebox(0,0)[lb]{\smash{{\SetFigFont{9}{10.8}{\rmdefault}{\mddefault}{\updefault}{\color[rgb]{0,0,0}$\mathcal{H}^+$}%
}}}}}
\put(6760,-4391){\rotatebox{315.0}{\makebox(0,0)[lb]{\smash{{\SetFigFont{9}{10.8}{\rmdefault}{\mddefault}{\updefault}{\color[rgb]{0,0,0}$\mathcal{I}^+$}%
}}}}}
\put(6783,-5771){\rotatebox{45.0}{\makebox(0,0)[lb]{\smash{{\SetFigFont{9}{10.8}{\rmdefault}{\mddefault}{\updefault}{\color[rgb]{0,0,0}$\mathcal{I}^-$}%
}}}}}
\put(5926,-5386){\makebox(0,0)[lb]{\smash{{\SetFigFont{9}{10.8}{\rmdefault}{\mddefault}{\updefault}{\color[rgb]{0,0,0}$\mathcal{R}$}%
}}}}
\end{picture}%

\]
Given appropriate definitions of asymptotic structure, then 
\[
\mathcal{R}={\rm clos}(J^-(\mathcal{I}^+)) \cap J^+(\mathcal{I}^-),
\]
interpreted in the topology of maximally extended Kerr (see~\cite{he:lssst}), where
$\mathcal{I}^+$, $\mathcal{I}^-$
are connected components of future and past null infinity corresponding to the same
end. It is in this sense that $\mathcal{R}$ corresponds to a domain of outer communications
(including future event horizon) of a spacetime containing both a black hole and
a white hole region. We shall not attempt here a formal development of these notions.

Alternatively, as is common in the formulation of black hole uniqueness
theory (see~\cite{alexakis}), one can characterize
$\mathcal{R}$ as follows: Let $\mathcal{U}$ denote a connected component
of the subset of points $x$ such that $T(x)$ is timelike, future pointing. Then
\[
\mathcal{R}={\rm clos}(J^-(\mathcal{U})) \cap J^+(\mathcal{U}).
\]
See
\[
\begin{picture}(0,0)%
\includegraphics{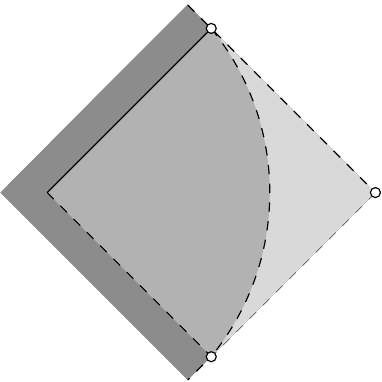}%
\end{picture}%
\setlength{\unitlength}{2960sp}%
\begingroup\makeatletter\ifx\SetFigFont\undefined%
\gdef\SetFigFont#1#2#3#4#5{%
  \reset@font\fontsize{#1}{#2pt}%
  \fontfamily{#3}\fontseries{#4}\fontshape{#5}%
  \selectfont}%
\fi\endgroup%
\begin{picture}(2441,2424)(4875,-6223)
\put(5551,-4561){\rotatebox{45.0}{\makebox(0,0)[lb]{\smash{{\SetFigFont{9}{10.8}{\rmdefault}{\mddefault}{\updefault}{\color[rgb]{0,0,0}$\mathcal{H}^+$}%
}}}}}
\put(6760,-4391){\rotatebox{315.0}{\makebox(0,0)[lb]{\smash{{\SetFigFont{9}{10.8}{\rmdefault}{\mddefault}{\updefault}{\color[rgb]{0,0,0}$\mathcal{I}^+$}%
}}}}}
\put(6783,-5771){\rotatebox{45.0}{\makebox(0,0)[lb]{\smash{{\SetFigFont{9}{10.8}{\rmdefault}{\mddefault}{\updefault}{\color[rgb]{0,0,0}$\mathcal{I}^-$}%
}}}}}
\put(5926,-5386){\makebox(0,0)[lb]{\smash{{\SetFigFont{9}{10.8}{\rmdefault}{\mddefault}{\updefault}{\color[rgb]{0,0,0}$\mathcal{R}$}%
}}}}
\put(6676,-4936){\makebox(0,0)[lb]{\smash{{\SetFigFont{9}{10.8}{\rmdefault}{\mddefault}{\updefault}{\color[rgb]{0,0,0}$\mathcal{U}$}%
}}}}
\end{picture}%

\]
Here $\mathcal{U}$ is the lightest shaded region and
$\mathcal{R}$ is the union of the two lighter shaded regions.

\subsection{The event horizon $\mathcal{H}^+$ as a Killing horizon}
\label{kilhorsec}
For each $g_{M,a}$ with $|a|<M$, the Killing vector field 
\[
K=T+  \frac{a}{r_+^2+a^2} \Phi
\]
is null and normal to $\mathcal{H}^+$.
Thus $\mathcal{H}^+$ is a Killing horizon.
The vector field $K$ is sometimes known as the
\emph{Hawking vector field}.\index{vector fields! $K$ (Hawking vector field)}
We note the identity
\begin{equation}
\label{eqforK}
\nabla_K K=
\kappa \, K,
\end{equation}
where
\begin{equation}
\label{eutuxws}
\kappa= \frac{r_+-r_-}{2(r_+^2+a^2)}>0.
\end{equation}
The\index{fixed functions! spacetime
functions! $\kappa$ (surface gravity $\kappa=\kappa(a,M)$)} quantity $\kappa$ is known as the \emph{surface gravity}.\index{surface gravity}
We note that $\kappa$ vanishes in the extremal case $|a|=M$.

We remark finally that the vector $K$ restricted to the horizon coincides with the
smooth extension of the coordinate vector field $\partial_{r^*}$
of the 
$(r^*,t,\theta, \phi)$ coordinate system.

\subsection{Asymptotics and angular momentum operators}
\label{angmo}
Besides the smooth dependence of $g_{M,a}$ on $a$, it will 
be useful to remark that for fixed $a$ and large $r$ values, $g_{M,a}$ is very close to $g_M$.

We may define a standard basis $\Omega_1$, $\Omega_2$, $\Omega_3$
of angular momentum operators\index{vector fields! $\Omega_i$ (angular momentum operators
for $g_M$, $i=1\ldots3$)} 
corresponding to $g_M$, such that moreover, $\Omega_1=\Phi$ say.
These are just a standard set of generators for the Lie algebra ${\rm so}(3)$
corresponding to the spherical symmetry of $g_M$.

We note that $\Omega_i$ span the tangent space to the $\mathbb S^2$
factors of the differential-topological product $(\ref{pmanif})$.

For future reference, let us introduce here the notation $\slashg$, $\nabb$\index{metrics! 
$\slashg$ (induced metric from $g_{M,a}$
on the $\mathbb S^2$
factors of the differential-topological product defining $\mathcal{R}$)}
\index{Lorentzian-geometric concepts! $\nabb$ (induced covariant derivative associated to $\slashg$)}
 to denote the induced metric and covariant 
derivative from $g_{a,M}$
on the $\mathbb S^2$ factors of $\mathcal{R}$ in the differential-topological 
product $(\ref{pmanif})$.

\subsection{The volume form}
\label{usefulcomps}
We shall often go back and forth between divergence
identities which arise geometrically and those which are computed explicitly
in various coordinates. For this the following remarks may be useful.

We note that
the volume form in Boyer-Lindquist coordinates satisfies
\[
dV= v(r,\theta) \, dt\, dr\, dV_{\slashg} \qquad \text{\rm with\ } v\sim 1
\]
using the alternative $r^*$ coordinate,
\[
dV =v(r^*,\theta) \,  dt\, dr^*\, dV_{\slashg} \qquad \text{\rm with\ } v\sim \Delta/r^2
\]
and in Kerr-star coordinates
\[
dV =v(r,\theta^*) \,  dt^*\, dr\, dV_{\slashg} \qquad \text{\rm with\ } v\sim 1.
\]

Let $\gamma$ denote the standard unit metric on the sphere 
in $(\theta,\phi)$ coordinates. 
We have that $\slashg \sim r^2\gamma$, and thus we may
replace $dV_{\slashg}$ in the above using
\[
dV_{\slashg} = v(r,\theta)\, r^2\sin\theta\, d\theta \,d\phi 
 \qquad \text{\rm with\ } v\sim 
1.
\]

\section{The geometry of Kerr}
\label{features}
The geometry of the Kerr solution has been treated at length in the literature, see
for instance~\cite{oneill}.
We will discuss  here only
those features that will be particularly relevant to the considerations
of this paper. 

\subsection{Surface gravity and the redshift\index{redshift effect}}
An important stabilising mechanism for the behaviour of waves near black hole event horizons
is what we shall here call the ``horizon-localised'' red-shift effect.  Recall 
that this is the red-shift relating two observers $A$
and $B$, where $B=\varphi_\tau(A)$ for $\tau>0$, 
both crossing the event horizon:
\[
\begin{picture}(0,0)%
\includegraphics{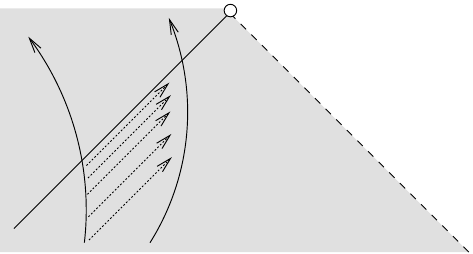}%
\end{picture}%
\setlength{\unitlength}{2368sp}%
\begingroup\makeatletter\ifx\SetFigFont\undefined%
\gdef\SetFigFont#1#2#3#4#5{%
  \reset@font\fontsize{#1}{#2pt}%
  \fontfamily{#3}\fontseries{#4}\fontshape{#5}%
  \selectfont}%
\fi\endgroup%
\begin{picture}(3763,2002)(4950,-6673)
\put(5776,-5311){\makebox(0,0)[lb]{\smash{{\SetFigFont{7}{8.4}{\rmdefault}{\mddefault}{\updefault}{\color[rgb]{0,0,0}$\mathcal{H}^+$}%
}}}}
\put(7726,-5611){\makebox(0,0)[lb]{\smash{{\SetFigFont{7}{8.4}{\rmdefault}{\mddefault}{\updefault}{\color[rgb]{0,0,0}$\mathcal{I}^+$}%
}}}}
\put(5326,-6511){\makebox(0,0)[lb]{\smash{{\SetFigFont{7}{8.4}{\rmdefault}{\mddefault}{\updefault}{\color[rgb]{0,0,0}$A$}%
}}}}
\put(6526,-5911){\makebox(0,0)[lb]{\smash{{\SetFigFont{7}{8.4}{\rmdefault}{\mddefault}{\updefault}{\color[rgb]{0,0,0}$B$}%
}}}}
\end{picture}%

\] 
At horizon crossing time, the frequency of waves received by $B$ (measured with
respect to proper time) from
$A$  are damped (in comparison to the frequency measured by $A$) by a factor
exponentially decaying in $\tau$.

The above effect depends in fact only on
the positivity of the \emph{surface gravity}.
Recall that we have already defined this notion
in the case of a Kerr geometry in Section~\ref{kilhorsec}.
More generally, if $\mathcal{H}^+$ is a Killing horizon  and $K$ is the distinguished\footnote{We take
typically a generator of the form $T+ c\Phi$, where $T$ is the stationary Killing field.}
Killing null generator
of such a horizon, then the surface gravity is defined to be the function $\kappa$
such that 
\begin{equation}
\label{SGdef}
\nabla_KK=\kappa\, K.
\end{equation}
Under various conditions on the ambient spacetime,
the function $\kappa$ can be shown to be constant. 

We leave it to the reader to directly relate the positivity of $\kappa$
(see $(\ref{SGdef})$) with frequency measurements of observers $A$ and $B$.
More relevant for the present considerations,
as shown in~\cite{jnotes}, the condition $\kappa>0$ translates into 
a positivity property near $\mathcal{H}^+$ in 
 the energy identity associated to a suitably
constructed vector field multiplier $N$.
In view of $(\ref{eutuxws})$, this construction is thus applicable in the Kerr case.
See Section~\ref{Nmult}. It is in this manifestation that the red-shift effect plays
a role.

We have already remarked that the surface gravity $\kappa$ 
vanishes precisely in the extremal 
case $|a|=M$.
It should be clear from this discussion that an additional new
difficulty that would arise in the extremal case $|a|=M$ is the failure
of the analogous identity for $N$.  Similar phenomena occur
in the extremal case of the spherically symmetric Reissner-Nordstr\"om 
family~\cite{he:lssst}\index{Reissner-Nordstr\"om}, i.e.~the
case of parameters
$Q=M$. In very recent work, Aretakis~\cite{aretakis} has constructed an analogue
of the vector field $N$ for the extremal Reissner-Nordstr\"om case,
which yields a spacetime integral estimate in which transversal
null derivatives of $\psi$ degenerate at the horizon. From this,
 the sharp
boundedness and decay results for waves on such a background
are obtained. In contrast to the non-extremal case, it is shown
moreover
in~\cite{aretakis} that
decay for translation-invariant
transversal derivatives to the horizon \emph{cannot} hold for general initial data,
and in fact, in general, higher transversal derivatives  blow up as advanced time progresses
along the horizon.

\subsection{The ergoregion\index{ergoregion} and superradiance\index{superradiance}}
\label{ergosec}
In general, the subset $\mathcal{E}\subset \mathcal{R}$\index{sets! $\mathcal{E}$
(the ergoregion)} 
where $T$ is spacelike is known
as the \emph{ergoregion}. The boundary $\partial\mathcal{E}$
of $\mathcal{E}$ (in the topology
of $\mathcal{R}$) is called the
\emph{ergosphere}. 

Under these conventions, we see easily that
in the Schwarzschild case $a=0$, $\mathcal{E}=\emptyset$.
If $a\ne0$, then $\mathcal{E}\ne\emptyset$, and
$\partial\mathcal{E} \cap \mathcal{H}^+=\{\theta^*=0,\pi\}\cap\mathcal{H}^+$
under our standard abuse of notation.

For all $0<|a|<M$, we have
\[
\sup_{\mathcal{E}} r=2M,
\]
and thus
\begin{equation}
\label{clearlyhave}
\lim_{a\to 0} \sup_{\mathcal{E}}r- r_+ =\lim_{a\to 0} (2M-r_+) =0.
\end{equation}
It follows that, 
for small $|a|\ll M$, the ergoregion\index{ergoregion} lies very close to the horizon, in particular,
it is contained in the region where the red-shift\index{redshift effect} mechanism operates, in the sense
of Section~\ref{Nmult}.
More on this below.

The ergosphere allows for a particle ``process'', originally discovered by Penrose~\cite{Penr2},
for extracting energy out of a black hole. This came to be known as the
\emph{Penrose process}.\index{Penrose process} In his thesis, Christodoulou~\cite{ri}
discovered the existence of a
quantity--the so-called \emph{irreducible mass} of the black  hole--which he showed
to be always
nondecreasing in a Penrose process. The analogy between this quantity and entropy
developed into a collection of ideas known as ``black hole thermodynamics''~\cite{beken}.
This remains the subject of intense investigation from the point of view
of high energy physics.

In the context of solutions $\psi$ of the
wave equation, if $p\in \mathcal{E}$, then it is no longer
the case that the energy density
 ${\bf J}^T_\mu [\psi] n^\mu(p)$ (see the notation of Section~\ref{multsandcomts}) is necessarily non-negative for 
 $n^\mu$ a future-directed timelike vector at $p$.
Thus, the conservation law for ${\bf J}^T_\mu[\psi]$ no longer yields a priori control
of a non-negative definite quantity in the derivatives of $\psi$.
The flux of energy to null infinity can thus be larger than the initial energy.
This is the phenomenon of \emph{superradiance}\index{superradiance},
first discussed by Zeldovich~\cite{Zeldovich}.
Moreover, a priori, this flux can be infinite. (Bounds for the strength of
superradiance were first derived heuristically by Starobinsky~\cite{Starobinsky}
in various asymptotic regimes.)
As discussed in the introduction,
it is for this reason
that the problem of obtaining any sort of boundedness property, even away from 
the horizon, is so difficult in Kerr. This was the difficulty which was first
overcome in~\cite{dr6}.

As in~\cite{dr6}, the present paper exploits the fact that the ``strength'' of superradiance\index{superradiance}  can be
treated as a small parameter. Essentially this means that given the vector
field $N$ referred to in the section above, and given an arbitrary $e>0$, then
for sufficiently small $a_0>0$, the vector field $T+eN$ is timelike with respect
to all $g_{M,a}$ with $|a|\le a_0$.  This is of course related
to $(\ref{clearlyhave})$. See Section~\ref{versus}.

Let us note that in the case of arbitrary $a_0<M$, $|a|\le a_0$,
superradiance\index{superradiance} is absent if $\psi$ is
assumed axisymmetric, i.e.~$\Phi(\psi)=0$. 
This is because, although $T$ fails to be everywhere causal,
the flux ${\bf J}^T_{\mu}[\psi]n_{\mathcal{H}^+}^\mu$ through the horizon is
nonnegative. This in turn is related to the fact that the null generator $n_{\mathcal{H}^+}$ 
of the horizon
is of the form $T+c(a,M)\Phi$.

In the general case of arbitrary $a_0<M$, $|a|\le a_0$, without axisymmetry,
the effect of superradiance\index{superradiance} cannot be treated as a small parameter and is
now a priori strongly coupled to the other difficulties. In particular, there are trapped
null geodesics (see Section~\ref{TNG}) 
that remain in the ergoregion\index{ergoregion} for all positive affine time. Nonetheless, it turns out
that the situation is much better than it initially appears and the difficulty
of
superradiance
can again be resolved, combining the ideas of~\cite{dr3} with those of the present
paper. See further
comments in Section~\ref{versus} and our companion paper~\cite{stabi}, where
the necessary 
constructions of our forthcoming Part III~\cite{ drf1} are given in detail.

\subsection{Separability\index{separability} of geodesic flow and 
trapped null geodesics\index{trapped null geodesics}}
\label{TNG}
The high frequency behaviour of solutions to wave equations is intimately related to
the properties of geodesic flow. 

In the Schwarzschild case $a=0$,
geodesic flow is easily understood because the Hamilton-Jacobi
equations separate. This in turn follows immediately from the dimensionality of the span of
the Killing fields $T, \Omega_1, \Omega_2, \Omega_3$. Geodesic flow on the 
Schwarzschild\index{Schwarzschild}
metric is described in detail in many textbooks.
One sees from the resulting equations that there are null geodesics which
for all affine time remain on the hypersurface $r=3M$, the so-called 
\emph{photon sphere}.\index{photon sphere}
With reference to suitable asymptotic notions, we can make the following more general statement:
If $\gamma(s)$ is a future-directed
inextendible null geodesic in Schwarzschild $(\mathcal{R},g_{M,0})$
with $\gamma(s)\in \mathcal{R}\setminus \mathcal{H}^+$
for all $s>s_0$,  such that moreover for all $p\in \mathcal{I}^+$, 
$\exists_s$ such that $\gamma(s)\not\in J^-(p)$, then $\lim_{s\to\infty} r(\gamma(s))=3M$.

As already mentioned in the introduction,
in view of general results due to Ralston~\cite{Ralston}, 
the above property restricts the type of decay
statement that one can hope to prove. The constructions of the energy currents 
of~\cite{BlueSter, dr3, dr5, BlueSof} and subsequent papers, used to prove integrated decay, are thus
intimately related to trapped null geodesics, in particular, the vector field
multiplier $X$ related to the currents ${\bf J}^{X,w}$ applied there,
vanishes precisely at $r=3M$.

Turning now to the general Kerr case,
remarkably, as discovered by Carter~\cite{cartersep},
geodesic flow admits, besides the conserved quantities associated
to the Killing fields $T$ and $\Phi$, a third non-trivial conserved quantity
(the `Carter constant'). Thus, 
geodesic flow remains 
separable and can be completely
understood.  
The dynamics is described in some detail in~\cite{chandrasekhar}.

In contrast to the Schwarzschild case,
there are now null geodesics with constant $r$ for an open range of 
Boyer-Lindquist-$r$ values. However, restricting to geodesics sharing
a fixed triple of the nontrivial conserved quantities, 
there is at most one $r$-value, depending on the triple, to which all 
null geodesics neither crossing
$\mathcal{H}^+$, nor approaching $\mathcal{I}^+$, must necessarily asymptote
to towards the future.

The above suggests that the underlying dynamics of geodesic flow is similar (for
the entire range $|a|<M$)
to the Schwarzschild case, but
to capture this in the high frequency regime,
the proper generalisation of the currents ${\bf J}^X$ 
must be frequency-localised so as to vanish at an $r$-value depending
on quantised versions of the conserved quantities corresponding to
trapped null geodesics.

There is a convenient way of doing phase space analysis in Kerr spacetimes
which is strongly linked to their geometry and the dynamics of geodesic flow. Namely,
as discovered again by Carter~\cite{cartersep2}, 
the wave equation itself can be formally separated, and the separation  
introduces frequencies $\omega$, $m$, and $\lambda_{m\ell}$.
which can be indeed thought of as quantised versions of  the conserved
quantities associated to geodesic flow.

None of these separability properties is in fact accidental!
Walker and Penrose~\cite{walker} showed shortly thereafter that
both the complete integrability of geodesic flow and the formal
separability of the wave equation
have their fundamental origin in the presence of a 
\emph{Killing tensor}. It turns out that, in view of Ricci flatness,
all three properties, i.e.~separability of wave equation, separability
of Hamilton-Jacobi, and existence of a Killing tensor, are in fact
\emph{equivalent}.
See~\cite{chong, kubiznak} 
for recent higher-dimensional generalisations
of this statement.

Carrying out the separation involves various technical issues as one must
take the Fourier transform in time, something which a priori is not possible.
On the other hand, it allows for a clean way to deal not only
with frequencies in the trapped regime (see Section~\ref{kloubi}), but with low frequencies
(see Section~\ref{elow}).
The flexibility provided by this geometric frequency
localisation is particularly important in our forthcoming Part III~\cite{drf1} 
where stability arguments
from delicate Schwarzschild constructions cannot be carried over.
(The reader can already refer to the discussion of Section 7
of our companion paper~\cite{stabi}
and the detailed constructions of Section~11 of~\cite{stabi} for the general $|a|<M$ case.)
We have thus taken this opportunity to formulate our constructions here
so as to yield  an independent proof of
the Schwarzschild case, which, though having the disadvantage of being dependent
on frequency localisation,
has the advantage that it does not require ``fine-tuning'' parameters in the
sense of all previous constructions~\cite{dr3, dr5, BlueSof, marzuola}.

\section{Preliminaries}

\subsection{Generic constants and fixed parameters}
We shall use $B$ to denote potentially large positive constants and $b$
potentially small positive constants, depending only on $M$, and, in the
case of Theorem~\ref{nmt2}, also on $a_0$, potentially blowing up in
the extremal limit $a_0\to M$. Note the algebra of constants:
$B+B=B$, $BB=B$, $Bb=B$, $B^{-1}=b$, etc.\index{general constants! $B$ (generic large
positive constant depending on $M$ and $a_0$)}
\index{general constants! $b$ (generic small positive constant depending on $M$ and $a_0$)}

Constants which additionally depend on other objects will be expressed
for example as follows: $B(\Sigma)$.

Our constructions will depend on various parameters\footnote{For the
convenience of the reader,
all constants and parameters are referred to in the index, with reference to the page
on which they are originally defined.}
whose choice is often deferred to late in the paper, for instance the frequency
parameters $\omega_1$, $\omega_3$, etc., introduced
in Section~\ref{freerange}, or the parameter $q$
of Section~\ref{oneofthesub}. 
Until a parameter is selected, e.g.~the parameter $\omega_1$, we shall use
the notation $B(\omega_1)$, etc., to denote constants depending
on $\omega_1$ \emph{in addition to $M$
and $a_0$}.
The final choices of parameters can be made to depend
only on $M$ in the case of Theorem~\ref{nmt},
and, on $M$,  $a_0$ in the case of Theorem~\ref{nmt2}.
Once such choices are made, $B(\omega_1)$ is replaced by $B$, following
our conventions.

We note finally that in the case of Theorem~\ref{nmt}, the parameter $a_0$
will be constrained to be small at various points in the paper, and the final smallness
assumption can be taken to the minimum of these constraints. All these constraints
depend finally only on  $M$, in accordance with the statement of the Theorem.

\subsection{Vector field multipliers and commutators}
\label{multsandcomts}

We recall the notation of~\cite{dr6}. Given a metric $g$,
let $\Psi$\index{$\Psi$-related! $\Psi$ (used for a general smooth function)} 
be sufficiently regular and satisfy
\[
\Box_g \Psi = F.
\]
We define
\[
{\bf T}_{\mu\nu}[\Psi]
\doteq \partial_\mu\Psi\partial_\nu\Psi -\frac12 g_{\mu\nu}g^{\alpha\beta}
\partial_\alpha\Psi \partial_\beta\Psi.
\]
\index{energy currents! $2$-currents! ${\bf T}_{\mu\nu}[\Psi]$ (energy-momentum tensor)}
Given a vector field $V_\mu$ and a function $w$ on $\mathcal{R}$, 
we will define the currents
\[
{\bf J}^V_\mu[\Psi] = {\bf T}_{\mu\nu}[\Psi] V^\nu
\]
\index{energy currents! $1$-currents! ${\bf J}^V_\mu[\Psi]$ (energy $1$-current associated to vector field $V$)}
\[
{\bf J}^{V,w}_\mu[\Psi] = 
{\bf J}_\mu^V[\Psi]+\frac18w\partial_\mu (\Psi^2)-\frac18(\partial_\mu w)\Psi^2
\]
\index{energy currents! $1$-currents! ${\bf J}^{V,w}_\mu[\Psi]$ (modified energy $1$-current associated to vector field $V$,
function $w$)}
\[
{\bf K}^V[\Psi] ={\bf T}_{\mu\nu}[\Psi]\nabla^\mu V^\nu
\]
\index{energy currents! $0$-currents! ${\bf K}^V[\Psi]$ (energy $0$-current associated to vector field $V$)}
\[
{\bf K}^{V,w}[\Psi] = {\bf K}
^V[\Psi] -\frac18\Box_gw (\Psi^2) +\frac14w \nabla^\alpha\Psi\nabla_\alpha\Psi
\]
\index{energy currents! $0$-currents! ${\bf K}^{V,w}[\Psi]$ (modified energy $0$-current associated to vector field $V$,
function $w$)}
\[
\mathcal{E}^V[\Psi] = -FV^\nu \Psi_{,v}
\]
\[
\mathcal{E}^{V,w}[\Psi]= \mathcal{E}^V(\Psi)-\frac14w\Psi F
\]

Applying the divergence identity between two homologous spacelike
hypersurfaces $\Sigma_1$, $\Sigma_2$, bounding a region $\mathcal{B}$,
with $\Sigma_2$ in the future of $\Sigma_1$, we obtain 
\[
\int_{\Sigma_2}{\bf J}^V_\mu [\Psi]n^\mu_{\Sigma_2}+
\int_{\mathcal{B}} {\bf K}^V[\Psi] + \mathcal{E}^V[\Psi]
=\int_{\Sigma_1}{\bf J}^V_\mu [\Psi]n^\mu_{\Sigma_1},
\]
where $n_{\Sigma_i}$ denotes the future directed timelike unit normal.

Let us note finally that for a fixed smooth function $\Psi:\mathcal{U}\to\mathbb R$,
$\mathcal{U}\subset\mathcal{R}$, then
${\bf T}_{\mu\nu}[\Psi]$ for $g_{M,a}$ will depend smoothly on the parameters $a,M$, i.e.~it
is smooth as a function $\mathcal{P}\times\mathcal{R}\to \mathbb R$.
Similarly, if $w$ is a fixed function and
$V$ is a fixed smooth vector field (or more generally, smooth maps $\mathcal{R}\to\mathbb R$,
$\mathcal{P}\times\mathcal{R}\to T\mathcal{R}$), then ${\bf J}^{V,w}_\mu$, ${\bf K}^{V,w}$,
$\mathcal{E}^{V,w}$, etc., are smooth in the above sense.
This remark is useful, for instance,  in carrying over positive definitivity properties 
from the Schwarzschild case, and more generally, in applying certain continuity
arguments in our forthcoming paper.

\subsection{Hardy inequalities}
At various points we shall refer to Hardy inequalities (see e.g.~before
Propositions~\ref{ftrs} and~\ref{lrp}) to estimate
a weighted $L^2$ norm (spacetime or spacelike) of $\psi$ from energy quantities.
In view of our comments concerning the volume form (see Section~\ref{usefulcomps}),
the reader can easily derive these
from the one-dimensional inequalities
\[
\int_{0}^2  x^{-1} |\log x|^{-2} f^2(x) \le C\int_{0}^{2} \left(\frac{df}{dx}\right)^2(x)dx+
C\int_1^2f^2(x) dx, 
\]
\[
\int_1^\infty  f^2(x)\le C\int_1^\infty x^2\left(\frac{df}{dx}\right)^2 (x)dx,
\]
where the latter holds for functions $f$ of compact support.

\subsection{Admissible hypersurfaces\index{admissible hypersurface}}
\label{admissible}
We shall here define admissible hypersurfaces of the first and second kind. 
The first case will represent spacelike hypersurfaces  connecting the future event horizon
to spacelike
infinity, and the latter case, to future null infinity. 

The prototype for admissibile hypersurfaces of the first kind is the hypersurface
\[
\Sigma_1\doteq \{t^*=0\}.
\]

The prototype for admissible hypersurfaces of the second kind is defined
as follows.  Let $\chi(z)$ be a nonnegative cutoff function such that $\chi=0$ for $z\le 0$
and $\chi(z)=1$ for $z\ge 1$
For fixed $M$, we choose an $R>0$ sufficiently large and an $\alpha>0$ sufficiently small, depending on $M$,
and consider the hypersurface 
\[
\Sigma_2\doteq \{t^*= \chi(\alpha\cdot (y^*-R)) \cdot (r+ M\log  r )\}.
\]

Note that $\Sigma_1$ and (for suitable choice\footnote{Let these choices be made
once and for all.} of $R$, $\alpha$) $\Sigma_2$ 
are spacelike for all $g_{M,a}$
where $0\le a_0<M$, $|a|\le a_0$.
Notice also that for all choices of $a$, 
the hypersurface $\Sigma_2$ asymptotes to null infinity with
respect to usual constructions of asymptotic structure.
For reasons concerning
dependence of constants on $a_0$ in the statements of our theorems, 
it is convenient to require these properties for all hypersurfaces to be considered.
This will thus be encorporated
into the definition below:

\begin{definition}
Let $0\le a_0< M$.
We will call a hypersurface-with-boundary $\Sigma$ \emph{admissible of
the first kind} with respect
to $a_0$, $M$, if the following hold:
\begin{enumerate}
\item
$\Sigma$ is spacelike with respect to  $g_{M,a}$ for all $|a|\le a_0$.
\item
$\Sigma\subset \mathcal{R}$ is closed as a subset, 
and $\partial\Sigma\subset \mathcal{H}^+$ is compact and thus
a topological sphere.
\item
There exists a $\tau\in (-\infty,\infty)$ such that for all $\epsilon>0$, there
exists an $y^*_{\epsilon}>0$ such that
\[
\Sigma\cap \{y^*\ge y^*_{\epsilon}\}\subset \bigcup_{|\tau'|\le \epsilon}
\varphi_{\tau+\tau'}(\Sigma_1)
\]
\end{enumerate}
We call
$\Sigma$  \emph{admissible of
the second kind}  with respect to $a_0$, $M$, if the above
hold where $\Sigma_1$ is replaced by $\Sigma_2$.

We will say that $\Sigma$ is admissible \ldots with respect to $M$ if
the there exists an $a_0>0$ sufficiently small such that the above is true.
\end{definition}

Note that for an admissible $\Sigma$, its future domain of dependence 
with respect to all $g_{M,a}$ is
given by
\[
D^+(\Sigma)\index{Lorentzian-geometric concepts! $D^+(\Sigma)$ (future domain of dependence of $\Sigma$)}
 =\bigcup_{\tau\in[0,\infty)}\varphi_\tau(\Sigma)
\]
The main theorem of the present paper  do not in fact distinguish
between the type of admissible
hypersurfaces. These definitions are useful 
more to keep track of the dependence of constants. 
In the theorem of~\cite{drf2}, stated here as Theorem~\ref{DT}, it
is essential to consider admissible hypersurfaces of the 
second kind.

\subsection{Well-posedness\index{well-posedness}}
\label{WPsec}
First let us quote a general well-posedness statement.
\begin{proposition}
\label{LWP}
Let $s\ge 1$, $M>a_0\ge 0$ be fixed, and
let $\Sigma$ be an arbitrary admissible hypersurface (of either kind) in $\mathcal{R}$. Let
$\uppsi \in H^s_{\rm loc}(\Sigma)$, $\uppsi' \in H^{s-1}_{\rm loc}(\Sigma)$.
 Then for all $|a|\le a_0$, there exists a unique solution $\psi$ of $\Box_g\psi=0$ with
 $g=g_{M,a}$ in $D^+(\Sigma)$
 such that 
 \[
 \psi\in C^0_{\tau\in[0,\infty)}(H^s_{\rm loc}(\varphi_\tau (\Sigma)))\cap C^1_{\tau\in [0\infty)}
 (H^{s-1}_{\rm loc}(\varphi_\tau(\Sigma)))
 \]
 and $\psi|_{\Sigma}=\uppsi$, $(n_{\Sigma}\psi)|_{\Sigma}=\uppsi'$.
 Moreover, if $X\subset \Sigma$ is open
 and $\tilde\uppsi\in H^s_{\rm loc}(\Sigma)$, $\tilde\uppsi' \in
 H^{s-1}_{\rm loc}(\Sigma)$
such that 
$\uppsi=\tilde\uppsi$ and 
$\uppsi'=\tilde\uppsi'$ in $X$, then the corresponding unique
solutions $\psi$, $\tilde\psi$ satsify $\psi =\tilde\psi$ in $D^+(X)$.

Finally the map $(\uppsi,\uppsi')\times a\mapsto \psi_a$, where
$\psi_a$ denotes $\psi$ above with $g=g_{M,a}$ is 
$C^0\times C^\infty$.
\end{proposition}

The above theorem implies that 
solutions of the wave equation arising from appropriate initial data exist globally, and
their regularity in $L^2$-based Sobolev spaces is preserved on any Cauchy surface,
and moreover, the solutions depend continuously on initial data and
smoothly on the parameters.

\subsection{A reduction}
\label{reducsec}
First let us note that by extending initial data, an easy domain of dependence argument, 
and the fact that $T$ is timelike for all Kerr metrics near infinity, one can show the following:
\begin{proposition}
\label{reducprop}
Fix $M>|a_0|\ge 0$.
Let $\Sigma$ be admissible with respect to $M$, $a_0$, of the first or second kind. 
Then there exists a constant $B=B(\Sigma,M)=B(\varphi_\tau(\Sigma),M)$
such that the following holds: 

Let $\psi$ be a solution of the wave equation
$(\ref{WAVE})$ for $g=g_{M,a}$ 
on $D^+(\Sigma)$, with $|a|\le a_0$, such 
that $\psi$, $n_\Sigma\psi$ are smooth and compactly supported on $\Sigma$.
Then there exists a $\tau'$ and
a smooth solution $\widetilde\psi$ of $(\ref{WAVE})$ on $\mathcal{R}$
such that $D^+(\Sigma)\subset \{t^*\ge \tau'\}$,
\[
\widetilde\psi = \psi
\]
in $D^+(\Sigma)$,
 $\widetilde\psi$, $\partial_t^*\widetilde\psi$ are of compact support on $\{t^*=\tau'\}$ and
\[
B^{-1} \int_{\Sigma} {\bf J}^N_\mu [\psi] n^\mu_{\Sigma}\le 
\int_{t^*=\tau'} {\bf J}^N_\mu[\widetilde\psi] n^\mu_{t^*=\tau'} \le 
B \int_{\Sigma} {\bf J}^N_\mu [\psi] n^\mu_{\Sigma},
\]
where $N$ is the vector field of Section~\ref{Nmult}.
\end{proposition}

Taking the difference $\psi-\psi_{\infty}$,
it follows easily from Proposition~\ref{reducprop}, a Hardy inequality
and a density argument 
that without loss of generality, we may assume in Theorems~\ref{nmt},~\ref{nmt2}
that $\Sigma=\{t^*=0\}$ and $\psi$ is smooth and of compact support. We shall assume
this reduction 
starting in Section~\ref{cutoffsec}.

\subsection{Uniform boundedness\index{uniform boundedness}}
Let us note that the results of~\cite{dr6} and the above proposition yield
in particular
\begin{theorem}
\label{btheorem2}
Fix $M>0$. Let $\Sigma$ be admissible with respect to $M$. 
Then there exists a constant $B=B(\Sigma,M)=B(\varphi_\tau(\Sigma),M)$
and a positive constant $\epsilon=\epsilon(M)$,
such that 
for all $|a|\le \epsilon$, if
$\psi$ be a sufficiently regular solution of $(\ref{WAVE})$ on $D^+(\Sigma)$
with $g=g_{M,a}$, then
\[
\int_{\varphi_\tau(\Sigma)} {\bf J}^N_\mu[\psi] n^\mu_{\varphi_\tau(\Sigma)} \le 
B \int_{\Sigma} {\bf J}^N_\mu [\psi] n^\mu_{\Sigma},
\]
\[
\int_{\mathcal{H}^+\cap D^+(\Sigma)}{\bf J}^{n_{\mathcal{H}^+}}_\mu
[\psi] n^\mu_{\mathcal{H}^+}  \le 
B \int_{\Sigma} {\bf J}^N_\mu [\psi] n^\mu_{\Sigma},
\]
where $N$ is the 
vector field of Section~\ref{Nmult}. The same statement holds with
arbitrary $|a|\le a_0<M$ if $\Sigma$ is assumed admissible
with respect to $M$, $a_0$, for $\psi$ with $\Phi\psi=0$.
\end{theorem}

We referred to $N$ of Section~\ref{Nmult} simply to fix the vector field
once and for all. If one does not wish to refer to that particular vector field,
one can alternatively take an arbitrary $\varphi_\tau$-invariant
timelike vector field $\tilde{N}$ such that
$\tilde{N}=T$ for sufficiently large $r$. 
The constant $B$ will then depend also on the choice of that $\tilde{N}$.

{\bf 
We shall avoid appealing to Theorem~\ref{btheorem2} in the case
of the proof of Theorem~\ref{nmt}, in order to have a self-contained proof.
See the discussion of Section~\ref{versus}.}

In the case of Theorem~\ref{nmt2}, the proof of
Theorem~\ref{btheorem2} should be considered as more elementary than the considerations
of the present paper in view precisely of the absense of superradiance, 
so, in particular, we shall freely use it.

This distinction between the proofs of Theorem~\ref{nmt} and~\ref{nmt2}
foreshadows the principle (essential in our forthcoming Part III~\cite{drf1})
that the superradiant and nonsuperradiant modes
must be treated separately in the general $|a|<M$ case. See the discussion in Section 7.3
of our companion~\cite{stabi}.

\section{A ${\bf J}^N$ multiplier current and the red-shift\index{redshift effect}}
\label{Nmult}
In this section, we define a timelike vector field $N$\index{vector fields! 
$N$ (`red-shift' vector field)} 
whose multiplier 
current ${\bf J}^N$ captures the red-shift effect.
The existence of this current is in fact a general property of stationary
black hole spacetimes
with Killing horizons of positive surface gravity. 
We shall discuss in Section~\ref{versus} how the 
properties of this current can be used to overcome the difficulty of superradiance
in the small $a_0$ case. 
We shall then recall  in Section~\ref{rscsec} 
the good commutation properties of the wave operator
$\Box_g$ with a related vector field $Y$ (used in the construction of $N$)
on $\mathcal{H}^+$,
properties which again arise from the positivity of the surface gravity.
Finally, in Section~\ref{hossec}, we shall apply commutations with $T$ and $Y$ 
to infer Theorem~\ref{h.o.s.} from Theorems~\ref{nmt}
and~\ref{nmt2}.

\subsection{The construction of $N$}
First, the completely general statement specialised to a given Kerr metric:
Theorem~7.1  of~\cite{jnotes} yields
\begin{proposition}
\label{specialises..}
Let $|a|<M$, $g_{a, M}$ be the Kerr metric and $\mathcal{R}$, etc., be as before.
There exist positive constants $b=b(a,M)$ and $B=B(a,M)$, 
parameters $r_1(a,M)>r_0(a,M)>r_+$, and\index{fixed parameters! $r$-parameters! $r_0$ (associated to $N$)}\index{fixed parameters! $r$-parameters! $r_1$ (associated to $N$)}
a $\varphi_\tau$-invariant timelike vector field $N=N(a,M)$ 
on $\mathcal{R}$ such that
\begin{enumerate}
\item
\label{fir}
${\bf K}^N[\Psi] \ge b\, {\bf J}^N_\mu [\Psi]  N^\mu$ for $r\le r_0$ 
\item 
\label{item2}
$-{\bf K}^N[\Psi] \le B\, {\bf J}^N_\mu [\Psi] N^\mu$, for $r\ge r_0$
\item
\label{lastitem}
$T= N$ for $r\ge r_1$,
\end{enumerate}
where the currents are defined with respect to $g_{M,a}$.
\end{proposition} 
\begin{proof}
We recall the proof briefly:  Given a Kerr metric $g=g_{M,a}$ and a constant $\sigma>0$, 
we first note that we may define a vector field $Y$ such that
 on $\mathcal{H}^+$ we have:
 \begin{enumerate}
 \item[(a)]
$Y$ is future directed
null with $g(Y,K)=-2$,
\item[(b)]
$\nabla_YY=-\sigma(Y+K)$ and
\item[(c)]
 $\mathcal{L}_{T}Y=\mathcal{L}_{\Phi}Y=0$,
 \end{enumerate}
where $K$ is the Hawking vector field of Section~\ref{kilhorsec}.\footnote{In
Theorem~7.1 of~\cite{jnotes}
we in fact constructed $Y$ to be only invariant under the flow defined by $K$, 
but in the presence of our two Killing fields $T$ and $\Phi$,
one can clearly arrange also for this.}
If we choose $\sigma$ sufficiently large and $r_0(a,M)$ sufficiently close
to $r_+$, then, defining
\[
N=K+Y, \qquad{\rm in}\qquad r\le r_0(a,M),
\]
then, using $(\ref{eutuxws})$, 
one sees easily that
requirement $\ref{fir}$ will hold, and moreover $N$ is timelike in this region.
The choice of $Y$, $r_0(a,M)$ can clearly be made to smoothly depend on $a$.
If $r_1(a,M)>r_0(a,M)$ is such that $T$ is timelike for $r\ge r_1(a,M)$ it suffices
to smoothly extend $N$ as a timelike vector field invariant 
to the Lie flow of $\Phi$ and $T$ such that
$N=T$ in $r\ge r_1(a,M)$, satisfying requirement~\ref{lastitem}. 
Requirement~\ref{item2} then follows
by compactness.
\end{proof}

Moreover, in view of~$(\ref{clearlyhave})$ we clearly have the following
\begin{proposition}
\label{rshere}
Fix $M>0$. Then a $\varphi_\tau$-invariant timelike vector field $N$ can be defined on the manifold $\mathcal{R}$
such that
for $a_0$ sufficiently small, and $|a|\le a_0< M$, 
then 1, 2 and 3 of Proposition~\ref{specialises..} hold
where 
then $b$ and $B$, $r_1$, $r_0$ can be chosen to
depend only on $M$ with $r_1\le 5M/2$, and\index{fixed parameters! $r$-parameters! $r_0$ (associated to $N$)}\index{fixed parameters! $r$-parameters! $r_1$ (associated to $N$)}
\ref{item2}.~can be strengthened to the statement
\[
-{\bf K}^N[\psi]\le B\, {\bf J}^T_\nu[\Psi] T^\nu, \qquad
\text{$T$ is timelike for $r\ge r_0$.}
\]
\end{proposition}

Let us note that in the case of $|a|\ll M$, both the above propositions follow
directly from
the existence
of a vector field $N$ on Schwarzschild satisfying the above statements
and smooth dependence of the Kerr family on $a$ (see Section~\ref{multsandcomts}),
as first observed and exploited in~\cite{dr3}. 
See also~\cite{dr4,dr5,dr6}. In this case then, we need thus not appeal to the more general
construction of~\cite{jnotes}.

In what follows, in the setting of Theorem~\ref{nmt}, we shall always
assume $a_0$ sufficiently small so that Proposition~\ref{rshere} applies.
Thus, in the setting of Theorem~\ref{nmt},  we have in particular
that $T$ is timelike for
$r\ge r_0$. In the setting of Theorem~\ref{nmt2},
all we know is that $T$ is timelike for $r\ge r_1$.

\subsection{The red-shift estimate\index{redshift effect}}
Given arbitrary $\tilde{r}$ satisfying $r_+<\tilde{r}\le r_0$ and $\tilde\delta>0$,
then applying the energy identity of ${\bf J}^{\chi N}$ 
in $D^+(\Sigma)\cap J^-(\varphi_\tau(\Sigma))$,
where $\chi(r)$ is a cutoff function with $\chi=1$ for $r\le \tilde{r}$ and
$\chi=0$ for $r\ge \tilde{r}+\tilde\delta$,
and a Hardy inequality
we obtain the following
\begin{proposition}
\label{ftrs}
Let $g=g_{M,a}$ for $|a|\le a_0<M$,
and let $r_0$ be as in the above Propositions.
Then the following is true.
Let $\Sigma$ be an admissible hypersurface for $M, a_0$.
For all $r_+\le \tilde{r}\le r_0$ and $\tilde\delta>0$,
there exists 
a positive constant $B=B(\Sigma,\tilde{r},\tilde\delta)=
B(\varphi_\tau(\Sigma),\tilde{r},\tilde{\delta})$,\index{variable parameters!
$\tilde{r}$ (used in application of red-shift estimates)}\index{variable parameters!
$\tilde{\delta}$ (used in application of red-shift estimates)}
such that for all solutions $\psi$ of $(\ref{WAVE})$ on $D^+(\Sigma)$
with $g=g_{M,a}$, then
\begin{align*}
\int_{D^+(\Sigma)\cap  J^-(\varphi_\tau(\Sigma))\cap\{r\le \tilde{r}\}}  &
({\bf J}_\mu^N[\psi]N^\mu +|\log(|r-r_+|)^{-2}||r-r_+|^{-1}\psi^2) \\ 
+&\int_{\mathcal{H}^+\cap D^(\Sigma)\cap   J^-(\varphi_\tau(\Sigma))}
{\bf J}^N_\mu[\psi]n_{\mathcal{H}^+}^\mu
+\int_{\varphi_\tau(\Sigma)\cap\{r\le \tilde{r}\}}{\bf J}^N_\mu[\psi] n^\mu \\
&\le B\int_{\Sigma} {\bf J}^N_\mu[\psi] n^\mu + B
\int_{D^+(\Sigma)\cap  J^-(\varphi_\tau(\Sigma))\cap \{\tilde{r}\le r\le \tilde{r}+\tilde{\delta}
\}}
({\bf J}_\mu^N[\psi]N^\mu
+\psi^2) 
\end{align*}
\end{proposition}
Recall that the additional
dependence of $B$ on $M$ and $a_0$ is now implicit according
to our conventions.

\subsection{The red-shift\index{redshift effect} vs.~superradiance\index{superradiance}}
\label{versus}
Superradiance concerns the fact that $T$ is not everwhere timelike. 
The smallness of superradiance with respect to the redshift in the setting of Theorem~\ref{nmt}
can be  quantified by the following
\begin{proposition}
\label{giasuper}
There exists an $a_0>0$ and 
 a differentiable nonnegative function $e_0(a)$\index{fixed parameters! small parameters! $e_0$ (auxilliary smallness
parameter $e_0(a)$ such that $T+e_0N$ is timelike in $r>r_+$ for $g_{M,a}$)} 
defined for $|a|\le a_0$,
such that 
$e_0(0)=0$, $a\cdot (e_0)'(a)\ge 0$, and such 
that $T+\tilde{e} N$ is\index{variable parameters!
$\tilde{e}$ (used in application of red-shift estimates)} timelike for $\tilde{e}>e_0$ in $r>r_+$,
and such 
that for all $\tau'\ge\tau$,
the following inequalities hold:
\label{inlieuof}
\[
\int_{\Sigma_{\tau'}}{\bf J}^{T+\tilde{e}N}_\mu[\psi]N^\mu \le\int_{\Sigma_\tau}
{\bf J}^{T+\tilde{e}N}_\mu[\psi]N^\mu
+B\tilde{e}\int_{D^+(\Sigma_\tau)\cap  J^-(\varphi_{\tau'}(\Sigma))\cap \{\tilde{r}\le r\le \tilde{r}+
\tilde{\delta}\}}
{\bf J}_\mu^N[\psi]N^\mu
\]
\begin{align*}
\int_{\mathcal{H}^+\cap\{\tau\le t^*\le \tau'\}}{\bf J}^{T+\tilde{e}N}_\mu[\psi]n^\mu_{\mathcal{H}^+}
\le &\int_{\Sigma_\tau}
{\bf J}^{T+\tilde{e}N}_\mu[\psi]N^\mu\\
&+B\tilde{e}\int_{D^+(\Sigma_\tau)\cap  J^-(\varphi_{\tau'}(\Sigma))\cap \{\tilde{r}\le r\le \tilde{r}
+\tilde{\delta}\}}
{\bf J}_\mu^N[\psi]N^\mu,
\end{align*}
with $B=B(\Sigma,\tilde{r},\tilde{\delta})$, $r_+<\tilde{r}\le r_0$, $\tilde{\delta}>0$, as in the previous 
proposition.
\end{proposition}
\begin{proof}
The existence of an $e_0(a)$ as described so that $T+\tilde{e}N$ is timelike is clear.
Now just add $\tilde{e}$ times the previous proposition to the energy identity of ${\bf J}^T$,
for an $\tilde{e}$ such that $T+\tilde{e}N$ is timelike, and use that $T$ is Killing,
i.e.~${\bf K}^T=0$, and that $T$ is timelike for $r\ge r_0$, and thus
${\bf J}_\mu^{T+\tilde{e}N}[\psi]N^\mu \sim {\bf J}_\mu^T[\psi]N^\mu $ on $r\ge r_0$.
\end{proof}

The above proposition is of course implied by Theorem~\ref{btheorem2} of~\cite{dr6}
but we  include it as an independent statement so as to circumvent appeal
to this theorem in the case of the proof of Theorem~\ref{nmt}. For this,
we shall also need the following
\begin{proposition} 
\label{needthis2}
Given arbitrary $\varepsilon>0$, there exists an $a_0>0$ depending on
$\varepsilon$, such that
for $|a|\le a_0$, and $1\ge \tilde{e} \ge e_0(a)$,
\begin{equation}
\label{firh}
\int_{\varphi_{\tau}(\Sigma)} {\bf J}^{T+\tilde{e} N}_\mu[\psi]N^\mu 
\le(1+B\tilde{e}) \int_{\Sigma} {\bf J}^{T+\tilde{e} N}_\mu[\psi]N^\mu, \qquad\forall 0\le\tau\le
\varepsilon^{-1},
\end{equation}
\begin{equation}
\label{sech}
\int_{D^+(\Sigma)\cap J^-(\varphi_{\varepsilon^{-1}}(\Sigma))\cap\{r_0\le r\le r_1\}}
{\bf J}^{N}_\mu[\psi] N^\mu\le B \int_{\Sigma} {\bf J}^{N}_\mu[\psi]N^\mu.
\end{equation}
\end{proposition}
\begin{proof}
The first inequality follows immediately
from the second. The second inequality follows from the fact that 
$(\ref{sech})$ holds for the Schwarzschild metric $g_M$ for all $\varepsilon$
(note the restriction on $r_1$ in Proposition~\ref{rshere} which guarantees
that the domain of integration on the left hand side is separated from the 
Schwarzschild photon sphere)
together with the stability considerations of Section~\ref{multsandcomts}.

We note that the necessary Schwarzschild result was proven in~\cite{dr3}.
The current paper can be read independently of~\cite{dr3}, as follows:
While the present proposition is appealed to 
in the proof of Theorem~\ref{nmt}, if the latter is specialised to the case $a_0=0$,
then, just as in the $m=0$ case, 
the use of the present proposition  is  in fact
not necessary, as appeal has only been made to circumvent use
of Theorem~\ref{btheorem2} in the case $a\ne 0$. Thus, the proof of Theorem~\ref{nmt}
specialised to $a_0=0$ yields the present proposition, which can then
be used in the proof of Theorem~\ref{nmt} for the general $|a|\ll M$ case.
\end{proof}

The significance of the above two propositions is the following: As in 
the proof of Theorem~\ref{btheorem2} of~\cite{dr6}, we shall be able to absorb
the boundary terms of our virial identities for sufficiently small $a$ using solely 
Propositions~\ref{giasuper} and~\ref{needthis2}
exploiting the fact that one can take $\tilde{e}$ very close to $0$. It is the use of
this argument that mathematically represents one of the insights first
used in~\cite{dr6}, namely:
\begin{enumerate}
\item[]
\emph{If superradiance\index{superradiance} is sufficiently weak, it can be overcome 
provided one can construct
a virial current with positive definite divergence,
whose boundary terms can be controlled with the positive
definite current formed by  adding a small amount of the red-shift
current ${\bf J}^N$ to the conserved
energy current ${\bf J}^{T}$.}
\end{enumerate}
What we are essentially able to avoid using here is the second
main insight of~\cite{dr6}, summarised as follows:
\begin{enumerate}
\item[]
\emph{To obtain just the boundedness statement,
it suffices to construct a virial current as above 
for the projection of the solution to its superradiant frequencies. Such
a projection can be defined for arbitrary axisymmetric stationary black hole
spacetimes whose Killing fields span the null generator of the horizon.
Moreover, for suitably small such perturbations of Schwarzschild,
this projection is  not ``trapped'', thus the existence of the virial current
follows from stability considerations from Schwarzschild!}
\end{enumerate}
The above insight allowed  boundedness to be understood without understanding
decay and without understanding trapping.  This is why the result of~\cite{dr6} is
so general, not requiring the properties of geodesic flow which are so crucial
for the present paper. 
We have no need for the second insight here, precisely because,
\emph{for the special case of the Kerr solution},
we are indeed proving decay, that
it is to say, we are indeed producing such a virial current for the total solution.

It is to highlight this distinction that we have 
been careful to avoid use of our general boundedness theorem~\cite{dr6} and
its superradiant projection. 
The reader hoping that one can understand the Kerr family without understanding
superradiance\index{superradiance} will be severely disappointed however!
In particular, the superradiant projection of~\cite{dr6}
is used heavily in our forthcoming~\cite{drf1} which considers the general case $|a|<M$, 
precisely because superradiance can no longer
be treated simply as a small parameter and the above insight is not sufficient.
See the discussion in Section 7.3 of our companion paper~\cite{stabi}, and 
Section~11 of~\cite{stabi}
for the details of the necessary constructions in all frequency ranges for the general case.

\subsection{Red-shift commutation}
\label{rscsec}
We specialise Theorem 7.2 of~\cite{jnotes}  to the Kerr case.

\begin{proposition}
\label{commuprop}
Let $g=g_{M,a}$ and $Y$ be as in Proposition~\ref{specialises..}. 
Then,
on $\mathcal{H}^+$, extending $K$
to a translation invariant standard null frame $E_1,E_2, K, Y$, then
for all $k\ge 0$ and multi-indices ${\bf m}=(m_1,m_2,m_3.m_4)$,
\[
\Box_g(Y^k\Psi)=\kappa_k Y^{k+1}\Psi + \sum_{|{\bf m}|\le k+1, m_4\le k} c_{{\bf m}} 
E_1^{m_1}E_2^{m_2}L^{m_3}Y^{m_4}\Psi
\]
where $\kappa_k>0$.
\end{proposition}
The above proposition, which is another manifestation of the red-shift effect, effectively
allows us not only to apply a transversal vectorfield to the horizon as a multiplier,
but also as a commutator. This is fundamental for retrieving higher order statements
as in Theorem~\ref{h.o.s.}.

\subsection{The higher order statement}
\label{hossec}
Commuting the wave equation with $T$ and $Y$, in the latter case applying
Propostion~\ref{commuprop}, we show the statement
of  Theorem~\ref{h.o.s.}, given
Theorems~\ref{nmt}   and~\ref{nmt2}.

For $j=1$:
We commute the equation first with $T$. It follows that
the estimates of Theorem~\ref{nmt}
or~\ref{nmt2} hold with $T\psi$ in place of $\psi$. 
Now we commute also with  $Y$ and follow~\cite{dr6, jnotes}.
The quantity $Y\psi$ satisfies a wave equation with an inhomogeneous term,
but by Proposition~\ref{commuprop}, the most dangerous term arising in 
$\mathcal{E}^N[Y\psi]$ (recall the notation of Section~\ref{multsandcomts}) has
a good sign sufficiently near the horizon. In the case of the assumptions
of Theorem~\ref{nmt},
we may assume that $T$ is strictly timelike outside this region. Elliptic estimates
then suffice to bound all remaining terms in $\mathcal{E}^N[Y\psi]$ from the good terms in
${\bf K}^N[Y\psi]$ and those already bounded by
the commutation by $T$. 
In the case
of $m=0$, $T$ is effectively timelike everywhere except the horizon from the perspective
of axisymmetric solutions, in the sense
that for every $x\in \mathcal{R}\cap\mathcal{H}^+$ there is a $T+c(x)\Phi$ which is
timelike in a neighborhood of $x$.
Thus, again, all terms not manifestly controlled can be easily obtained by elliptic estimates in view
of the fact that $NN\psi$ is controlled locally in $L^2$. 

The case of higher $j$ is similar and the required positivity of the most dangerous
additional terms appearing in $\mathcal{E}^N[Y^j\psi]$ is precisely
that provided by Proposition~\ref{commuprop} above for $j=k\ge 2$.

Recall that in our forthcoming Part III~\cite{drf1}, it will be shown that the statement of Theorem~\ref{nmt2} is true without the restriction
$(\ref{axisym})$. The above arguments can be extended to show that, in
general, given parameters $M>a$ for which the
statement of Theorem~\ref{nmt2} holds without the restriction $(\ref{axisym})$, 
one can infer the validity of Theorem~\ref{h.o.s.} for those parameters $M$, $a$.
Let us first note that Proposition~\ref{commuprop} still applies, yielding
positivity for the most dangerous terms in $\mathcal{E}^N[Y\psi]$ in a small neighborhood
of the horizon. 
The only additional difficulty is that one can no longer appeal to smallness of $a_0$
(as in the case of Theorem~\ref{nmt}) to argue that $T$ is timelike outside this region
and thus, outside this region, all second derivatives can be controlled by the integrated
decay bound applied to $T\psi$ (a statement which, as we saw above,
 can similarly be inferred in the case of Theorem~\ref{nmt2} using the assumption of axisymmetry).
In the more general setting, in addition to commutations with $T$ and $Y$ as above,
we must simply also commute
with  $\chi\Phi$
(where $\chi=1$ for $r\le R$, and $\nabla\chi$ is thus supported away from trapping), 
and exploit the fact that the span of $T$ and $\Phi$
is timelike away from the horizon. The error terms resulting in the commutation
with $\chi\Phi$ are easily absorbed with the integrated decay statement. 
Thus, commutations with $Y$, $T$ and $\chi\Phi$ ensure that one always
has a timelike direction in the span of commutators up to and including the horizon.
Application of elliptic estimates allows one as before to absorb all error
terms and obtain all second derivatives. One then proceeds by induction on $j$. 
We shall insert the details of this more general 
argument with the completion of Part III~\cite{drf1}.

\section{A current for large $r$}
We construct here a current ${\bf J}^{X,w}$ as follows.  Recall (see e.g.~\cite{dr6}) that
in the Schwarzschild geometry $g_M$, defining 
\[
X= f(r^*)\partial_r
\]
\[
w=2f'(r^*) +4\frac{1-2M/r}{r}f-2\delta\frac{1-2M/r}{r^{1+\delta}}f,
\]
we have 
\begin{eqnarray*}
{\bf K}^{X,w}_{g_M}[\Psi]	&=&	\left(\frac{f'}{1-2M/r}-\frac{f\delta}{2r^{1+\delta}}\right)(\partial_{r^*}\Psi)^2 +
\frac{f\delta}{2r^{1+\delta}}(\partial_t\Psi)^2
\\&&\qquad+
f\left(\frac{r-3M}{r^2}-\frac{\delta(r-2M)}{2r^{2+\delta}}\right)|\nabb\Psi|^2-\frac12\Box w \Psi^2
\end{eqnarray*}
where $'$ denotes differentiation with respect to $r^*$.
Let $\chi$ be a cutoff function such that
$\chi=1$ for $r\ge R$, $\chi=0$ for $r\le R-1$ and let $\delta>0$.
Choose now 
\[
f=\chi (1-r^{-\delta})\partial_r.
\]
We note that for $R$ sufficiently large\index{fixed parameters! $r$-parameters! $R$ (parameter associated with large-$r$
current)}, we have
\begin{eqnarray*}
{\bf K}^{X,w}_{g_M}(\Psi)	&\ge&	b(\delta)\left( r^{-1-\delta}(\partial_{r}\Psi)^2 +
r^{-1-\delta}(\partial_t\Psi)^2+r^{-1}|\nabb\Psi|^2_{g_M}+r^{-3-\delta}\Psi^2\right)
\end{eqnarray*}
and that this inequality is preserved when $X$, $w$, are defined for $g_{M,a}$
again as above, where $|\nabb\Psi|^2_{g_M}$ is now replaced by $|\nabb\Psi|^2_{\slashg}$
defined in the sense of Section~\ref{usefulcomps}.

From the energy identity of ${\bf J}^{X,w}$  and
a Hardy inequality to control $0$'th order terms on the boundary, 
we have
\begin{bigprop}
\label{lrp}
Fix $M>0$, let $\Sigma$ be an admissible hypersurface for $M$, and let $e_0$
be as above. For each
$\delta>0$, there exist positive  values\index{fixed parameters! small parameters! $\delta$ (associated with the large-$r$ current)}
$R_0<R$, and positive constants $B=B(\delta, \Sigma)=B(\delta,\varphi_\tau(\Sigma))$ 
such that for all Kerr metrics $g_{M,a}$ with $|a|\le a_0<M$, if
$\psi$ denotes a solution of $(\ref{WAVE})$ with $g=g_{M,a}$ and $\psi_{\infty}=0$, then
for all $\tau\ge 0$
\begin{align*}
\int_{D^+(\Sigma)\cap J^-(\varphi_\tau(\Sigma))\cap\{r\ge R\}}& r^{-1}( r^{-\delta}|\partial_r\psi|^2 
 +r^{-\delta}|\partial_t\psi|^2+|\nabb \psi|^2_{\slashg}+
r^{-2-\delta}\psi^2)\\
\le& B\int_{\Sigma} {\bf J}^{T+e_0N}_\mu[\psi] n^\mu_\Sigma+
 B\int_{\varphi_\tau(\Sigma)} {\bf J}^{T+e_0N}_\mu[\psi] n^\mu_\Sigma
 \\
&+ B
\int_{D^+(\Sigma)\cap J^-(\varphi_\tau(\Sigma))\cap\{R_0\le r\le R\}}
( |\partial_r\psi|^2  +|\partial_t\psi|^2+|\nabb \psi|^2_{\slashg}
+\psi^2).
\end{align*}
\end{bigprop}

\section{Separation}   
\label{ayrilik}
As described in Section~\ref{TNG}, we will exploit Carter's separation of the wave
equation~\cite{cartersep2} to frequency-localise our virial type estimates
in a manner particularly
suited to the local and global geometry of Kerr.
The separation of $\Box_g\psi=0$ requires taking the Fourier transform in $t$, and
then expanding into what are known as oblate spheroidal harmonics. 
Since, a priori, we do not know
that solutions $\psi(t,\cdot)$ of our problem are $L^2(dt)$, 
we must in fact apply this separation to solutions of the inhomogeneous equation 
\begin{equation}
\label{INHOMOG}
\Box_g\Psi = F,
\end{equation}
where $\Psi$ is related to $\psi$ by the application of a suitable cutoff. We will defer
the discussion of such cutoffs to Section~\ref{cutoffsec}, and will in the meantime
in Section~\ref{carterssep} below
give a general discussion of
solutions of $(\ref{INHOMOG})$ in $r>r_+$ for which one can indeed take
the Fourier transform (see Section~\ref{carterssep}) in $t$. First, however, 
we quickly review the classical oblate spheroidal harmonics.

\subsection{Oblate spheroidal harmonics}
Let $L^2(\sin\theta\, d\theta\, d\phi)$ denote the space of complex-valued $L^2$ functions on the
sphere with standard spherical coordinates $\theta,\phi$.   Let $L^2_m$ denote
the eigenspace of $\partial_\phi$ with eigenvalue $m$.
\index{frequency labels! $m$ (azimuthal wave
number, labels eigen-spaces of $\partial_\phi$)}
Recall that
\[
L^2= \bigoplus_{m\in \mathbb Z} L^2_m
\]
orthogonally,
and that 
\begin{equation}
\label{identifi}
L^2_m= L^2_{\mathbb R}(\sin\theta\, d\theta)\otimes \mathbb Ce^{im\phi}.
\end{equation}

For
$\xi\in \mathbb R$, define the operator $P(\xi)$ on a suitable dense
subset of $L^2(\sin\theta\, d\theta\, d\phi)$ by
\[
P(\xi)\, f= -\frac 1{\sin\theta} \frac{\partial}{\partial\theta} \left (\sin\theta \frac{\partial}{\partial\theta}f\right)
-\frac{\partial^2 f}{\partial\phi^2}\frac{1}{\sin^2\theta}
- \xi^2 \cos^2\theta f.
\]
Note that $P(0)$ is the standard spherical Laplacian, whereas $P(\xi)$
\index{operators! $P(\xi)$ (Laplacian on oblate spheroid with parameter $\xi$)}
 in fact corresponds
to the Laplacian on an oblate spheroid.
We have certainly that $[P(\xi),\partial_\phi]=0$ and thus $P(\xi)$ preserves $L^2_m$,
and under the identification $(\ref{identifi})$, $P(\xi)$ induces
an operator $P_m(\xi)$
\index{operators! $P_m(\xi)$ (restriction of $P(\xi)$ to $L^2_m$)}
 defined on a dense subset of $L^2(\sin\theta\, d\theta)$
by
\[
P_m(\xi)\, f= -\frac 1{\sin\theta} \frac{d}{d\theta} \left (\sin\theta \frac{d}{d\theta}f\right)
+\frac{m^2}{ {\sin^2\theta}}
f- \xi^2 \cos^2\theta f.
\]

The following proposition is classical and can be distilled from standard elliptic theory
and the form of $P_m$:
\begin{proposition}
\label{oblateprop}
Let $\xi\in\mathbb R$. There
exists a complete basis of $L^2(\sin\theta\, d\theta\, d\phi)$
of eigenfunctions of $P(\xi)$ of the form
$S_{m\ell}(\xi,\cos\theta)e^{im\phi}$,
\index{frequency labels! $\ell$ (labels eigenspaces of $P_m$)}
with real eigenvalues $\lambda_{m\ell}(\xi)$:
\[
P(\xi)\, S_{m\ell}(\xi,\cos\theta)e^{im\phi}= \lambda_{m\ell}(\xi)\,
S_{m\ell}(\xi,\cos\theta)e^{im\phi}
\]
where for each $m\in \mathbb Z$, the collection
$S_{m\ell}(\xi,\cos\theta)e^{im\phi}$ 
are a basis in $L^2_m$ of real eigenfunctions of $P_m(\xi)$:
\[
P_m(\xi)\, S_{m\ell}(\xi,\cos\theta) =\lambda_{m\ell}(\xi)\, S_{m\ell}(\xi,\cos\theta)
\]
indexed\index{frequency labels! $\lambda_{m\ell}(\xi)$ (eigenvalues of the operator $P(\xi)$)}
by the set $\ell\ge |m|$. The functions $S_{m\ell}(\xi,\cos\theta)$ are smooth in $\xi$
and $\theta$, and, $\lambda_{m\ell}(\xi)$ is smooth in $\xi$.
Finally, for $\xi=0$, these reduce to spherical harmonics, i.e.
$S_{m\ell}(0, \cos\theta)e^{im\phi}= Y_{m\ell}$, where $Y_{m\ell}$
denote the standard spherical
harmonics, and
\[
\lambda_{m\ell}(0)= \ell(\ell+1).
\]
\end{proposition}
The functions $S_{m\ell}$  are known classically as \emph{oblate spheroidal harmonics}.

Let us introduce the notation
\begin{equation}
\label{intnot}
\alpha^{(\xi)}_{m\ell} =
\int_0^{2\pi} d\varphi  
\int_{-1}^1 d(\cos\theta)\,\alpha(\theta,\phi) \,e^{-im\phi}  S_{m\ell}(\xi,\cos\theta).
\end{equation}

The completeness and orthogonality
conditions of Proposition~\ref{oblateprop} are given explicitly below:
\[
\int_0^{2\pi} d\varphi 
\int_{-1}^1 d(\cos\theta) e^{im\phi} 
S_{m\ell} (a\omega,\cos\theta) \,e^{-im'\phi} S_{m'\ell'}(\xi,\cos\theta)=
\de_{mm'} \de_{\ell\ell'},
\]
\begin{equation}
\label{explicitly}
\alpha(\theta,\phi)= \sum_{m\ell}\alpha^{(\xi)}_{m\ell}S_{m\ell}(\xi,\cos\theta) e^{im\phi}.
\end{equation}
The statement $(\ref{explicitly})$ is to be understood in $L^2(\sin\theta\,d\theta\,d\phi)$, but, if
$\alpha(\theta,\phi)$ is in fact, say, a smooth function on the standard sphere,
then $(\ref{explicitly})$ holds pointwise in the topology of the sphere. (We shall not
require such pointwise statements, however.)

In the application to the Kerr geometry, $\xi=a\omega$ for $\omega\in \mathbb R$,
where $\omega$ will be a frequency associated to the Fourier transform
\index{frequency labels! $\omega$ (frequency associated to Fourier transform in $t$)}
 with respect
to $t$ of a cutoff version
of the solution $\psi$.
We note that the case of complex $\xi$ is also of interest in the formal mode analysis
associated to this problem (see for instance~\cite{whiting}).  In view of our setting, however,
we shall only need to consider real $\xi$ here.

Let us note that from the smooth dependence 
statement of Proposition~\ref{oblateprop} (applied with the application
$\xi=a\omega$ in mind)
one immediately obtains
\begin{proposition}
\label{1stpr}
Given any $\omega_1>0$, $\lambda_1>0$, and $\epsilon>0$, then there exists an $a_1>0$
such that for $|a|<a_1$,
$|\omega|\le \omega_1$, and $\lambda_{ m\ell}\le \lambda_1$, 
\[
|\lambda_{m\ell}(a\omega)- \ell(\ell+1) |\le  \epsilon.
\]
\end{proposition}

Rewriting the equation for the oblate spheroidal function
$$
-\frac 1{\sin\theta} \frac{d}{d\theta} \left (\sin\theta \frac{d}{d\theta}S_{m\ell}\right)+\frac {m^2}{\sin^2\theta}
S_{m\ell} = \lambda_{m\ell} S_{m\ell} + a^2\omega^2 \cos^2\theta S_{m\ell},
$$
the smallest eigenvalue of the operator on the left hand side of the above equation is $|m|(|m|+1)$.
This implies that 
\begin{proposition}
\label{2ndpr}
\begin{equation}
\label{Totherestimate}
\lambda_{m\ell}(a\omega) \ge |m|(|m|+1)-a^2\omega^2.
\end{equation} 
\end{proposition}
The above propositions will be all that we require about $\lambda_{m\ell}$. 
For a more detailed analysis of $\lambda_{m\ell}$, see~\cite{spec}.

Let us finally note that 
by integration by parts and the above orthogonality we obtain
\begin{proposition}
\label{forangs}
\begin{align*}
\int|{}^\gamma\nabla \alpha|^2_\gamma\sin\theta\, d\theta \,d\phi-\int \xi^2\cos^2\theta \,|\alpha|^2\,
 \sin\theta\, d\theta
\, d\phi   
&=
\int (P(\xi)\alpha) \bar\alpha \,\sin\theta\, d\theta\, d\phi\\
& = \sum_{m\ell}\lambda_{m\ell}(\xi) |\alpha^{(\xi)}_{m\ell}|^2
\end{align*}
\end{proposition}
Here, $\gamma$\index{metrics! $\gamma$ (standard unit metric on fixed-$(r,t)$ spheres)} denotes the standard unit metric defined
by standard spherical coordinates $(\theta,\phi)$ on the fixed $(r,t)$ spheres,
and ${}^\gamma\nabla$\index{Lorentzian-geometric concepts! ${}^\gamma\nabla$  (covariant derivative with respect
to standard unit-sphere metric $\gamma$)}
 denotes the covariant derivative
with respect to this metric.

\subsection{Carter's separation}
\label{carterssep}

Let  $\Psi(t,r,\theta,\phi)$, $\Upsilon(t,r,\theta,\phi)$, 
be functions with the property that $\Psi$, $\Upsilon$ 
are smooth and of Schwartz class in $t$, smooth
in $r>r_+$, and smooth on the $(\theta,\phi)$ spheres.

Let $\widehat\Psi(\omega, r, \theta,\phi)$\index{$\Psi$-related! $\widehat\Psi(\omega, r, \theta,\phi)$ (Fourier transform of $\Psi(t,\cdot)$
with respect to $t$)}
 denote the Fourier transform with respect to $t$,
i.e.
\[
\widehat\Psi(\omega,r,\theta,\phi) = \frac{1}{\sqrt{2\pi}}
	 \int_{-\infty}^{\infty}\Psi(t,r,\theta,\phi)\, e^{i\omega t}
dt.
\]
Note\index{frequency labels! $\omega$ (frequency associated to Fourier transform in $t$)} that under our assumptions, $\widehat\Psi$ is smooth and of
Schwartz class in $\omega\in \mathbb R$, smooth in
$r>r_+$, and smooth on the spheres.

Given $a\in \mathbb R$, for fixed $\omega$, $r$,
we may define the coefficients
$\widehat\Psi^{(a\omega)}_{m\ell}(\omega, r)$ by $(\ref{intnot})$, 
choosing
the value $\xi=a\omega$.
Let us agree to drop the ${\ }\widehat{\ }{\ }$ from the notation, and replace the argument
pair $(\omega,r)$ by the singlet
$(r)$, since the $\omega$-dependence is also implicit in the superscript, and denote
these coefficients 
\[
\Psi^{(a\omega)}_{m\ell}(r).
\]
\index{$\Psi$-related! $\Psi^{(a\omega)}_{m\ell}(r)$ (coefficients of $\widehat\Psi$ with respect
to Carter's separation)}

We may now state a suitable version of
Carter's celebrated separation of the wave operator on Kerr.
\begin{theorem}[Carter~\cite{cartersep2}]
\label{septh}
Let $g_{a,M}$ be a Kerr metric for $|a|<M$, let $t,r,\theta,\phi$ be
the Boyer-Lindquist coordinates defined in Section~\ref{BLlc}, and let 
$\Psi(t,r,\theta,\phi)$ be Schwartz class in $t$,
smooth in $r>r_+$, and smooth on the  $(\theta, \phi)$ spheres, with
\[
\Box_g\Psi = F.
\]
Then defining the coefficients 
$\Psi_{m\ell}^{(a\omega)}(r)$, $F_{m\ell}^{(a\omega)}(r)$ as above,
the following holds:
\begin{align}
\label{CartersODE}
\nonumber
\Delta \frac{d}{dr} \left (\Delta \frac{d\Psi_{m\ell}^{(a\omega)}}{dr}\right)& + \left (a^2m^2 + (r^2+a^2)^2\omega^2-\Delta (\lambda_{m\ell}+a^2\omega^2) \right) \Psi_{m\ell}^{(a\omega)}\\
&=(r^2+a^2)\Delta\,
F_{m\ell}^{(a\omega)}.
\end{align}
\end{theorem}

\subsection{Basic properties of the decomposition}
\label{idiotntes}
The coefficients can be related back to $\widehat\Psi$ by
\[
\widehat\Psi(\omega, r,\theta,\phi) = \sum_{m\ell} \Psi^{(a\omega)}_{m\ell}(r)
S_{m\ell}(a\omega,\cos \theta)e^{im\phi}.
\]
The above equality is to be interpreted in $L^2(\sin\theta\,d\theta\,d\phi)$. 
(It in fact also holds pointwise under the smoothness assumptions given here, but
we shall not need to make use of this.)
From Plancherel's formula, and Proposition~\ref{oblateprop},  we  have
\begin{align*}
&\int_0^{2\pi}\int_0^\pi\int_{-\infty}^{\infty} 
 |\Psi|^2(t,r,\theta,\varphi) \sin\theta\, d\varphi\, d\theta\, dt=\int_{-\infty}^\infty \sum_{m,\ell} 
|\Psi^{(a\omega)}_{m\ell}(r)|^2 
d\omega,\\
&\int_0^{2\pi}\int_0^\pi\int_{-\infty}^{\infty} \Psi\cdot\Upsilon \sin\theta d\varphi\, d\theta\, dt=
\int_{-\infty}^\infty \sum_{m,\ell}  \Psi^{(a\omega)}_{m\ell}\cdot\bar \Upsilon^{(a\omega)}_{m\ell}
d\omega.
\end{align*}

With respect to Boyer-Lindquist coordinates,
we clearly have for $r>r_+$
\[
\frac{d}{dr}\Psi^{(a\omega)}_{m\ell}= (\partial_r\Psi)^{(a\omega)}_{m\ell}.
\]
On the other hand, from well-known properties of the Fourier transform,
\[
(\partial_t\Psi)^{(a\omega)}_{m\ell} = -i\omega \Psi^{(a\omega)}_{m\ell}.
\]
In particular, we have
\begin{align*}
&\int_0^{2\pi}\int_0^\pi\int_{-\infty}^{\infty} 
|\partial_r\Psi|^2(t,r,\theta,\varphi) \sin\theta\, d\varphi\, d\theta\, dt=\int_{-\infty}^\infty \sum_{m,\ell} 
\left|\frac{d}{dr}\Psi^{(a\omega)}_{m\ell}(r)\right|^2 
d\omega,
\end{align*}
\begin{align*}
&\int_0^{2\pi}\int_0^\pi\int_{-\infty}^{\infty} 
|\partial_t\Psi|^2(t,r,\theta,\varphi) \sin\theta\, d\varphi\, d\theta\, dt=\int_{-\infty}^\infty \sum_{m,\ell} 
\omega^2 |\Psi^{(a\omega)}_{m\ell}(r)|^2 
d\omega.
\end{align*}

Finally, let us note that from Proposition~\ref{forangs} and the above,
we obtain 
\begin{align*}
\int_{-\infty}^\infty \sum_{m\ell} \lambda_{m\ell}(a\omega) |\Psi^{a\omega}_{m\ell}(r)|^2d\omega
=&
\int_{-\infty}^\infty \int_0^{2\pi}\int_0^\pi |{}^\gamma\nabla \Psi|^2_\gamma
\sin\theta\, d\theta \,
d\phi\, dt\\
&-a^2\int_{-\infty}^{\infty}
 \int_0^{2\pi}\int_0^\pi
|\partial_t\Psi|^2 \cos^2\theta \,\sin\theta\, d\theta\, d\phi\, dt,
\end{align*}
where $\gamma$, ${}^\gamma\nabla$  are defined previously.

\section{Cutoffs}
\label{cutoffsec}
{\bf From now onwards, we shall assume the reduction of Section~\ref{reducsec}.}
We are considering always solutions $\psi$ of $(\ref{WAVE})$ for $g=g_{M,a}$
for parameters to be constrained at each instance.
To apply the results of Section~\ref{ayrilik} to $\psi$, we must first cut off in time.

Let $\xi(z)$ be a smooth cutoff function such that $\xi=0$ for $z\le 0$, $\xi=1$ for $z\ge 1$. 
Let $\varepsilon>0$.\footnote{We introduce the parameter $\varepsilon$ solely
for the purpose of avoiding reference to Theorem~\ref{btheorem2} in proving
Theorem~\ref{nmt}. If we were to allow appeal
to Theorem~\ref{btheorem2}, one could simply take $\varepsilon=1$. This
remark should be considered in the context of the $m=0$ or $a=0$ case, where
Theorem~\ref{btheorem2} is a much more elementary statement.}\index{fixed parameters!
small parameters! $\varepsilon$ (associated with the cutoff)}
Let $\tau'\ge 2\varepsilon^{-1}$ be fixed\index{fixed parameters! large parameters! $\tau'$ ($t^*$-parameter associated
to the cutoff)}, 
we define $\xi_{\tau',\varepsilon}(t^{*})= \xi(\varepsilon t^{*})\xi (\varepsilon(\tau'- t^*))$
and
\[
\psi_{\hbox{\Rightscissors}}(t^*,\cdot)=  \xi_{\tau',\varepsilon}(t^*)
\psi(t^*,\cdot)
\]
with respect to coordinates $(t^{*},r,\theta,\phi)$.\index{$\psi$-related! 
$\psi_{\hbox{\Rightscissors}}$ (cutoff version of $\psi$)}
We note that $\psi_{\hbox{\Rightscissors}}:\mathcal{R}\to \mathbb R$
is smooth and supported in $0\le t^* \le \tau'$.
The function $\psi_{\hbox{\Rightscissors}}$ 
is a solution of the inhomogeneous equation
\[
\Box_g
\psi_{\hbox{\Rightscissors}}
=F,\qquad F=2 \nabla^\a\xi_{\tau',\varepsilon}\, \nabla_\a \psi + (\Box_g\xi_{\tau',\varepsilon})
\, \psi.
\]
Note that $\nabla{\xi}_{\tau',\varepsilon}$ and $F$ are supported in 
\begin{align*}
\{0\le t^* \le \varepsilon^{-1}\}
\cup
\{\tau'-\varepsilon^{-1}\le t^* \le \tau'\}
\end{align*}
and that, with respect to Kerr-star coordinate derivatives:
\[
| \Box_g\xi_{\tau',\varepsilon}|\le 
 \varepsilon^2 B
\]
\[
|\nabla^\alpha\xi_{\tau', \varepsilon}\nabla_\alpha\psi|\le 
B\varepsilon (|\partial_{r}\psi|^2+ |\partial_{t^*}\psi|^2+|\nabb\psi|^2).
\]

Finally, we note that restricted to $r>r_+$, the function $\psi_{\hbox{\Rightscissors}}$ is smooth in Boyer-Lindquist coordinates
and for each fixed $r>r_+$ is compactly supported in $t$. 
In what follows, we may thus apply Theorem~\ref{septh} 
to $\Psi=\psi_{\hbox{\Rightscissors}}$.

\section{The frequency localised multiplier estimates}
We shall construct  in this section 
analogues of the ${\bf J}^{X,w}$ currents
used in~\cite{dr3, dr5}, localised  however
to each $\Psi^{(a\omega)}_{m\ell}$.

\subsection{The separated current templates}
\label{sct}

To describe the analogue of multipliers of the form ${\bf J}^{X,w}$ localised to
frequency triplet
$(\omega,m,\ell)$, it will be convenient to define the following current templates.

First, we
may recast the ode $(\ref{CartersODE})$ (applied to $\Psi=\psi_{\hbox{\Rightscissors}}$)
in a more compact form as follows:
Recall the definition $(\ref{r*def})$ of $r^*$ and set
\[
u^{(a\omega)}_{m\ell}(r)=(r^2+a^2)^{1/2}
 \Psi^{(a\omega)}_{m\ell} (r),\qquad 
 H^{(a\omega)}_{m\ell}(r)=\frac{\Delta F^{(a\omega)}_{m\ell}(r)}{(r^2+a^2)^{1/2}}.
\]
Then $u$ satisfies
\index{$\Psi$-related! $u$ (suitably rescaled version of $\Psi^{(a\omega)}_{m\ell} (r)$
eventually applied to $\Psi=\psi_{\hbox{\Rightscissors}}$)}
\index{$\Psi$-related! $H$ (suitably rescaled version of $ F^{(a\omega)}_{m\ell}(r)$)}
\begin{equation}
\label{e3iswsntouu}
\frac{d^2}{(dr^*)^2}u^{(a\omega)}_{m\ell}+(\omega^2 - V^{(a\omega)}_{m\ell }(r))u^{(a\omega)}_{m\ell} = 
H^{(a\omega)}_{m\ell}
\end{equation}
where
\[
V^{(a\omega)}_{m \ell}(r)= \frac{4Mram\omega-a^2m^2+\Delta (\lambda_{m\ell}+\omega^2a^2)}{(r^2+a^2)^2}
+\frac{\Delta(3r^2-4Mr+a^2)}{(r^2+a^2)^3}
-\frac{3\Delta^2 r^2}{(r^2+a^2)^4}.
\index{fixed functions! spacetime functions! $V$ (potential $V=V^{(a\omega)}_{m \ell}(r^*)$ arising from separation)}
\]
Observe the Schwarzschild case:
\begin{equation}
\label{at0freq}
V^{(0\omega)}_{m\ell}(r) = (r-2M)\left(\frac{\lambda_{m\ell}}{r^3}+\frac{2M}{r^4}\right),
\end{equation}
\begin{equation}
\label{at0freq'}
\left(\frac{dV}{dr^*}\right)^{(0\omega)}_{m\ell}(r)= \frac{r-2M}{r}\left(\frac{2\lambda_{m\ell}(3M-r)}{r^4}+\frac{2M(8M-3r)}{r^5}\right).
\end{equation}

In what follows, let us  suppress the dependence of
$u$, $H$ and $V$ on $a\omega$, $m$, $\ell$ in our notation.\footnote{Before suppressing
the dependence, it might be useful to remark that in fact, the functions
$u$, $H$, $V$, etc., depend on both $a\omega$ and $a$, whereas the
parameters $\lambda_{m\ell}$ depend only on $a\omega$. For the former 3, will view
reference to the $a$-dependence as implicit in the reference to $a$ in $a\omega$. Hence,
our writing $0\omega$ instead of $0$ in $(\ref{at0freq})$.}

Now, with the notation $'=\frac{d}{dr^*}$,
for arbitrary functions $h(r^*)$, $y(r^*)$, let us define the currents
\[
\text {\fontencoding{LGR}\selectfont \koppa}^y[u] = y\left(|u'|^2+(\omega^2-V)|u|^2\right),
\]
\[
\text {\fontencoding{LGR}\selectfont \Coppa}^h[u] = h {\rm Re} (u'\bar u)-\frac12 h' |u|^2.
\]
\index{energy currents! fixed-frequency current templates! $\text {\fontencoding{LGR}\selectfont \koppa}$ (also pronounced `koppa')}
\index{energy currents! fixed-frequency current templates!  $\text {\fontencoding{LGR}\selectfont \Coppa}$ (pronounced `koppa')}
\index{seed functions! $y$ (seed function $y=y(r^*)$ used with template $\text {\fontencoding{LGR}\selectfont \koppa}$)}
\index{seed functions! $h$ (seed function $h=h(r^*)$ used with template $\text {\fontencoding{LGR}\selectfont \Coppa}$)}
We compute:
\begin{equation}\label{eq:Q2for}
(\text {\fontencoding{LGR}\selectfont \koppa}^y[u])'= y ' \left(|u'|^2+(\omega^2-V)|u|^2\right) -yV'|u|^2+ 2y\,{\rm Re} (u'\bar H),
\end{equation}
\begin{equation}\label{eq:Q1for}
(\text {\fontencoding{LGR}\selectfont \Coppa}^h[u])'=h\left(|u'|^2 +(V-\omega^2)|u|^2\right) -\frac12 h'' |u|^2 +h\, {\rm Re} (u\bar H).
\end{equation}

The currents we shall employ will be various combinations  of
$\text {\fontencoding{LGR}\selectfont \Coppa}$, $\text {\fontencoding{LGR}\selectfont \koppa}$
with 
suitably selected functions $h$, $y$.
Note that the choice of these functions may
depend on $a$, $a\omega$, $m$, $\ell$, but, again, we temporarily
suppress this from the notation. 

A particular combination of $\text {\fontencoding{LGR}\selectfont \Coppa}$ and
$\text {\fontencoding{LGR}\selectfont \koppa}$ shall appear so often that we shall give
it a new name; consideration of this current is
 motivated by removing the $\omega^2$ term
from the right hand side of $(\ref{eq:Q2for})$, $(\ref{eq:Q1for})$.
Let us define
$$
{\text{Q}}^f[u]=\text {\fontencoding{LGR}\selectfont \Coppa}^{f'}[u]+
\text {\fontencoding{LGR}\selectfont \koppa}^{f}[u]=
f \left ( |u'|^2 + (\omega^2-V ) |u|^2\right) + f'
{\rm Re}\left(u'\bar u\right)-
\frac 12f'' |u|^2.
$$
We see immediately
\begin{equation}\label{eq:Qfor}
(\text{Q}^f[u])'= 2 f' |u'|^2  - f V' |u|^2 + {\rm Re}(2 f \bar{H} u'+f' \bar{H} u)-\frac 12 f''' |u|^2.
\end{equation}
\index{energy currents! fixed-frequency current templates!  $\text{Q}$ (${\text{Q}}^f[u]=\text {\fontencoding{LGR}\selectfont \Coppa}^{f'}[u]+
\text {\fontencoding{LGR}\selectfont \koppa}^{f}[u])$}
\index{seed functions! $f$  (seed function $f=f(r^*)$ used with template $\text {Q}$)}

A word of warning: For fixed $g_{M,a}$, we 
shall often refer to $r^*$-ranges by their corresponding
$r$ ranges, and functions appearing in most estimates will be written in terms of $r$. 
Moreover, given an $r$-parameter such as $R$, then $R^*$ will denote $r^*(R)$.
It is important to remember at all times that $'$ always means $\frac{d}{dr^*}$!

\subsection{The frequency ranges}
\label{freerange}
Let $\omega_1$\index{fixed parameters! small parameters! $\omega_1$ (frequency parameter)}, 
$\lambda_1$\index{fixed parameters! large parameters! $\lambda_1$ (frequency parameter)}
be (potentially large) parameters
to be determined, and $\lambda_2$\index{fixed parameters! small parameters! $\lambda_2$ (frequency parameter)}
be a (potentially small) parameter to be determined.
We define the frequency ranges
$\mathcal{F}_{\mbox{$\flat$}}$\index{frequency ranges! $\mathcal{F}_{\mbox{$\flat$}}$ (frequency triplet range,
`bounded' frequencies)}, 
$\mathcal{F}_{\lessflat}$\index{frequency ranges! $\mathcal{F}_{\lessflat}$ (frequency triplet range,
`angular-dominated' frequencies)},
$\mathcal{F}_{\mbox{$\natural$}}$\index{frequency ranges! $\mathcal{F}_{\mbox{$\natural$}}$ (frequency triplet range, `trapped' frequencies)},
 $\mathcal{F}_{\mbox{$\sharp$}}$\index{frequency ranges! $\mathcal{F}_{\mbox{$\natural$}}$ (frequency triplet range, `time-dominated' frequencies)} 
 by
\begin{itemize}
\item
$\mathcal{F}_{\mbox{$\flat$}}=\{(\omega, m, \ell)$ :
$|\omega|\le \omega_1$, $\lambda_{m\ell}(a\omega) \le \lambda_1\}$
\item
$\mathcal{F}_{\lessflat}=\{(\omega, m, \ell)$ :
$|\omega|\le \omega_1$, $\lambda_{m\ell}(a\omega) >\lambda_1\}$
\item
$\mathcal{F}_{\mbox{$\natural$}}=\{(\omega, m, \ell)$ :
$|\omega|\ge \omega_1$, $\lambda_{m\ell}(a\omega) \ge \lambda_2\omega^2\}$
\item
$\mathcal{F}_{\mbox{$\sharp$}}=\{(\omega, m, \ell)$ :
$|\omega|\ge \omega_1$, $\lambda_{m\ell}(a\omega) <\lambda_2\omega^2\}$.
\end{itemize}
The nature of our constructions will be of quite different philosophy for
each of the above ranges.

\subsection{The $\mathcal{F}_{\mbox{$\flat$}}$ range (bounded frequencies)}
\label{elow}
This is a compact frequency range and, in view of the fact that
we have already constructed multipliers in the Schwarzschild case (see~\cite{dr3, dr5}),
by stability considerations, their positivity properties carry over for $|a|\ll M$. See our
proof in~\cite{jnotes}.

To give a completely self-contained presentation in this paper which allows
one to treat Theorem~\ref{nmt} and Theorem~\ref{nmt2} together, 
we will present here 
a different construction with a greater range of validity. 
These constructions will also be useful for our
companion paper~\cite{drf1} where one must further decompose
into superradiant and nonsuperradiant frequencies.

First some preliminaries:
Let us define 
\[
V_{\rm new}=\frac{4Mram\omega-a^2m^2}{(r^2+a^2)^2},
\]
\begin{eqnarray}
\label{plumi}
\nonumber
V_+ &=& V-V_{\rm new}\\
\nonumber
	&=& \frac{\Delta}{(r^2+a^2)^{4}}
\left( (\lambda_{m\ell}+\omega^2a^2)(r^2+a^2)^2+ (2Mr^3+a^2r^2+a^4-4Mra^2) \right).
\\
\end{eqnarray}

Remark that by Proposition~\ref{2ndpr}, the first term of $(\ref{plumi})$ is
nonnegative, whereas the second term is easily seen to be strictly positive
for $r>r_+$ and all $|a|<M$.
Thus, we have in particular
\[
V_+> 0
\]
for $r>r_+$.
We have, moreover, according to our conventions, for all $|a|\le a_0$,
\[
B(\Delta/r^2) (\lambda_{m\ell}+a^2\omega^2) r^{-2}  +B r^{-3}\ge
V_+ 
\ge (\Delta/r^2) b (\lambda_{m\ell}+a^2\omega^2)r^{-2}  +b(\Delta/r^2) r^{-3}.
\] 

In the case of Theorem~\ref{nmt}, we may now fix an
arbitrary $r_0> r_{c}>2M$\index{fixed parameters! $r$-parameters! $r_c$ (near-horizon parameter)}, where $r_0$ is as in Proposition~\ref{rshere},
and it follows that for sufficiently small $a_0$ depending in particular on this choice,
we have
for
 $|a|\le a_0$ and $(\omega,m,\ell)\in \mathcal{F}_{\mbox{$\flat$}}$ 
 the following inequality in the region $r\ge r_c$:
\begin{equation}
\label{Vpos1}
B(\Delta/r^2) (\lambda_{m\ell}+a^2\omega^2) r^{-2}  +B r^{-3}
\ge 
V 
\ge
(\Delta/r^2) b (\lambda_{m\ell}+a^2\omega^2)r^{-2}  +b(\Delta/r^2) r^{-3},
\end{equation}
whereas, for all $r>r_+$, we have
\begin{equation}
\label{Vpos2}
|V'(r^*)|\le 
B(\Delta/r^2) (\lambda_{m\ell}+a^2\omega^2) r^{-3}  +B(\Delta/r^2)r^{-4}.
\end{equation}
We are using here Proposition~\ref{2ndpr}.
Finally, again choosing $a_0$ sufficiently small, 
it follows that for all $r\le r_{c}$ 
\begin{equation}
\label{andforlarger0}
V'(r^*)\ge b(\Delta/r^2) (\lambda_{m\ell}+a^2\omega^2),
\end{equation}
whereas there exists a constant $R_6$ such that for $r\ge R_6$,
\begin{equation}
\label{andforlarger}
-V'(r^*) \ge (\Delta/r^2) b (\lambda_{m\ell}+a^2\omega^2)r^{-3}  +b(\Delta/r^2) r^{-4}
\end{equation}
in this frequency range. 

In the general case $a_0<M$, $|a|\le a_0$, but $m=0$ (Theorem~\ref{nmt2}), 
the above four inequalities similarly
hold for  $(\omega,m,\ell)\in \mathcal{F}_{\mbox{$\flat$}}$, with $(\ref{Vpos1})$
holding in fact for all $r>r_+$, and with $(\ref{andforlarger0})$ holding
in a region $r\le r_c$, where $r_c=r_++s_c$ for some $s_c$ depending
only on $a_0$, with $r_c\to 0$ as $a_0\to M$.
Let the choice of $s_c$ and thus $r_c(a)$ be now fixed.

We will now
split the frequency range $\mathcal{F}_{\mbox{$\flat$}}$ into two subcases,
considering each separately.

\subsubsection{The subrange $|\omega|\le \omega_3$ (the near-stationary subcase)}
\label{oneofthesub}
The motivation for the current to be constructed here
is that in the Schwarzschild or $m=0$ case,
applying $\text {\fontencoding{LGR}\selectfont \Coppa}^h$ with $h=1$
immediately excludes nontrivial stationary solutions $\omega=0$.

We will fix an $\omega_3>0$\index{fixed parameters! small parameters! $\omega_3$ (frequency parameter)}
which will be constrained in this subsection
to be small. Because $\omega_3$ is not exactly $0$ the naive current $\text {\fontencoding{LGR}\selectfont \Coppa}^h$ with $h=1$ must
be modified.

We begin with a $\text {\fontencoding{LGR}\selectfont \Coppa}$ 
current  which will be defined with non-constant seed function $h=h(r^*)$.

In the case of Theorem~\ref{nmt},
the function $h$ will be independent of the parameters $\omega$, $m$, $\ell$,
but will depend on $a$.\footnote{This latter dependence could easily be removed
by altering the construction slightly. It arises for instance
because the values of $r^*_c$ depends on $a$, even though $r_c$ does not.}
In the case of Theorem~\ref{nmt2},
the choice of $\omega_3$ will 
will potentially depend on $a_0$.
Always remember that in this latter case, constants $b$, $B$ will in general depend also
on $a_0$, following our conventions.

Note first that  given arbitrary $a_0<M$ and $q>0$\index{fixed parameters! 
small parameters! $q$ (associated
with $\mathcal{F}_{\mbox{$\flat$}}$)}, $p>0$\index{fixed parameters! small
parameters! $p$ (associated
with $\mathcal{F}_{\mbox{$\flat$}}$)}, $R_{3}>0$\index{fixed parameters!
large parameters! $R_3$ (large parameter associated
with $\mathcal{F}_{\mbox{$\flat$}}$)}
such that $e^{-p^{-1}}R_3$ is sufficiently large
and $p$ sufficiently small, for each $|a|\le a_0$,
we can define a function $h(r^*)$, such that the following hold:
For
$r\le r_{c}$,
\[
0\le h\le R_3^{-2}, \qquad |h''(r^*)|\le q/|r^*|^2,
\]
with $h(r^*)$ moreover\index{seed functions! $h$ (seed function $h=h(r^*)$ used with template $\text {\fontencoding{LGR}\selectfont \Coppa}$)}
of compact support when restricted to $-\infty<r^*\le r_{c}^*$,
whereas
for $r_{c}\le r\le e^{-p^{-1}}R_3$, 
\[
h\ge  R_3^{-2} \Delta/r^2, \qquad h''(r^*)\le 0,
\]
whereas
for $e^{-p^{-1}} R_3\le r \le R_3$,
\[
|h'(r^*)| \le 4R_{3}^{-2} p /r,
\qquad
|h''(r^*)| \le 4R_{3}^{-2} p/r^2,
\]
whereas, finally, for $r\ge R_{3}$,
\[
h=0.
\]
This $h$ will be useful in view of the positivity of $V_+$. 
We note finally that $h$ can be chosen so that $h$ restricted to $r\ge r_{c}$
is independent of the parameter $q$.

Let us also consider a current $\text {\fontencoding{LGR}\selectfont \koppa}$ 
defined with a $y=y(r^*)$. Like $h$, the function $y$
will be independent of $\omega$, $m$, and $\ell$ in the allowed range. 
Given the parameters $R_{3}$, $p$, with $e^{-p^{-1}}R_{3}$ sufficiently
large and $p$ sufficiently small, the function will satisfy the following properties:
We set  $y(r_{c}^*)=0$,
and for $r\le r_{c}$,\index{seed functions! $y$ (seed function $y=y(r^*)$ used with template $\text {\fontencoding{LGR}\selectfont \koppa}$)}
\[
y'(r^*)=h,
\]
noting that $y$ is bounded below in $r\le r_c$ by a potentially large negative constant
depending on $q$, in view of the fact that $h$ is identically
$0$ for sufficiently low $r^*$. 
For $r_{c}\le r\le e^{-p^{-1}}R_{3}$,   we require
\[
y'(r^*)\ge 0,\qquad 
(yV)'(r^*)\le \frac12 R_{3}^{-2} \Delta/r^2 ,
\]
whereas for $e^{-p^{-1}} R_{3}\le  r\le R_{3}$
\[
y'(r^*)\ge bR^{-2}_{3},\qquad -(yV)' (r^*)\ge bR^{-2}_{3}/r^2,
\]
whereas finally,
for  $r\ge R_{3}$,
\[
y=1.
\]  	
In general, for sufficiently large $e^{-p^{-1}}R_{3}$ 
and small $p$, we can indeed construct such a $y$
in view essentially of $(\ref{Vpos1})$, $(\ref{Vpos2})$ and $(\ref{andforlarger})$.
Let us add that $y$ restricted to $r\ge r_{c}$ can be chosen independently
of the choice of $h$.

Consider now the current $\text {\fontencoding{LGR}\selectfont \Coppa}^h+
\text {\fontencoding{LGR}\selectfont \koppa}^y$. 
In $r\le r_{c}$, recall the one-sided bound
\[
-yV' \ge 0
\]
which follows from $(\ref{andforlarger0})$.
It follows that in $r\le r_{c}$,
\[
\text {\fontencoding{LGR}\selectfont \Coppa}'(r^*)+\text {\fontencoding{LGR}\selectfont \koppa}'(r^*)
 \ge 
-q|r^{*}|^{-2}|u|^2
+h{\rm Re}(u\bar H)+2y{\rm Re}(u'\bar H).
\]

In the case of Theorem~\ref{nmt},
for $r_{c} \le r \le e^{-p^{-1}} R_{3}$, choosing $p$ appropriately, choosing $\omega_3$, $a_0$ sufficiently small
\[
\text {\fontencoding{LGR}\selectfont \Coppa}'+
\text {\fontencoding{LGR}\selectfont \koppa}'
 \ge R_{3}^{-2} (\Delta/r^2) |u'|^2 + b\, V_+ |u|^2   +h{\rm Re}(u\bar H) +2y{\rm Re}(u'\bar H).
\]
In the case of Theorem~\ref{nmt2}, given arbitrary $a_0<M$, one can again choose
$\omega_3$ so that the above holds for all $|a| \le a_0$.

Now in the case of both theorems,
for $e^{-p^{-1}}R_3\le r \le R_3$, we have for $p$, $\omega_3$ suitably small 
\[
\text {\fontencoding{LGR}\selectfont \Coppa}'
+\text {\fontencoding{LGR}\selectfont \koppa}'
\ge bR_{3}^{-2} (\Delta/r^2)  |u'|^2  +h{\rm Re}(u\bar H) +2y{\rm Re}(u'\bar H).
\]
Finally, for $r\ge R_{3}$ we have
\[
\text {\fontencoding{LGR}\selectfont \Coppa}'
+\text {\fontencoding{LGR}\selectfont \koppa}'
=\text {\fontencoding{LGR}\selectfont \koppa}' \ge 2y{\rm Re}(u'\bar H).
\]

We will choose in particular $R_3\ge R$.
We obtain finally for $r_{\infty}\ge R_3$, $r_{-\infty}\le r_{c}$,  
\begin{align*}
b(\lambda_1) \int_{r^*_{c}}^{R^*}
\frac12 &(\Delta/r^2) |u'|^2  + (\Delta/r^2)r^{-3}(1+(\lambda_{m\ell}+a^2\omega^2)
+\omega^2)|u|^2    \, dr^*\\
\le&
\int_{r^*_{-\infty}}^{r_{c}^*}
q|r^{*}|^{-2}|u|^2\, dr^*\\
 &+\int_{r^*_{-\infty}}^{r^*_{\infty}}\left(2y{\rm Re}(u'\bar H)+h{\rm Re}(\bar Hu) \right)\,dr^* \\
&+\text {\fontencoding{LGR}\selectfont \koppa}({r^*_{\infty}})-
(\text {\fontencoding{LGR}\selectfont \Koppa}+\text {\fontencoding{LGR}\selectfont \koppa})
({r^*_{-\infty}}).
\end{align*}

The above will be the prototype of the type of
inequality we shall derive for all frequencies.
We note that the bad dependence of the constant $b(\lambda_1)$
as $\lambda_1\to\infty$
arises only from the bound
$(\lambda_{m\ell}+a^2\omega^2)\le \lambda_1+a^2\omega_3^2$
which we have used to introduce the term $(\lambda_{m\ell}+a^2\omega^2)$.
We choose this particular combination because it is 
manifestly nonnegative  (see Proposition~\ref{2ndpr})  and will in particular
bound the angular derivatives upon summation. 

The constants $\omega_3$, $R_3$ and $p$ are now chosen, but not yet $q$.
In accordance with our conventions, the constant $b$ in the inequality
above is in particular
independent of $q$. 
It is worth warning, however, that $y$ restricted to $r\le r_{c}$ still depends
on the choice of $q$, and that the boundary term 
$\lim_{r^*_{\infty}\to-\infty} \text {\fontencoding{LGR}\selectfont \koppa}(r^*_{-\infty})<\infty$,
but diverges as $q\to 0$.
We must thus be careful to absorb this boundary term properly.

\subsubsection{The subrange $|\omega|\ge \omega_3$ (the non-stationary subcase)}
The construction of this section will yield a positive current
for $\omega_3$  arbitrarily small and
$\omega_1$, $\lambda_1$
arbitrarily large, if, in the case of Theorem~\ref{nmt}, 
we are prepared to restrict
to $|a|\le a_0$, with $a_0$ depending on these choices. 
The choice of $\omega_3$ has in fact already been determined but
$\omega_1$, $\lambda_1$ are determined later. The
relevant constants arising grow
as $\omega_3\to0$, $\omega_1\to \infty$, $\lambda_1\to \infty$ and
we shall for now continue to track these dependences.

Let us note that in fact, the
 construction of this section is quite general and depends only on the
asymptotic properties of $V$ (modulo `small' terms), 
and not, in particular, on the sign of $V'$,
which will be crucial in Section~\ref{kloubi}.
We note that similar constructions have a long tradition in spectral theory
and are typically used to prove continuity of the spectrum away from $\omega=0$.

First, the idea: We search for a $\text {\fontencoding{LGR}\selectfont \koppa}$-current
with seed function $y$.
Recall that, dropping the terms arising from the cutoff, we have
\begin{eqnarray*}
\text {\fontencoding{LGR}\selectfont \koppa}' (r^*)&=& y'|u'|^2+\omega^2 y'|u|^2 - (yV)' |u|^2. 
\end{eqnarray*}
In both the Schwarzschild case and the $m=0$ case
we have that $V>0$. If we restrict to $y$ with $y'> 0$ such that
$y$ is moreover bounded at the ends,
we have in these cases
\[
\int_{-\infty}^{\infty} (yV)' |u|^2 = -\int_{-\infty}^{\infty} yV(u \bar{u}'+u'\bar{u})  \le
 \frac12\int_{-\infty}^{\infty} y' |u'|^2  + 2\int_{-\infty}^{\infty} (y^2V^2/y')|u|^2.
\]
Thus, to control the $(yV)' |u|^2$ term from the $y'|u|^2$ term,
it suffices if 
\[
4y^2V^2/y' \le \omega_3^2 y'.
\]
If we assume in addition $y>0$, we may rewrite this condition as
\[
y'/y\ge 2\omega_3^{-1} V
\]
which is ensured if we define, say, 
$y= e^{2\omega_3^{-1} \int_{-\infty}^{r^*} V dr^*}$.
Note that in these two special cases,
$y$ is well-defined and bounded in view of the asymptotics of $V$ as $r^*\to\pm \infty$.

We adapt the above heuristic to our situation. We shall argue somewhat differently, however. 
First of all, we do not want to choose
$y$ to depend on $\omega$, and secondly, for technical reasons related
to how we sum the terms arising from the inhomogeneity $F$ (see Section~\ref{erfromcut}), 
the above choice of $y$ would not be sufficiently flat at infinity.

Given arbitrary $a_0<M$ in the case $m=0$ (Theorem~\ref{nmt2}) or sufficiently
small $a_0$ in the case of Theorem~\ref{nmt}, and given in addition
an arbitrary sufficiently small parameter $\epsilon_2>0$\index{fixed parameters!
small parameters! $\epsilon_2$
(associated with $\mathcal{F}_{\mbox{$\flat$}}$)}
and an arbitrary
$R_2\ge R$\index{fixed parameters! $r$-parameters! $R_2$
(large parameter associated with $\mathcal{F}_{\mbox{$\flat$}}$)}, 
using the properties $(\ref{Vpos1})$, $(\ref{Vpos2})$
and $(\ref{andforlarger})$, it follows that for all $|a|\le a_0$,
we may decompose $V$ yet again 
as 
\begin{equation}
\label{wedeco}
V=V_{+,{\rm flat}}+V_{\rm junk}
\end{equation}
\index{fixed functions! spacetime functions! $V_{+,{\rm flat}}$ (satisfies $V=V_{+,{\rm flat}}+V_{\rm junk}$)}
\index{fixed functions! spacetime functions! $V_{\rm junk}$ (satisfies $V=V_{+,{\rm flat}}+V_{\rm junk}$)}
where 
\[
V_{+,{\rm flat}}\ge0
\]
for all $r\ge r_+$,
and specifically,
\[
\qquad V_{+,{\rm flat}}=0, \qquad V'_{\rm junk}<0
\]
for $r\ge \epsilon_2^{-1}R_2$, whereas
\[
V_{\rm junk}=0
\]
for $R\le r\le R_2$,
whereas 
\begin{equation}
\label{evawhereas}
b (\Delta/r^2) r^{-3}
\le V_{+,{\rm flat}} \le B(\lambda_1+a_0^2\omega_1^2) (\Delta/r^2) r^{-2}
\end{equation}
\[
|V_{\rm junk}|\le Ba_0(\lambda_1+a_0^2\omega_1^2)r^{-3}, \qquad
|V_{\rm junk}'|\le Ba_0 (\lambda_1+a_0^2\omega_1^2)(\Delta/r^2) r^{-4}
\]
for $r\le R_2$, and
whereas finally,
\[
|V_{\rm junk}|\le B(\lambda_1+a_0^2\omega_1^2) r^{-2}, \qquad
|V_{\rm junk}'|\le B\epsilon_2(\lambda_1+a_0^2\omega_1^2) r^{-2} 
\]
for $r\ge R_2$.
Note that the decomposition $(\ref{wedeco})$ in the region $r\le R$ can be
chosen independently of the parameters $\epsilon_2$, $R_2$.

In the case of Theorem~\ref{nmt2}, we may further  impose that in fact
\begin{equation}
\label{othertheocase}
V_{\rm junk}=0
\end{equation}
for $r\le R_2$.

Finally we may define
$V_{\rm ind}(r^*)\ge 0$\index{fixed functions! spacetime functions! $V_{\rm ind}$ (independent of $a\omega$, $m$, $\ell$,
satisfies $V_{\rm ind}\ge V_{+,{\rm flat}}$)} 
to be independent of $\omega, m, \ell$ in this range
such that 
\[
B(\lambda_1+a_0^2\omega_1)(\Delta/r^2)r^{-2} \ge V_{\rm ind}\ge V_{+,{\rm flat}}\ge 0
\]
in $r\le \epsilon_2^{-1}R_2$, 
whereas
\[
V_{\rm ind}=0
\] 
for $r\ge \epsilon_2^{-1}R_2$.
Let us also note that $V_{\rm ind}$ can be chosen
independent of the parameters $\epsilon_2$, $R_2$ in the region $r\le R$
and that
\[
\int_{r^*}^{\infty} V_{\rm ind} dr^*
\ge b(\lambda_1+a_0^2\omega_1^2)r^{-1}
\]
for $r^*\le R^*$, 
where we are using $(\ref{evawhereas})$ 
and the properties of the decompositions.
In particular, there  is no dependence on $\epsilon_2$, $R_2$, $\omega_3$
in the
above inequality.
We have on the other hand for all $r^*>-\infty$,
\[
\int_{r^*}^{\infty} V_{\rm ind} dr^*
\le B(\lambda_1+a_0^2\omega_1^2)
\]
and $V_{\rm ind}$ can be chosen such that there is no dependence of $B$ on 
$\epsilon_2$, $R_2$.

Using $(\ref{wedeco})$ we may now write
\begin{eqnarray}
\label{2inegreat}
\nonumber
\text {\fontencoding{LGR}\selectfont \koppa}' &=& y'|u'|^2+\omega^2 y'|u|^2 - (yV)' |u|^2 +2y{\rm Re}(u'\bar H) \\
	&=& y'|u'|^2+\omega^2 y'|u|^2 - (yV_{+,{\rm flat}})'|u|^2- y'V_{\rm junk} - yV_{\rm junk}'
	+2y{\rm Re}(u'\bar H).
\end{eqnarray}
Note that for general $y\ge 0$, $y'>0$, we have
\begin{eqnarray*}
\int_{r_{-\infty}^*}^{r_\infty^*} (yV_{+,{\rm flat}})' |u|^2 &=& -\int_{r_{-\infty}^*}^{r_\infty^*} 
yV_{+,{\rm flat}}( u \bar{u}' +u'\bar{u})
 + yV_{+,{\rm flat}}|u|^2(r_{\infty}^*)-yV_{+,{\rm flat}}|u|^2(r_{-\infty}^*) \\
&\le&
 \frac12\int_{r_{-\infty}^*}^{r_\infty^*} y' |u'|^2  + 
 2\int_{r_{-\infty}^*}^{r_\infty^*} (y^2V_{+,{\rm flat}}^2/y')|u|^2\\
 &&\hbox{}+ yV_{+,{\rm flat}}|u|^2(r_{\infty}^*)-yV_{+,{\rm flat}}|u|^2(r_{-\infty}^*) \\
&\le&
\frac12\int_{r_{-\infty}^*}^{r_\infty^*} y' |u'|^2  + 2\int_{r_{-\infty}^*}^{r_\infty^*} (y^2V_{\rm ind}^2/y')
|u|^2
\\
 &&\hbox{}+ yV_{\rm ind}|u|^2(r_{\infty}^*).
\end{eqnarray*}

We now define\index{seed functions! $y$ (seed function $y=y(r^*)$ used with template $\text {\fontencoding{LGR}\selectfont \koppa}$)} 
\begin{equation}
\label{hereydef}
y= e^{-2\omega_3^{-1} \int_{r^*}^{\infty} V_{\rm ind} dr^*}.
\end{equation}
Note that in this case, we have that
\[
y'=
2\omega_{3}^{-1}V_{\rm ind}e^{-2\omega_3^{-1}\int_{r^*}^{\infty} V_{\rm ind} dr^*}.
\]
We thus have for this choice of $y$ that
\begin{eqnarray*}
\int_{r_{-\infty}^*}^{r_\infty^*} (yV_{+,{\rm flat}})' |u|^2
&\le&
\frac12\int_{r_{-\infty}^*}^{r_\infty^*} y' |u'|^2  + \frac12
\int_{r_{-\infty}^*}^{r_\infty^*} \omega_3^2 y' |u|^2
\\
 &&\hbox{}+ yV_{\rm ind}|u|^2(r_{\infty}^*).
\end{eqnarray*}

We certainly have
\[
0\le y\le 1, \qquad y'\ge 0
\]
for all $r>r_+$.
In $r_{c}\le r\le R$, we have 
\[
y'(r^*)\ge b
(\Delta/r^2) r^{-3} e^{-2\omega_3^{-1}br^{-1}(\lambda_1+a_0^2\omega_1^2)},
\]
whereas in general for $r\le R$ we have
\[
y'(r^*)\le B\omega_3^{-1}(\lambda_1+a_0^2\omega_1^2)(\Delta/r^2) r^{-2}.
\]
For $r\ge \epsilon^{-1}_2R_2$, we have 
\[
y=1, \qquad y'=0.
\]
Finally, in $R_2\le r\le \epsilon_2^{-1}R_2$, we have 
\[
0\le y'(r^*) \le B\omega_3^{-1}
(\lambda_1+a_0^2\omega_1^2) r^{-2}.
\]

Putting everything together, 
integrating $(\ref{2inegreat})$, and in the case of Theorem~\ref{nmt},
restricting to sufficiently small $a_0$ (depending on the final choices 
of $\omega_1$, $\lambda_1$ and $\omega_3$), 
we obtain finally for $r_{\infty}\ge \epsilon_2^{-1}R_2$, $r_{-\infty}\le r_{c}$,  
the inequality:
\begin{align*}
b(\omega_3, \omega_1, \lambda_1) \int_{r^*_{c}}^{R^*}
&(\Delta/r^2) |u'|^2  + (\Delta/r^2)(1+(\lambda_{m\ell}+a^2\omega^2)+\omega^2)|u|^2\\
\le&
\int_{r^*_{-\infty}}^{r_{c}^*}
Ba_0
(\Delta/r^2) |u|^2\\
&+\int_{R^*_2}^{(\epsilon_{2}^{-1} R_2)^*} 
B(\epsilon_2 r^{-2}+r^{-3})\omega^2 |u|^2
  \\
 &+\int_{r^*_{-\infty}}^{r^*_{\infty}}2y{\rm Re}(u'\bar H)\\
&+\text {\fontencoding{LGR}\selectfont \koppa}({r^*_{\infty}})+ yV_{\rm ind}|u|^2(r_{\infty}^*)
-\text {\fontencoding{LGR}\selectfont \koppa}({r^*_{-\infty}}).
\end{align*}

As before, the integrals are with respect to $dr^*$.
Note that the first term on the right hand side above is absent in the
case of Theorem~\ref{nmt2}, in view of $(\ref{othertheocase})$. Note also that the final
appeal to 
small $a_0$ in the case of Theorem~\ref{nmt} here was to absorb an additional term 
$\int_{r^*_{c}}^{R^*} Ba_0(\Delta/r^2)(\lambda_1+a_0\omega_1^2)(1+
\omega_3^{-1}(\lambda_1+a_0\omega_1^2))r^{-3}|u|^2$
into the term on the left hand side.

\subsection{The $\mathcal{F}_{\lessflat}$ range (angular dominated frequencies)}

The considerations of this section will constrain
$\lambda_1$ to be suitably large, depending on the choice of $\omega_1$.

In the case of Theorem~\ref{nmt},
for sufficiently small $a_0$, then there exist constants $c_1, c_2>0$\index{fixed parameters!
small parameters! $c_i$ (associated with  $\mathcal{F}_{\lessflat}$, $i=1,2$)},
constants $r_{3M-}<3M<r_{3M+}$,\index{fixed parameters! $r$-parameters! $r_{3M-}$ (associated with 
 $\mathcal{F}_{\lessflat}$)}
 \index{fixed parameters! $r$-parameters! $r_{3M+}$ (parameter associated with 
 $\mathcal{F}_{\lessflat}$)}
a constant $R_6\ge R$
and a function $f=f(r^*)$ satisfying\index{seed functions! $f$  (seed function $f=f(r^*)$ used with template $\text {Q}$)}
\[
f'(r^*)\ge 0
\]
for $r>r_+$,
\[
f'(r^*) \ge c_1
\]
for $r_{c}\le r\le R_6$ 
\index{fixed parameters! $r$-parameters!$R_6$ (large parameter associated with 
 $\mathcal{F}_{\lessflat}$)}
\[
f= 1
\]
for $r\ge R_6+1$,
and near the horizon
\[
f\sim -1+(r-r_+),
\]
and
such that for all $|a| \le a_0$ 
and all $(\omega,m,\ell)\in\mathcal{F}_{\lessflat}$, we have
\[
fV'\le c_2V
\]
for $r_{3M-}\le r \le r_{3M+}$, whereas
\[
fV'<0
\]
for $r\le r_{3M-}$, $r\ge r_{3M+}$.
Moreover, the constants $r_{3M-}<3M<r_{3M+}$ can be chosen arbitrarily close
to $3M$ as $a_0\to 0$.

In the case of Theorem~\ref{nmt2}, recalling that
 $V>0$ and the asymptotics of $V'$,
there similarly exist such $f$, $c_1$, $c_2$, $r_{3M-}$, $r_{3M+}$.

Given $a_0<M$,
we choose constants
\[
-\infty<r_{mp1}^*<r_{mp2}^*<r_{3M-}^*<r_{3M+}^*<r_{at1}^*<r_{at3}^*<\infty,
\]
such\index{fixed parameters! $r$-parameters! $r_{mpi}$ (associated with  $\mathcal{F}_{\lessflat}$, $i=1,2$)}\index{fixed parameters! $r$-parameters! $r_{ati}$ (assoicated wtih  $\mathcal{F}_{\lessflat}$, $i=1,2$)} 
that $f\le -\frac12$ for $r\le r_{mp2}$, $f\ge \frac12$ for $r\ge r_{at1}$,
and we fix a function $h$
such that $h=0$ in $r\le r_{mp_1}$, $h=c_2$ in $r_{mp_2}<r_{at1}$,
$h=0$ in $r\ge r_{at3}$. 

Now we may choose $\lambda_1$ depending on $\omega_1$
sufficiently large so that  in $r\le r_{mp2}$, $r\ge r_{at1}$
\[
-\frac12 fV' \ge -\frac12 f''' -\frac12h''
\]
in the $\mathcal{F}_{\lessflat}$ range.
Note that this is possible in view of the asymptotics of $f$ and $V$ and the choice of
large $\lambda_1$, no matter what
the details of the choices of $f$, $h$.

It now follows by construction
that the current $\text{Q}+\text {\fontencoding{LGR}\selectfont \Koppa}$ satisfies
\[
\text{Q}'+\text {\fontencoding{LGR}\selectfont \Koppa}' \ge    2f' |u'|^2 + b(\Delta/r^2)
r^{-3}( \lambda_{m\ell}+a^2\omega^2) |u|^2 +{\rm Re}
(2f \bar{H} u'+f'\bar{H}u) + h {\rm Re} (u\bar{H}).
\]
We obtain thus in particular (for $r^*_{-\infty}<r_{mp_1}$, $r^*_{\infty}\ge  R^*$)
\begin{align*}
 \int_{r^*_{c}}^{R^*}
 &b |u'|^2  + b(\Delta/r^2) r^{-3}(1+(\lambda_{m\ell}+a^2\omega^2)+\omega^2) |u|^2\\
\le
 &\int_{r^*_{-\infty}}^{r^*_{\infty}}2f{\rm Re}(u'\bar H)+(f'+h){\rm Re}(\bar Hu)\\
&+(\text{Q}+\text {\fontencoding{LGR}\selectfont \Koppa})({r^*_{\infty}})-\text{Q}({r^*_{-\infty}}).
\end{align*}

(We note finally that an alternative approach to the construction of a current
for this range would be considerations similar to
Section~\ref{oneofthesub}.)

\subsection{The $\mathcal{F}_{\mbox{$\natural$}}$ range (trapped frequencies)}
\label{kloubi}
This is the frequency range of trapping.
Here the current seed function
will in general depend on $\omega$, $m$, $\ell$, and $a$.
The dependence is smooth in $\omega$ and $a$.

Given an arbitrary choice of $\lambda_2$, this section will require $\omega_1$ to be sufficiently
large (specifically, $1\ll \lambda_2\omega_1^2$). Note that this 
constraint on the largeness becomes tighter in the limit $\lambda_2\to 0$.
The frequency parameter $\lambda_2$ will be chosen in Section~\ref{sharpest!}.

For $(\omega, m, \ell)\in 
\mathcal{F}_{\mbox {$\natural$}}$, 
we have
\begin{equation}
\label{yilyil}
\lambda_{m\ell} +\omega^2 a^2 \ge (\lambda_2+a^2) \omega^2 \ge (\lambda_2 +a^2)\omega_1^2.
\end{equation}

We set 
$$
V_0=({\lambda_{m\ell} +\omega^2 a^2})\frac {r^2-2Mr+a^2}{(r^2+a^2)^2}
$$
so that 
$$
V_1=V-V_0= \frac{4Mram\omega -a^2m^2}{(r^2+a^2)^2}+
\frac{\Delta(3r^2-4Mr+a^2)}{(r^2+a^2)^3} -\frac{3\Delta^2 r^2}{(r^2+a^2)^4}.
$$
\index{fixed functions! spacetime functions! $V_0$ (proportional to $\lambda_{m\ell}+a^2\omega^2$, satisfies $V=V_0+V_1$)}
\index{fixed functions! spacetime functions! $V_1$ (satisfies $V=V_0+V_1$)}
We easily see that 
\begin{eqnarray}
\label{AP}
r^3 |V_1'|+ \left|\left (\frac {(r^2+a^2)^4}{\Delta r^2} V_1'\right)'\right|&\le& B {\Delta} r^{-2}
\left(a^2m^2+ |a m\omega|+1\right)
.
\end{eqnarray}
On the other hand
\begin{align}
\label{elinde}
\nonumber
V_0'(r^*)&= 2\frac{\Delta}{(r^2+a^2)^4} ({\lambda_{m\ell} +\omega^2 a^2}) \left ( (r-M)(r^2+a^2) - 2r(r^2-2Mr+a^2)\right)
\\ &= -2\frac{\Delta r^2}{(r^2+a^2)^4} \left({\lambda_{m\ell} +\omega^2 a^2}\right)
\left(r-3M+a^2\frac {r+M}{r^2}
\right).
\end{align}

We notice that for all $|a|<M$, the function $s(r)=r^3-3Mr^2+a^2r+a^2M$,
and thus $V_0'$, has a unique simple zero on $(r_+,\infty)$, which
we denote $z_{a,M}$\index{fixed parameters! $r$-parameters! $z_{a,M}$ (unique zero of $V_0'$ for
$\mathcal{F}_{\mbox{$\natural$}}$ range)}.
To see this, in view of the fact that $s'(r)$ has at most two zeros,
it suffices to remark that $\lim_{r\to\infty} s(r)=\infty$ and,
\[
s(r_+)=-M(r_+^2-a^2)<0, \qquad 
\frac{d}{dr}s(r_+)=3r_+^2-6Mr_+^2+a^2=-2a^2<0.
\]
We easily see moreover that $|z_{a,M}-3M|\le Ba^2$.

Given any $r$-neighborhood of $z_{a,M}$, it follows from the
inequality $(\ref{AP})$ (applied with the first term on the left hand side in mind) 
and the identity $(\ref{elinde})$,
taking into account also the asymptotics of $V_0'$ as $r\to r_+$, $r\to \infty$,
and using Proposition~\ref{2ndpr} and the smallness of $a_0$ in the case 
of Theorem~\ref{nmt}, and the vanishing of $m$ in the case of Theorem~\ref{nmt2},
that, given arbitrary $\lambda_2$, then for sufficiently large $\omega_1$ (depending
on $\lambda_2$, so that the right hand side of $(\ref{yilyil})$ is sufficiently large), 
for frequencies in $\mathcal{F}_{\mbox{$\natural$}}$,
$V'$ has at least one zero in this neighborhood, and no zeros outside
this neighborhood.

To see now that $V'$ has in fact a unique simple zero in this neighborhood,
we note first  that there exists a positive $c=c(a,M)$ such that
\begin{equation}
\label{fwsee}
\left(\frac{(r^2+a^2)^4}{\Delta r^2} V_0'\right)'(z_{a,M})
\le - c\Delta r^{-2} ({\lambda_{m\ell} +\omega^2 a^2}).
\end{equation}
To see this,
we compute
\[
\left(\frac{(r^2+a^2)^4}{\Delta r^2} V_0'\right)'(z_{a,M})
=-2\Delta z_{a,M}^{-2}(\lambda_{m\ell}+a^2\omega^2)\left(
1-\frac{a^2}{z_{a,M}^2}-\frac{2Ma^2}{z_{a,M}^3}\right).
\]
Since $s(z_{a,M})=0$, we have
\[
1-\frac{a^2}{z_{a,M}^2}-\frac{2Ma^2}{z_{a,M}^3}
=2+\frac{a^2}{z_{a,M}^2}-3Mz_{a,M}
=z_{a,M}^{-2}(2z_{a,M}^2-3Mz_{a,M}+a^2)\ge
1-\frac{M}{z_{a,M}}>0,
\]
which yields $(\ref{fwsee})$.

From $(\ref{fwsee})$ and $(\ref{AP})$ (now with the second term on the left hand side in 
mind), it follows that given $\lambda_2$ arbitrary,
then for $\omega_1$
sufficiently large (depending on $\lambda_2$),
and, in the case of Theorem~\ref{nmt}, for $|a|\le a_0$ 
with $a_0$ sufficiently small,
we have for frequencies in $\mathcal{F}_{\mbox{$\natural$}}$:
\begin{equation}
\label{convexity}
\left(\frac{(r^2+a^2)^4}{\Delta r^2} V'\right)'\le -(c/2)
\Delta r^{-2}({\lambda_{m\ell} +\omega^2 a^2})
\end{equation}
in a neighborhood of $z_{a,M}$ as described previously
containing any zero of $V'$. 
(We have used here Proposition~\ref{2ndpr} to estimate $a^2m^2$ 
from 
$a^2(\lambda_{m\ell}+a^2\omega^2)$, and it is here that the smallness of $a_0$
is relevant.
In the case of Theorem~\ref{nmt2}, $m=0$, and we need not restrict to small $a_0$.)

The relation $(\ref{convexity})$ in the neighborhood of possible zeros shows
now that $V'$ has \emph{exactly} one simple zero, in fact:

\begin{proposition} 
Under the assumptions of Theorem~\ref{nmt} or~\ref{nmt2},
then for $\omega_1$ sufficiently large (depending on $\lambda_2$), 
$V'$ has a unique zero $r_{m\ell}^{(a\omega)}$\index{fixed parameters! $r$-parameters!
$r_{m\ell}^{(a\omega)}$ (unique zero of $V'$ for
$\mathcal{F}_{\mbox{$\natural$}}$ range)} for 
$(\omega, m, \ell)\in\mathcal{F}_{\mbox{$\natural$}}$ depending smoothly on
the parameters.
For fixed  $m$, $\omega$,
\[
\lim_{\ell\to\infty}r_{m\ell}^{(a\omega)}=z_{a,M}.
\]
There exist $r_{\mbox{$\natural$}}^-(a_0,M,\omega_1,\lambda_2)<r_{\mbox{$\natural$}}^+(a_0,M,\omega_1,\lambda_2)$
with $r_{\mbox{$\natural$}}^\pm \to 3M$ as $a_0\to 0$ such
that\index{fixed parameters! $r$-parameters!  $r_{\mbox{$\natural$}}^\pm$ (associated
to $\mathcal{F}_{\mbox{$\natural$}}$)}
\[
-V' (r^*) \ge b\,\chi_3( r^*-((r_{\mbox{$\natural$}}^-)^*+(r_{\mbox{$\natural$}}^+)^*)/2) (\Delta/r^2) (\lambda_2\omega_1^2+a\omega_1^2)
\]
for $r\le r_{\mbox{$\natural$}}^-$, and
\[ 
V' (r^*) \ge b\, \chi_3(r^*-((r_{\mbox{$\natural$}}^-)^*+(r_{\mbox{$\natural$}}^+)^*)/2)
 r^{-3}(\lambda_2\omega_1^2+a\omega_1^2)
\]
for $r\ge r_{\mbox{$\natural$}}^+$,
where $\chi_3$ is a fixed function
such that $\chi_3(x)=|x|$ for $|x|\le .5$, $\chi_3(x) \ge .5$ 
for $|x|\ge .5$ and $\chi_3=1$ for $|x|=1$.
\end{proposition}
In general, limit points of the collection $r_{m\ell}^{(a\omega)}$ correspond
to $r$-values which admit trapped null geodesics. Let us note moreover
that as $a_0\to M$, we have
$r_{\mbox{$\natural$}}^-\to M=r_+(M,M)$.

Note that we have included the 
$\chi_3( r^*-((r_{\mbox{$\natural$}}^-)^*+(r_{\mbox{$\natural$}}^+)^*)/2)$ merely
so that our estimates do not degenerate in the limit $a_0\to 0$, in accordance
with our convention that $b$, $B$ be independent of $a$ for Theorem~\ref{nmt}.

One can now clearly construct functions $f$ whose properties are summarised
in the proposition below:\index{seed functions! $f$  (seed function $f=f(r^*)$ used with template $\text {Q}$)}
\begin{proposition}
Let $\lambda_2$ be given, $\omega_1$ be sufficiently
large (depending on $\lambda_2$), and, in the case of Theorem~\ref{nmt},
let $|a|\le a_0$ for sufficiently small $a_0$, whereas, in the case of Theorem~\ref{nmt2},
let $|a|\le a_0$ for arbitrary $a_0<M$. 
Then, for each $(m,\ell, \omega)\in \mathcal{F}_{\mbox{$\natural$}}$,
there exists a function
$f=f_{m\ell}^{(a\omega)}(r^*,a)$ depending smoothly on $a\omega$ and $a$ such that
\begin{enumerate}
\item $f'(r^*)\ge 0$ for all $r^*$, and $f'\ge b(\Delta/r^2)r^{-2}>0$ for $r_{c}\le r\le R$
\item $f< 0$ for $r< r_{m\ell}^{(a\omega)}$ and $f> 0$ for $r> r_{m\ell}^{(a\omega)}$ ,
\item
\label{proper3}
 $-fV'-\frac 12 f'''\ge b (\Delta/r^2) r^{-3}>0$,
 \item
 $fV' \ge b\chi_3^2(r^*-((r_{\mbox{$\natural$}}^-)^*+(r_{\mbox{$\natural$}}^+)^*)/2) (\Delta/r^2) (\lambda_2\omega_1^2+a^2\omega_1^2)$ in $r\le r_{\mbox{$\natural$}}^-$,
 $fV'\ge b\chi_3^2(r^*-((r_{\mbox{$\natural$}}^-)^*+(r_{\mbox{$\natural$}}^+)^*)/2) r^{-3}(\lambda_2\omega_1^2+a^2\omega_1^2)$ in $r\ge r_{\mbox{$\natural$}}^+$.
 \item
 $\lim_{r^*\to-\infty}f(r^*)=-1$,
 \item 
for $r^*\ge R^*_4$, $f=1$, for some $R_4$.\index{fixed parameters! $r$-parameters!
$R_4$ (large parameter
associated with $\mathcal{F}_{\mbox{$\natural$}}$)}
\end{enumerate}
\end{proposition}

As always, the convention is that $b$ depends only on $M$, and, in the
case of Theorem~\ref{nmt2}, also on $a_0$.

We apply now the energy identity corresponding to the current
$\text{Q}^f$.
We obtain, for $r^*_{-\infty}<r^*_{c}$, $r^*_{\infty}>R^*$,
\begin{align*}
b\int_{r^*_{c}}^{R^*}
&(\Delta/r^2)r^{-2}
|u'|^2  +  (\Delta/r^2)r^{-2}|u|^2\\
&+ 
\Delta/r^2 r^{-3}\chi_3^2(r^*-((r_{\mbox{$\natural$}}^-)^*+(r_{\mbox{$\natural$}}^+)^*)/2) (1-\chi_{[r_{\mbox{$\natural$}}^-,r_{\mbox{$\natural$}}^+]})((\lambda_{m\ell}+
a^2\omega^2)+\omega^2) |u|^2 \\
\le
 &\int_{r^*_{-\infty}}^{r^*_{\infty}}2f{\rm Re}(u'\bar H)+f'{\rm Re}(\bar Hu)\\
&+\text{Q}(r^*_{\infty})-\text{Q}(r^*_{-\infty}),
\end{align*}
where $\chi_{[r_{\mbox{$\natural$}}^-,r_{\mbox{$\natural$}}^+]}$
denotes the indicator function of $[r_{\mbox{$\natural$}}^-,r_{\mbox{$\natural$}}^+]$.

In essentially replacing $f$ by $(1-\chi_{[r_{\mbox{$\natural$}}^-,r_{\mbox{$\natural$}}^+]})
\chi_3(r^*-((r_{\mbox{$\natural$}}^-)^*+(r_{\mbox{$\natural$}}^+)^*)/2)$, we have thrown away some information.  
This is because this extra positivity cannot be characterized by a differential operator
after summation,
whereas for convenience, we have stated our main theorem as a classical integrated
energy estimate. For a more refined ``pseudodifferential'' statement, one merely
should retain the frequency dependent $f$ multiplying
the term $(\lambda_{m\ell}+a^2\omega^2) +\omega^2$.
It is this sum which is proven to be bounded.

It is worth adding here that in the case $m=0$, or alternatively, in the Schwarzschild
case $a=0$, by a slight variant
of the above construction, one in fact could have chosen
 $f$ independently of $\omega$, $m$, 
$\ell$,
centred always at $z_{a,M}$. Cf.~the construction of Section~4.1.1 of~\cite{jnotes} for 
Schwarzschild.

\subsection{The $\mathcal{F}_{\mbox{$\sharp$}}$ range (time-dominated frequencies)}
\label{sharpest!}
First some general facts:
Note that for all $|a|<M$, we have
\begin{equation}
\label{kaiautotox}
\frac{\Delta a^2}{(r^2+a^2)^2}<c<1
\end{equation} 
for all $r>r_+$.
For small enough $\lambda_2$ and large enough $\omega_1$, 
we have, using also $(\ref{kaiautotox})$ and Proposition~\ref{2ndpr}, say
\[
\omega^2-V \ge \frac{1-c}2\omega^2,
\]
for frequencies
in $\mathcal{F}_{\mbox{$\sharp$}}$.
Finally, there exists an $R_5$\index{fixed parameters! $r$-parameters!
$R_5$ (large parameter 
associated with $\mathcal{F}_{\mbox{$\sharp$}}$)}
 such that,
for all $|a|<M$ and all frequencies (not just in $\mathcal{F}_{\mbox{$\sharp$}}$), 
we have $V' (r)< 0$ for $r\ge R_5$.
We may take $R_5\ge R$.

We shall define here a $\text {\fontencoding{LGR}\selectfont \koppa}^y$-current.

In the case of Theorem~\ref{nmt2}, we note that under the decomposition
$V=V_0+V_1$ of the previous section, since $m=0$, then $V_1$ is independent
of $\omega$, and say
\[
|V_1'|\le C\Delta/r^2 r^{-3} 
\]
in $\mathcal{F}_{\mbox{$\sharp$}}$, for $\omega_1$ sufficiently large.
$V=V_+$. Recalling $z_{a,M}$ from the previous section,
we may now choose, for $\omega_1$ sufficiently large,
a $y=y(r^*)$ such that $y'\ge0$ everywhere,
$y(z_{a,M})=0$, 
$y'\ge 2C\omega_1^{-1} \Delta/r^2 r^{-2}$ in $(r_+,R_5]$ and $y=1$ for $y\ge R_5+1$
say, without loss of generality we have
selected $R_5$  such that also $V_1'\le 0$ for $r\ge R_5$.

In the case of Theorem~\ref{nmt}, one does not in general have a unique
vanishing point for $V'_0$ in this frequency range. 
Let us note, however,
that for $(\omega, m, \ell)\in  \mathcal{F}_{\mbox{$\sharp$}}$, we have
\[
|V'| \le B(\Delta/r^2)r^{-3}((\lambda_{m\ell}+a_0^2\omega^2)+1).
\]
We may now define a function $y$ such that say 
$y\ge 0$, $y'\ge B\Delta/r^2 r^{-2}$
in $(r_{+},R_5]$ and $y=1$ for $y\ge R_5+1$.

We apply  now $(\ref{eq:Q2for})$ with $\text {\fontencoding{LGR}\selectfont \koppa}^y$. 
For $r\le r_{c}$, we have that
\[
\text {\fontencoding{LGR}\selectfont \koppa}' 
\ge 2y{\rm Re}(u'\bar H)
\]
where, 
in the case of Theorem~\ref{nmt}, we must possibly further restrict
$\omega_1$ to be large.

On the other hand, 
for $r_{c}\le r \le  R_5$, we have, choosing $\omega_1$, sufficiently large, 
$\lambda_2$ sufficiently small, and in the case of Theorem~\ref{nmt}, choosing $a_0$
sufficiently small,
\[
\text {\fontencoding{LGR}\selectfont \koppa}' \ge b (\Delta/r^2) r^{-2}(\omega^2 +
(\lambda_{m\ell}+a^2\omega^2) + 1)|u|^2+ 2y{\rm Re}(u'\bar H).
\]
Finally, 
for $r\ge R_5$,
 we have 
\[
\text {\fontencoding{LGR}\selectfont \koppa}'\ge  2y{\rm Re}(u'\bar H).
\]

In view of the inequality $R_5\ge R$, 
we obtain thus for $r^*_{\infty}>R^*$, $r^*_{-\infty}<r^*_{c}$
\begin{align*}
\int_{r^*_{c}}^{R^*}
 &b\Delta/r^2 r^{-2}\left(|u'|^2  +(\omega^2 + (\lambda_{m\ell}+a^2\omega^2) +1)|u|^2\right)\\
\le&\int_{r^*_{-\infty}}^{r^*_{\infty}}2y{\rm Re}(u'\bar H)\\
&+\text {\fontencoding{LGR}\selectfont \koppa}' ({r^*_{\infty}})-
\text {\fontencoding{LGR}\selectfont \koppa}' ({r^*_{-\infty}}).
\end{align*}

\section{Summing}
We now wish to reinstate the dropped indices $m,\ell,a\omega$.
For all $r^*_{-\infty}$, $r^*_{\infty}$ we have obtained an identity which we may write as
\begin{align*}
\int_{r^*_{c}}^{R^*}{\bf M}^{(a\omega)}_{m\ell} \,dr^*
\le&
\left(\int_{r^*_{-\infty}}^{r_c^*}+\int_{R^*}^{r^*_{\infty}}\right){\bf E}^{(a\omega)}_{m\ell} \,dr^*\\
&+\int_{r^*_{-\infty}}^{r^*_{\infty}}{\bf C}^{(a\omega)}_{m\ell} \,dr^*\\
&+ {\bf B}^{(a\omega)}_{m\ell}(r^*_{\infty})
- {\bf B}^{(a\omega)}_{m\ell}(r^*_{-\infty}).
\end{align*}
Here ${\bf C}$ represents the terms containing the inhomogeneous term $F$
arising from the cutoff.
Summing over $m$, $\ell$
and integrating over $\omega$, we obtain
\begin{align}
\label{Sainteg}
\nonumber
\int_{-\infty}^{\infty}\int_{r^*_{c}}^{R^*}&
\sum_{m\ell}
{{\bf M}^{(a\omega)}_{m\ell}} \,d\omega\, dr^*
\le
\int_{-\infty}^{\infty}
\left(\int_{r^*_{-\infty}}^{r_c^*}+\int_{R^*}^{r^*_{\infty}}\right){\bf E}^{(a\omega)}_{m\ell} \,d\omega \,dr^*
\\
\nonumber
&+\int_{-\infty}^{\infty}\int_{r^*_{-\infty}}^{r^*_{\infty}}\sum_{m\ell} {{\bf C}^{(a\omega)}_{m\ell}}\, 
d\omega\, dr^*\\
&
+\int_{-\infty}^\infty \sum_{m\ell}\left( {\bf B}^{(a\omega)}_{m\ell}(r^*_{-\infty})
-{\bf B}^{(a\omega)}_{m\ell}(r^*_{-\infty})\right) d\omega,
\end{align}
where we have used the regularity properties of the relevant functions to interchange
$\int_{r^*_{-\infty}}^{r^*_{\infty}}$ with $\sum_{m\ell}$.

We thus have that the $\liminf$ of the left hand side as $r^*_{\pm\infty}\to\pm\infty$
is less than equal to the $\limsup$  of the right hand side.
This will be the main inequality.

Let us note finally that we may now consider the entirety of
our frequency parameters ($\omega_1$, $\lambda_1$, $\lambda_2$, $\omega_3$)
to have been chosen, with the choices depending only on $M$, and, in the
case of Theorem~\ref{nmt2}, possibly $a_0$.
Thus, constants which depend only on these may in what follows be denoted simply 
by $B$, $b$.

\subsection{The main term}
\label{maintermse}
For all $r^*_{-\infty}$ sufficiently negative and $r^*_{\infty}$  sufficiently positive
we have in view of the properties of Section~\ref{idiotntes} that
\begin{align*}
b\int_{\{r_{c}\le r\le R\}} r^{-3} &\left ((\pa_{r^*}  \psi_{\hbox{\Rightscissors}})^2
+\psi_{\hbox{\Rightscissors}}^2\right)\\
& 
+r^{-3}(1-\chi_{[r_{\mbox{$\natural$}}^-,r_{\mbox{$\natural$}}^+]})\chi_3^2
(r^*-((r_{\mbox{$\natural$}}^-)^*+(r_{\mbox{$\natural$}}^+)^*)/2)( (\partial_t \psi_{\hbox{\Rightscissors}})^2 + |\nabb \psi_{\hbox{\Rightscissors}}|^2_{\slashg})\\
&\le  \int_{-\infty}^{\infty}\int_{r^*_{-\infty}}^{r^*_{\infty}}
\sum_{m\ell}
{{\bf M}^{(a\omega)}_{m\ell}} \,d\omega\, dr^*.
\end{align*}
In accordance with our conventions, the 
integral on the left is now with respect to the volume form.
In particular, the inequality holds for the $\liminf$ of the right hand side.
We obtain immediately
\begin{align}
\label{aptomaint}
\nonumber
 b\int_{\{0\le t^*\le \tau' \} \cap \{r_{c}\le r\le R\}} &r^{-3}((\pa_{r^*}\psi)^2+\psi^2)\\
 \nonumber
 &
 +r^{-3}(1-\chi_{[r_{\mbox{$\natural$}}^-,r_{\mbox{$\natural$}}^+]})
 \chi_3^2(r^*-((r_{\mbox{$\natural$}}^-)^*+(r_{\mbox{$\natural$}}^+)^*)/2)
 ((\partial_t\psi)^2+|\nabb\psi|^2_{\slashg})\\
 \nonumber
 \le& 
\liminf_{r^*_{\pm\infty}\to\pm\infty}
 \int_{-\infty}^{\infty}\int_{r^*_{-\infty}}^{r^*_{\infty}}
\sum_{m\ell}
{{\bf M}^{(a\omega)}_{m\ell}} \,d\omega\, dr^*\\
\nonumber
&+B(r_{c}, R)\varepsilon^{-1}
 \int_{\{t^*=\tau'-\varepsilon^{-1}\} }{\bf J}^{T+e_0N}_\mu[\psi]n^\mu_{\{t^*=\tau'-\varepsilon^{-1}\}}\\
 &
+B(r_{c}, R)\varepsilon^{-1} \int_{\{t^*=0\} }{\bf J}^{T+e_0N}_\mu[\psi]n^\mu_{\{t^*=0\}},
\end{align}
where the last term arises because the domain of integration 
on the left hand side includes $[0,\varepsilon^{-1}]$ and $[\tau'-\varepsilon^{-1},\tau']$,
and (in the case of Theorem~\ref{nmt}) we have 
appealed to (the first inequality of) Proposition~\ref{needthis2}.
(In the case of Theorem~\ref{nmt2}, we may define $e_0$ and $\varepsilon$ to be
1 and appeal to Theorem~\ref{btheorem2}.)

\subsection{Error terms near the horizon and infinity}
\label{fromhorfrominf}
We have
\begin{align}
\label{fromhorer}
\nonumber
\int_{-\infty}^{\infty}\int_{r^*_{-\infty}}^{r_{c}^*}
 \sum_{m\ell}{\bf E}^{(a\omega)}_{m\ell} \,d\omega \,dr^*
&\le 
B\cdot
(a_0+q)\int_{\{r\le r_{c}\}\cap\{0\le t^*\le \tau'\}}  {\bf J}^N_\mu [\psi_{\hbox{\Rightscissors}}]N^\mu\\
\nonumber&\qquad\qquad\qquad
+|\log(r-r_+)|^{-2}(r-r_+)^{-1} \psi_{\hbox{\Rightscissors}}^2\\
&\le 
\nonumber
B\cdot(a_0+q) \int_{\{r\le r_{c}\}\cap\{0\le t^*\le \tau'\}}  {\bf J}^N_\mu[\psi]N^\mu\\
&\qquad\qquad\qquad
+|\log(r-r_+)|^{-2}(r-r_+)^{-1}\psi^2.
\end{align}
Here $q$ is the parameter of Section~\ref{oneofthesub}, which remains to be chosen.
In the case of Theorem~\ref{nmt2}, this estimate holds without the $a_0$ term.

On the other hand, we have
\begin{align}
\label{erfrominf}
\nonumber
\int_{-\infty}^{\infty}
\int_{R^*}^{r^*_{\infty}}\sum {\bf E}^{(a\omega)}_{m\ell} \,d\omega \,dr^*
\le& B\int_{\{R_2\le r\le \epsilon_2^{-1}R_2\}\cap\{0\le t^*\le \tau'\} } 
(\epsilon_2r^{-2}+r^{-3}){\bf J}^T_\mu [\psi_{\hbox{\Rightscissors}}]T^\mu\\
\nonumber
 \le& B_\delta\varepsilon^{-1}(\epsilon_2 +R_2^{-1}) \left(\int_{\{t^*=0\}} {\bf J}^{T+e_0N}_\mu [\psi] n^\mu +\int_{\{t^*=\tau'-\varepsilon^{-1}\}} {\bf J}^{T+e_0N}_\mu [\psi] n^\mu \right)\\
 &+
 B_\delta(\epsilon_2R_2^{-1+\delta}+R_2^{-2+\delta})\int_{\{r\ge R\}\cap\{0\le t^*\le \tau'\}} r^{-1-\delta}{\bf J}^N_\mu [\psi] n^\mu
 \end{align}
where the spatial integrals arise in estimating the $0$'th order term
(via Hardy) which arises from the band where $\psi$ and $\psi_{\hbox{\Rightscissors}}$ 
do not coincide. We have used the first inequality of Proposition~\ref{needthis2}
in the case of Theorem~\ref{nmt}, while we have simply appealed to
Theorem~\ref{btheorem2} in the case of Theorem~\ref{nmt2}, taking in the latter
case $e_0=1$.

One should perhaps
note in advance that these terms will be absorbed in Section~\ref{integdecaysec} 
by the main term
of Section~\ref{maintermse}, after appealing to  Propositions~\ref{ftrs},~\ref{giasuper}
and~\ref{lrp}, 
and choosing the parameters accordingly.

\subsection{Error terms from the cutoff}
\label{erfromcut}
These are the terms containing $H$.

Let $R_7\ge R$ be such that for $r\ge R_7$, the current seed functions 
satisfy $y=f=1$
for all frequencies. We may take $R_7$ to be specifically:
\begin{equation}
\label{R7def}
R_7=\max\{R_3, \epsilon_2^{-1}R_2, R_4, R_6+1, R_5+1\}.
\end{equation}
Note that all these parameters have already been selected
with the exception of $\epsilon_2$, $R_2$.
\index{fixed parameters! $r$-parameters! $R_7$ (large parameter associated with the seed
functions being constant)}

We now split the error into two parts:
\begin{align}
\label{3parts}
\nonumber
\int_{-\infty}^{\infty}\int_{r^*_{-\infty}}^{r^*_{\infty}}\sum_{m\ell} {{\bf C}^{(a\omega)}_{m\ell}}\, 
d\omega\, dr^*
=&
\int_{-\infty}^{\infty}\int_{R_7^*}^{r^*_{\infty}}\sum_{m\ell}{}{{\bf C}^{(a\omega)}_{m\ell}}\, 
d\omega\, dr^*\\
&+\int_{-\infty}^{\infty} \int_{r^*_{-\infty}}^{R_7^*}\sum_{m\ell}{{\bf C}^{(a\omega)}_{m\ell}}\, 
d\omega\, dr^*.
\end{align}

The integrand in $r^*$ of the second term of $(\ref{3parts})$
can be written (where we have used the properties
of Section~\ref{idiotntes})
\begin{align*}
 \int_{-\infty}^\infty& \sum_{m,\ell} 
c^{(a\omega)}_{m\ell}(r){\rm Re}(\bar F^{(a\omega)}_{m\ell}(r) \Psi^{(a\omega)}_{m\ell}(r)+
d^{(a\omega)}_{m\ell}(r){\rm Re}(F^{(a\omega)}_{m\ell}(r)
(\partial_{r^*}\Psi)^{(a\omega)}_{m\ell}(r) )d\omega\\
\le&
 \int_{-\infty}^\infty \sum_{m,\ell} 
\epsilon_3^{-1}  (\Delta/r^2)r^{2} \upupsilon(r) |F^{(a\omega)}_{m\ell}|^2(r)+
\epsilon_3 (\Delta/r^2) r^{2}
(r^{-3}|\Psi^{(a\omega)}_{m\ell}|^2 +r^{-3}|(\partial_{r^*}\Psi)^{(a\omega)}_{m\ell}|^2 d\omega\\
=& \epsilon_3^{-1} (\Delta/r^2)r^{2}\upupsilon(r)
 \int_{-\infty}^{\infty}\int_0^{2\pi}\int_0^\pi F^2 \sin \theta \, d\phi \, d\theta \,dt\\
 & +
\epsilon_3  (\Delta/r^2)r^2 \int_{-\infty}^{\infty}\int_0^{2\pi}\int_0^\pi
r^{-3}((\psi_{\hbox{\Rightscissors}})^2
+(\partial_{r^*}\psi_{\hbox{\Rightscissors}})^2) \sin \theta \, d\phi \, d\theta \,dt \\
\le&  \epsilon_3^{-1}B(\Delta/r^2)\upupsilon(r) r^2  \int_{-\infty}^{\infty}\int_0^{2\pi}\int_0^\pi
 F^2 \sin \theta \, d\phi \, d\theta \,dt\\
&+B
\epsilon_3 (\Delta/r^2)r^2 \int_{0}^{\tau'}\int_0^{2\pi}\int_0^\pi r^{-3}(\psi^2+(\partial_{r^*}\psi)^2)
 \sin \theta \, d\phi \, d\theta \,dt^*,
\end{align*}
(where the $\partial_{r^*}$ derivative is still in $(t,r,\theta,\phi)$ coordinates),
where $\epsilon_3>0$\index{fixed parameters! small parameters! $\epsilon_3$ (associated to bounding
cutoff terms)} can be chosen arbitrarily, and
where $\upupsilon(r)$ is a nonnegative function depending on the
choice of $R_7$ and $q$,
with
\[
\sup_{r>r_+} \upupsilon(r) \le B(R_7, q).
\]
The reader should note that unfavourable powers of $r$ have all been incorporated
in the definition of $\upupsilon$ in view of the fact that in the domain of integration
$r\le R_7$.

Integrating thus with respect to $r^*$, and recalling our comments concerning
the volume form in Section~\ref{usefulcomps}, and
using the estimates for $F$ in Section~\ref{cutoffsec},
we obtain
\begin{align}
\label{toukatof}
\nonumber
\int_{-\infty}^{\infty} \int_{r^*_{-\infty}}^{R_7^*}\sum_{m\ell}{{\bf C}^{(a\omega)}_{m\ell}}\, 
d\omega\, dr^*
\le&  B(R_7, q) 
\epsilon_3^{-1}\varepsilon^2
\left(\int_{\{0\le t^*\le \varepsilon^{-1}\}\cap\{r\le R_7\}}+\int_{\{\tau'-\varepsilon^{-1}\le t^*\le \tau'\}
\cap\{r\le R_7\}}
\psi^2+{\bf J}_\mu^{N}[\psi]N^\mu\right)\\
\nonumber
&+B \epsilon_3 \int_{\{0\le t^*\le \tau' \}\cap\{r\le R_7\}}
 r^{-3}(\psi^2+(\partial_{r^*}\psi)^2)\\
 \nonumber
 \le&B(R_7, q)\epsilon_3^{-1}\varepsilon
 \left( \int_{\{t^*=0\} }+\int_{\{t^*=\tau'-\varepsilon^{-1}\} }\right)
 {\bf J}^{N}_\mu[\psi]n^\mu\\
 &
 +B \epsilon_3 \int_{\{0\le t^*\le \tau'\}\cap\{r\le R_7\}}
 r^{-3}(\psi^2+(\partial_{r^*}\psi)^2)
\end{align}
We have appealed in the above to the first inequality of
Proposition~\ref{needthis2}.
Note that it is ${\bf J}^N$ (equivalently ${\bf J}^{T+N}$)
and not ${\bf J}^{T+e_0N}$ which appeared above, and this is
why it is essential that we have the extra smallness parameter $\varepsilon$,
arising from the estimate of $F$.

We turn to the first term of $(\ref{3parts})$. In view of the fact that
$f(r)=1$ (or $y(r)=1$) is independent of $\omega$, $m$, $\ell$, we have in fact
that this term equals precisely
\begin{eqnarray}
\label{ctoo}
\nonumber
\int_{R^*_7}^{r^*_\infty} \int_{-\infty}^\infty&& \sum_{m,\ell} 
{\rm Re}((r^2+a^2)^{1/2}\Psi^{(a\omega)}_{m\ell})'(\Delta (r^2+a^2)^{-1/2}\bar F^{(a\omega)}_{m\ell})  d\omega\, dr^*\\
\nonumber
&=& \int\int\int_{R^*_7}^{r^*_\infty}\int_{-\infty}^\infty\left( \partial_{r^*}((r^2+a^2)^{1/2}\psi_{\hbox{\Rightscissors}})\right)
\Delta(r^2+a^2)^{-1/2} F \sin\theta\, d\phi\, d\theta\, dt^* dr^*,\\
\end{eqnarray}
where   we have here used the properties of Section~\ref{idiotntes}.
This integral is supported only in $\{0\le t^*\le \varepsilon^{-1}
\}\cup \{\tau'-\varepsilon^{-1}\le t^*\le \tau'\}\cap\{r\ge R_7\}$.
Recalling that 
\begin{equation}
\label{recallingthat}
F=2\nabla^\alpha\xi_{\tau',\varepsilon}\nabla_\alpha\psi+ (\Box_g\xi_{\tau',\varepsilon})\psi,
\end{equation}
and noting
that $\partial_{r^*}\psi_{\hbox{\Rightscissors}}=\xi_{\tau',\varepsilon} \partial_{r^*}\psi$
for sufficiently large $r$,
we see immediately that the contribution to $(\ref{ctoo})$ of the first term of
$(\ref{recallingthat})$ can be controlled (using the first inequality of
Proposition~\ref{needthis2} and a Hardy inequality) by
\begin{equation}
\label{easilyby}
B
\left( \int_{\{t^*=0\}}+\int_{\{t^*=\tau'-\varepsilon^{-1}\} }\right) {\bf J}^{T+e_0N}_\mu[\psi]n^\mu.
\end{equation}
One should note in the estimate above that
the $\varepsilon^{-1}$ factor coming from the integral
in time cancels the $\varepsilon$ factor appearing in the estimate for $F$.

We also note that $\Delta(r^2+a^2)^{-1/2}$ differs from $(r^2+a^2)^{1/2}$ by
terms which are lower order in $r$ and can as above be bounded by $(\ref{easilyby})$.
We are thus left to bound
\[
 \int\int\int_{R_7^*}^\infty
 \int_{-\infty}^\infty\left( \partial_{r^*}((r^2+a^2)^{1/2}\xi_{\tau',\varepsilon}\psi)\right)
(r^2+a^2)^{1/2} \Box_g\xi_{\tau',\varepsilon}\psi \sin\theta\, d\phi\, d\theta\, dt^* dr^*,
\]
which we may rewrite as 
\[
 \int\int\int_{R_7^*}^\infty
 \int_{-\infty}^\infty\left( \partial_{r^*}((r^2+a^2)^{1/2}\xi_{\tau',\varepsilon}\psi)\right)
\left((r^2+a^2)^{1/2} \xi_{\tau',\varepsilon }\psi\right) (\xi_{\tau',\varepsilon}^{-1}\Box_g\xi_{\tau',\varepsilon}) 
\sin\theta\, d\phi\, d\theta\, dt^* dr^*.
\]

It follows that we may integrate the above by parts with
respect to $r^*$ say, to yield a boundary term supported only in 
 $\{0\le t^*\le \varepsilon^{-1}\}\cup \{\tau'-\varepsilon^{-1}\le t^*\le \tau'\}\cap\{r\ge R\}$, 
 easily bound by $(\ref{easilyby})$, 
 a vanishing boundary term at infinity for sufficiently large $r^*$ in view of the compactness
 of the support, and finally a term
 \[
\frac12 \int\int\int_{R_7^*}^\infty
\int_{-\infty}^\infty\left((r^2+a^2)^{1/2}\xi_{\tau',\varepsilon}\psi\right)^2
\partial_{r^*}(\xi_{\tau',\varepsilon}^{-1}
\Box_g\xi_{\tau',\varepsilon}) \sin\theta\, d\phi\, d\theta\, dt^* dr^*,
\]
which, in view of the bound in $r\ge R_7$:
\[
|\partial_{r^*}\xi\Box_g \xi|
+|\xi_{\tau',\varepsilon }\partial_{r^*}\Box_g \xi_{\tau',\varepsilon}|\le
B \varepsilon r^{-2}
\]
is also easily bound by $(\ref{easilyby})$.
(In fact the above bound holds with $\varepsilon^2$ but this is not here
necessary.)

Thus we have obtained that 
\begin{equation}
\label{eleos...}
\limsup_{r_{\pm\infty}^*\to\pm\infty}\int_{-\infty}^{\infty}\int_{r^*_{-\infty}}^{r^*_{\infty}}\sum_{m\ell} {{\bf C}^{(a\omega)}_{m\ell}}\, 
d\omega\, dr^*
\le 
B\left( \int_{\{t^*=0\}}+\int_{\{t^*=\tau'-\varepsilon^{-1}\} }\right) {\bf J}^{T+e_0N}_\mu[\psi]n^\mu.
\end{equation}

\subsection{Boundary terms}
\label{bterms}
By the properties of Section~\ref{idiotntes}, we obtain the estimate
\begin{align*}
\int_{-\infty}^\infty\sum_{m\ell}&{\bf B}^{(a\omega)}_{m\ell}(r^*_{-\infty})d\omega \\
\le& B(q) \int_{-\infty}^\infty\int_0^{2\pi}\int_0^\pi \left( (\partial_t\psi_{\hbox{\Rightscissors}})^2 +(\partial_{r^*}\psi_{\hbox{\Rightscissors}})^2+|a|(\partial_\phi\psi_{\hbox{\Rightscissors}})^2\right)
(r^*_{-\infty})
\sin\theta \,dt\,d\phi\,d\theta\\
&+ B(q)(|r^*_{-\infty}|+1)^{-1}  \int_{-\infty}^\infty\int_0^{2\pi}\int_0^\pi (\psi_{\hbox{\Rightscissors}}^2+
|\nabb\psi_{\hbox{\Rightscissors}}|^2)(r^*_{-\infty}) \sin\theta
\,dt\,d\phi\,d\theta.
\end{align*}
Note that we have absorbed factors of $r$ in $B$. The factor $(|r^*_{-\infty}|+1)^{-1}$ could in
fact be replaced by an exponential.

Taking the limit supremum 
as $r^*_{-\infty}\to-\infty$, we have that the second term on the right side tends to $0$,
while the integral in the first term has a well-defined limit which can be expressed
as an integral on the event horizon. We have thus
\begin{align}
\label{kiedwari}
\nonumber
\limsup_{r^*_{-\infty}\to-\infty}\int_{-\infty}^{\infty} &\sum_{m\ell}{\bf B}^{(a\omega)}_{m\ell}(r^*_{-\infty})d\omega \\
\nonumber
\le& 
B(q)\int_{-\infty}^\infty\int_0^{2\pi}\int_0^\pi\left( (\partial_t\psi_{\hbox{\Rightscissors}})^2 
+(\partial_{r^*}\psi_{\hbox{\Rightscissors}})^2+
|a|(\partial_\phi\psi_{\hbox{\Rightscissors}})^2\right) (r_+)\,\sin\theta \,dt\,d\phi\,d\theta\\
\nonumber
\le& B(q) \int_{\mathcal{H}^+}{\bf J}^{T+e_0N}_\mu
[\psi_{\hbox{\Rightscissors}}]n^\mu_{\mathcal{H}^+}+
B(q)a_0 \int_{\mathcal{H}^+} (\partial_{\phi^*}\psi_{\hbox{\Rightscissors}})^2\\
\nonumber
\le &
 B(q) \int_{\mathcal{H}^+\cap\{0\le t^*\le \tau'\}} {\bf J}^{T+e_0N}_\mu
[\psi]n^\mu_{\mathcal{H}^+}+
B(q)a_0 \int_{\mathcal{H}^+\cap\{0\le t^*\le \tau'\}}{\bf J}^{N}_\mu
[\psi]n^\mu_{\mathcal{H}^+}\\
\nonumber
&+
B(q)\varepsilon^2 \left(\int_{\mathcal{H}^+\cap\{0\le t^*\le \varepsilon^{-1} \}}+
\int_{\mathcal{H}^+\cap\{\tau'-\varepsilon^{-1} \le t^*\le \tau' \}}   \right) \psi^2\\
\nonumber
\le &
 B(q) \int_{\mathcal{H}^+\cap\{0\le t^*\le \tau'\}} {\bf J}^{T+e_0N}_\mu
[\psi]n^\mu_{\mathcal{H}^+}+
B(q)a_0 \int_{\mathcal{H}^+\cap\{0\le t^*\le \tau'\}}{\bf J}^{N}_\mu
[\psi]n^\mu_{\mathcal{H}^+}\\
&+
B(q)\varepsilon\left( \int_{\{t^*=0\} }+\int_{ \{t^*=\tau'-\varepsilon^{-1}\} }\right) 
{\bf J}^N_\mu[\psi]N^\mu.
\end{align}
Here we have used the fact (see Section~\ref{kilhorsec}) that on $\mathcal{H}^+$,
$\partial_{r^*}-\partial_t = C(M,a)\partial_\phi$, with $|C|(M, a)\le Ba_0$,
Proposition~\ref{needthis2} and 
a Hardy inequality to deal with the $0$'th order horizon term.

Note that in the case of Theorem~\ref{nmt2}, the above inequality holds
formally setting $e_0$ and $a_0$ to $0$.

The other boundary term vanishes for sufficiently large $r^*_{\infty}$ in view of the
compactness of the support as assumed in the reduction of Section~\ref{reducsec}.

\subsection{Finishing the proof}
\label{integdecaysec}
We now add $e$ times the estimate of Propositions~\ref{ftrs} and~\ref{lrp} (applied to 
the original $\psi$ with $\Sigma=
\{t^*=0\}$ and $\tau=\tau'$) to the limit of~$(\ref{Sainteg})$, choosing $e$ sufficiently small,
$\tilde{r}$ sufficiently close to $r_+$, and $\tilde{\delta}$ sufficiently small
so that $\tilde{r}+2\tilde{\delta}\le r^-_{\mbox{$\natural$}}$, 
so that the bulk terms on the right
hand side of the former
two propositions are absorbed in the lower bound for the main term of $(\ref{Sainteg})$,
as derived in Section~\ref{maintermse}. The constant $e$ is now fixed, and
we arrive at an inequality where in particular 
\begin{equation}
\label{contip}
b \int_{\{0\le t^*\le \tau' \}}  er^{-1-\delta} {\bf J}^{N}[\psi] +(e|\log(r-r_+)|^{-2}(r-r_+)^{-1} +
r^{-3-\delta})\psi^2 
\end{equation}
appears on the left hand side, together with the nonnegative boundary integrals
on $\mathcal{H}^+$ and $\{t^*=\tau\}$. The constant $e$ being chosen, we may in what
follows
encorporate it into constants $b$, $B$, in accordance with our usual conventions.
On the right hand side are error terms which must be controlled, as well as
a term of the form
\begin{equation}
\label{dataterm}
B(q, R_7,\epsilon_3)\int_{ \{t^*=0\} }{\bf J}^N[\psi]N^\mu,
\end{equation}
our data term. 
We wish to absorb the right hand side error terms into $(\ref{contip})$, producing
only more terms $(\ref{dataterm})$ on the right.

Choosing $q>0$ sufficiently small (and, in addition, $a_0>0$ sufficiently small
in the Case of Theorem~\ref{nmt2}) 
it follows that one can absorb the error terms arising from the right hand
side of $(\ref{fromhorer})$ in
Section~\ref{fromhorfrominf} 
with $(\ref{contip})$.  We choose $\epsilon_3$ so that the second term on the right hand
side of
$(\ref{toukatof})$ is absorbed also by $(\ref{contip})$.
We may choose $\epsilon_2$, $R_2$
so that the second term  of $(\ref{erfrominf})$ is absorbed by $(\ref{contip})$.
Note that the choice of $R_2$, $\epsilon_2$ 
determines $R_7$ in view of the comments after $(\ref{R7def})$.

In the case of Theorem~\ref{nmt2}, since we can safely appeal  to 
Theorem~\ref{btheorem2}, choosing $\varepsilon=1$,
all remaining terms on the right hand
side are controlled by  
$(\ref{dataterm})$
where we now drop the dependence of $B$ on the already-chosen constants,
and the proof of $(\ref{protasn1b})$ is complete.

In the case of Theorem~\ref{nmt}: 
Restricting to sufficiently small $a_0$, we may choose $e_0<\tilde{e}=e'<e$
such that the spacetime term on the right hand side of the first inequality
of Proposition~\ref{giasuper} (applied for $\tilde\tau$ in place of
$\tau'$, and $0$ in place of $\tau$,
with $0\le\tilde\tau\le\tau'$ arbitrary, and with $\tilde{r}$, $\tilde{\delta}$ chosen as above)
can also be absorbed. (The other term is of the form $(\ref{dataterm})$.)
We choose $\varepsilon$ sufficiently small so that
the first term of $(\ref{toukatof})$ and the third term of $(\ref{kiedwari})$
are absorbed by the positive boundary terms just generated. 
Given now arbitrary $B_0>0$, 
by restricting to smaller $a_0$,
there exists an $e_0<\tilde{e}=e''<e$
such that, adding $B_0$ times the inequalities of Proposition~\ref{giasuper},
the spacetime term on the right hand side can still be absorbed by $(\ref{contip})$.
It suffices to choose $B_0$ large enough so that the left hand side
absorbs the remaining future  boundary terms (the 
term of $(\ref{aptomaint})$, the term of $(\ref{erfrominf})$, the term
of $(\ref{eleos...})$, and the first two terms of $(\ref{kiedwari})$).
We have obtained $(\ref{protasn1b})$ for Theorem~\ref{nmt} for $j=0$.

We retrieve the boundedness statement $(\ref{bndts1b})$  as well as the statement
$(\ref{fluxho...})$
by recalling that we still have the boundary terms on the left hand side of the inequality
arising from our application above
of Proposition~\ref{ftrs} (alternatively, by
 revisiting Proposition~\ref{giasuper}).

\section{Acknowledgements}
The authors thank Stefanos Aretakis for comments on the manuscript.
M.~D.~is supported in part by a grant from the European Research Council. 
I.~R.~is supported in part by NSF grant DMS-0702270.

\printindex

\end{document}